\documentclass[11pt]{article}

\usepackage[a4paper,margin=1in]{geometry}
\usepackage{amsmath,amssymb,amsthm}
\usepackage{mathtools}
\usepackage{booktabs}
\usepackage{tabularx}
\usepackage{array}
\usepackage{enumitem}
\usepackage{xcolor}
\usepackage[colorlinks=true,linkcolor=blue!50!black,citecolor=blue!50!black,urlcolor=blue!50!black]{hyperref}

\newcolumntype{Y}{>{\raggedright\arraybackslash}X}
\newcolumntype{R}[1]{>{\raggedright\arraybackslash}p{#1}}

\theoremstyle{plain}
\newtheorem{theorem}{Theorem}[section]
\newtheorem{proposition}[theorem]{Proposition}
\newtheorem{lemma}[theorem]{Lemma}
\newtheorem{corollary}[theorem]{Corollary}
\theoremstyle{definition}

\theoremstyle{remark}
\newtheorem{remark}[theorem]{Remark}

\newcommand{\Sect}{\mathcal{S}}
\newcommand{\bbe}{\mathbf{e}}
\newcommand{\tn}{\tilde\nabla}
\newcommand{\const}{\mathrm{const}}
\newcommand{\Real}{\mathbb{R}}
\newcommand{\ip}[2]{\langle #1 , #2 \rangle}
\newcommand{\nrm}[1]{\lVert #1 \rVert}
\DeclareMathOperator{\Ker}{Ker}

\title{\bfseries The warp factor of supersymmetric $D=11$ near-horizon geometries:\\
single-point rigidity, the $\mathrm{Spin}(7)$ perfect square,
and global constraints}
\author{Usman Kayani}
\date{}

\begin{document}
\maketitle

\begin{abstract}
\noindent
On a compact connected section $\Sect$ of a supersymmetric $M$-horizon the
Killing-spinor bilinears give a pointwise identity relating the warp factor
$\Delta$, the rotation one-form $h$ and the two spinor norms, with nothing
assumed about the norm of the Killing spinor. We derive three consequences of
that identity in the unreduced form it takes when the usual constancy
assumption is dropped.

\medskip\noindent
First, a single-point rigidity theorem: if $\Delta$ and $h$ vanish at one point
of the section, then the flux vanishes identically, $\Sect$ is Ricci-flat and
the horizon is $\Real^{1,1}\times\Sect$. A condition at a single point replaces
the global spinorial hypothesis usually imposed.

\medskip\noindent
Second, the $\mathrm{Spin}(7)$ decomposition of the flux fixes $h$ algebraically
and exhibits the invariant flux content of the scalar square $\Delta=4\Phi^2$
known in an adapted gauge:
\[
\Delta=\tfrac1{108}\big\lVert w_{\mathbf 7}\pm\sigma(\mathcal{Y}_{\mathbf 7})
\big\rVert^2\,.
\]
Only two of the flux modules, the $\mathbf 7$-summands $w_{\mathbf 7}$ and
$\mathcal{Y}_{\mathbf 7}$, reach the warp factor, and the square is degenerate
rather than definite: it vanishes on the linear subspace
$w_{\mathbf 7}=\mp\sigma(\mathcal{Y}_{\mathbf 7})$ rather than at the origin, so
positivity of $\Delta$ yields no case list.

\medskip\noindent
Third, the two results combine into a pointwise budget in which the single
constant of supersymmetry is shared between the deviation from staticity and the
$\mathrm{Spin}(7)$ mismatch of the flux, each bounded by that constant.

\medskip\noindent
The same identities characterise the constancy hypothesis, which is equivalent
to $V=-fh$ for the bilinear one-form $V$ and fails on the static branch of the
known warped $AdS_2$ solutions; they fix a weighted global integral of the warp
factor; and they reduce the Komar angular momentum of the horizon to a positive
bilinear integral that vanishes exactly on the static branch. Two no-go
statements, for hidden symmetries built from the Killing spinor and for a second
isometry built from the sector bilinears, are stated with their hypotheses in
the text.
\end{abstract}

\tableofcontents

\section{Introduction}
\label{sec:intro}

The near-horizon geometry of an extremal black hole is a solution in its own
right, and supersymmetric near-horizon geometries are far more rigid than
generic ones. The mechanism behind that rigidity is by now well understood in
outline. One writes the metric in Gaussian null coordinates adapted to the
Killing horizon, integrates the Killing spinor equations (KSEs) along the two
lightcone directions, and finds that the resulting system on the spatial
horizon section $\Sect$ is governed by a pair of Dirac-like operators whose
kernels can be counted. The counting shows that the number of preserved
supersymmetries is even and that the isometry algebra contains
$\mathfrak{sl}(2,\Real)$. This programme has been carried through for heterotic,
type~IIB, type~IIA and massive~IIA horizons, for $D=4$ and $D=5$ gauged
supergravities, for $D=6$ $(1,0)$ supergravity, and --- the case of interest
here --- for $M$-horizons in $D=11$
\cite{11index,iibindex,iiaindex,miiaindex,4dindex,5dindex,KayaniThesis}.

The counting is not what this paper is about. We ask instead what those same
identities say about the warp factor $\Delta$ of the near-horizon metric, and we
ask it with nothing assumed about the norm of the Killing spinor. That one
hypothesis, $\nrm{\eta_-}=\const$, would delete a gradient term from the
identity everything below runs on and leave the much stronger-looking
$\Delta+h^2=\const$ in its place. No maximum principle supplies it, it is
equivalent to an on-shell condition that can fail, and it does fail: the warped
$AdS_2$ solutions of \cite{KimPark} are static $M$-horizons with $h^2$
non-constant \cite[fn.~4]{Mhorizons}. So it is imposed nowhere here. The cost is
that each global statement has to be extracted from an identity carrying a
gradient; the return is that the statements hold on every supersymmetric
$M$-horizon with compact section rather than on a subclass, and that several of
them are void, not merely weaker, if the gradient term is discarded.

\medskip\noindent\emph{Two global facts.} These frame what follows, and neither
is new. Integrating the twist-divergence field equation \eqref{eq:divh} over a
compact section gives, with no supersymmetry assumption, the mean-value identity
$2\overline{\Delta}+\overline{h^2}=\tfrac13\overline{Y^2}+\tfrac1{72}\overline{X^2}$
and hence
\begin{equation}
\overline{\Delta}\;\le\;\tfrac16\overline{Y^2}+\tfrac1{144}\overline{X^2}\,,
\qquad\text{equality}\iff h\equiv0
\label{eq:intro-bound}
\end{equation}
(Theorem~\ref{thm:bound}). On the saturation locus $h\equiv0$ the remaining
field equations force $\mathrm{d}Y=0$ and $\Delta$ constant, so the geometry is
the unwarped product $AdS_2\times\Sect$ when $\Delta>0$ and
$\Real^{1,1}\times\Sect$ when $\Delta=0$, with
$\Delta=\tfrac16Y^2+\tfrac1{144}X^2$
(Proposition~\ref{prop:staticbosonic}). The second is the static $M$-horizon
theorem of \cite[\S2.1,\,\S3.3]{staticM}; the first is one integration of a
field equation, used in the same way in \cite[\S5]{11index} and
\cite[\S3.2.1]{KL2013review}. We state both because the rest of the paper is
organised around the saturation locus, and we claim neither: see
Remark~\ref{rem:staticbosonic-attrib}.

\medskip\noindent\emph{The master identity.} The warp--flux identities (3.22),
(3.23) of \cite{Mhorizons} have $\eta_-$ companions which, together with the
sector pairing $\eta_+=\Gamma_+\Theta_-\eta_-$, eliminate the vector bilinear
and give
\begin{equation}
\big(\Delta+h^2\big)\nrm{\eta_-}^2 + h^i\tn_i\nrm{\eta_-}^2
 \;=\; \nrm{\eta_+}^2 \;=\; \const \;\ge\; 0
\label{eq:intro-master}
\end{equation}
at every point of any $M$-horizon with a compact connected section and $N_->0$
--- by the counting $N=2N_-$ of \cite{11index}, every supersymmetric one
(Theorem~\ref{thm:master}). Alongside it stands the divergence identity
$\tn^i(\tn_i\nrm{\eta_-}^2+h_i\nrm{\eta_-}^2)=0$, which is what the $\eta_-$
Lichnerowicz identity becomes on a parallel spinor. Relation
\eqref{eq:intro-master} is not itself new: it is the $D=11$ member of a family
written down to close the $\mathfrak{sl}(2,\Real)$ algebra, and
Section~\ref{sec:lit} derives it from (6.9) of \cite{11index} in three lines.
Everything listed next is a consequence of that pair, used in the unreduced form
above.

\medskip
\hrule
\medskip

\noindent\emph{What the paper adds.} The identity \eqref{eq:intro-master} is
known; what is new is what follows from it when no norm assumption is made.
Three results carry the paper, and Section~\ref{sec:status} states each one
again with its hypotheses and its provenance.

\smallskip\noindent\textbf{A. Single-point rigidity} (Theorem~\ref{thm:rigid}).
If $\Delta$ and $h$ vanish at one point of a compact connected section, then
$X=Y=0$ and $h\equiv0$ identically, $\Sect$ is Ricci-flat and the geometry is
$\Real^{1,1}\times\Sect$.

\smallskip\noindent\textbf{B. The warp factor is a degenerate perfect square}
(Theorem~\ref{thm:spin7square}). In the $\mathrm{Spin}(7)$ frame the Killing
spinor defines,
$\Delta=\tfrac1{108}\lVert w_{\mathbf 7}\pm\sigma(\mathcal{Y}_{\mathbf 7})\rVert^2$:
the two $\mathbf 7$-summands $w_{\mathbf 7}$ and $\mathcal{Y}_{\mathbf 7}$ and
no more. That $\Delta$ is a square is known
--- it is $4\Phi^2$ in the adapted gauge of \cite{Mhorizons} --- and what is
new is the invariant flux content of that square: which modules enter, the
covariant formula, and the zero locus, which is a linear subspace rather than
the origin, so that positivity of $\Delta$ yields no case list
(Remark~\ref{rem:phisquare}).

\smallskip\noindent\textbf{C. The supersymmetric budget}
(Corollary~\ref{cor:budget}).
$\lvert V\rvert^2/f+\tfrac{f}{108}
\lVert w_{\mathbf 7}\pm\sigma(\mathcal{Y}_{\mathbf 7})\rVert^2=c$ at every
point. In derivation this is B substituted into a balance law, and it is
labelled a corollary throughout; in content it is the statement we would put
forward, being the one place where the pointwise algebra of the flux and the
global constant of supersymmetry meet.

\medskip
\hrule
\medskip

\noindent A and B are proved from the identities and from the representation theory of
$\mathrm{Spin}(7)$; the machine verification of Appendix~\ref{app:verify}
confirms their coefficients and exhibits witnesses, but is not what establishes
them. The negative results are the opposite way round: in
Proposition~\ref{prop:noclosure}, Theorem~\ref{thm:hidden} and
Proposition~\ref{prop:twosector} the exact rank computations \emph{are} the
proof. Everything else below is either quoted input or a consequence of A--C,
and the list (a)--(f) says which is which. The argument is one chain,

\begin{center}
known master identity $\;\longrightarrow\;$ weighted global identity
 $\;\longrightarrow\;$ single-point rigidity
 $\;\longrightarrow\;$ $\mathrm{Spin}(7)$ twist--flux relation and its obstruction
 $\;\longrightarrow\;$ global charges, balance identities and curvature
 $\;\longrightarrow\;$ two no-go theorems delimiting what is left,
\end{center}

\noindent in which the first link is quoted from \cite{11index} and the rest are
the contribution. The links do not carry equal weight, and it is better to say
so at the outset. The mathematical core is A and B --- items (b) and (c) below
--- together with the exact obstruction to closure that follows from them, and
C, item (d), is where they meet. The identities that carry the core to global
statements --- the weighted identity (a) and the reading
$8\pi G^{(11)}\lvert J\rvert=\int_\Sect\lvert V\rvert^2/f$ of the Komar
charge --- are consequences, and are stated as such: each is an elimination
between identities already listed, not a further input. The remaining material
--- the entropy representation, the total scalar curvature, the conformal
factor of positive scalar curvature, and the two no-go statements of (f) ---
is a set of further consequences and obstructions, included because a
classification attempt has to know which routes are closed, and each is
labelled with the hypotheses it needs in Section~\ref{sec:status}. We have not found the later links stated previously in
\cite{Mhorizons,11index,GGPreview,superalgebras,KL2013review,N4d11}; that is a
statement about our search, not a priority claim, and the comparison is made
equation by equation in Section~\ref{sec:status}, which is the one place in this
paper where attribution is settled.

\emph{(a) A weighted global identity.} Integrating the master identity
\eqref{eq:intro-master} against the field equations gives
\eqref{eq:weightedbound},
$\int_\Sect\Delta\nrm{\eta_-}^2
=\int_\Sect(\tfrac13Y^2+\tfrac1{72}X^2)\nrm{\eta_-}^2-c\,\mathrm{Vol}(\Sect)$
with $c=\nrm{\eta_+}^2$, the supersymmetric refinement of
\eqref{eq:intro-bound}. Its $c=0$ case is the integral argument of
\cite[\S5]{11index}, which does carry the same weight but runs only after
$\Theta_-\eta_-=0$ has been imposed; what is added is the $c$-term and the
validity on any supersymmetric $M$-horizon with no condition on
$\Theta_-\eta_-$. Since $\Delta\ge0$ it has an inequality form,
$\int_\Sect Qf\ge c\,\mathrm{Vol}(\Sect)$ for the flux invariant
$Q=\tfrac13Y^2+\tfrac1{72}X^2$, and the equality case is exactly the locus
$\Delta\equiv0$, which (c) characterises algebraically
(Corollary~\ref{cor:fluxlower}).

\emph{(b) Single-point rigidity.} The saturation locus of
\eqref{eq:intro-bound} is rigid from below: if $\Delta$ and $h$ vanish at a
\emph{single} point $p$ of a compact connected section, then $X=Y=0$ and
$h\equiv0$ identically, $\Sect$ is Ricci-flat and the geometry is
$\Real^{1,1}\times\Sect$ (Theorem~\ref{thm:rigid}). The chain is
\begin{equation*}
\Delta(p)=h(p)=0
\;\Longrightarrow\;\nrm{\eta_+}^2=0
\;\Longrightarrow\;\Theta_-\eta_-=0 \text{ everywhere}
\;\Longrightarrow\;X=Y=h=\Delta=0 ,
\end{equation*}
where the first implication is the master identity evaluated at $p$ together
with the constancy of $\nrm{\eta_+}$, and the last is the kernel theorem of
\cite[\S5]{11index}. The middle links are the point of the theorem:
\emph{the vanishing of $\Delta$ and $h$ at one geometric point implies the
global spinorial kernel condition $\Theta_-\eta_-=0$.} That bridge is not
available bosonically --- the field equations relate $\Delta$ and the flux only
in the mean, and vanishing at one point says nothing.

\emph{(c) $\mathrm{Spin}(7)$ twist--flux identity, and an obstruction.} A
nowhere-vanishing $\eta_-$ reduces the structure group of $\Sect$ to
$\mathrm{Spin}(7)$, and the restrictions this imposes were left open in
\cite[\S6.2.1]{11index} for a single Killing spinor. We solve the vector
identities algebraically for the twist one-form in terms of the flux and the
$\mathrm{Spin}(7)$ data $(Z,\varphi)$ (Theorem~\ref{thm:htoflux}) and expand
$\Delta$ in the same data \eqref{eq:Deltaexplicit}. Only two $\mathbf 7$-modules
of the flux survive in that expansion (Proposition~\ref{prop:spin7delta}), and
they assemble into a perfect square,
$\Delta=\tfrac1{108}\lVert w_{\mathbf 7}\pm\sigma(\mathcal{Y}_{\mathbf 7})\rVert^2$
(Theorem~\ref{thm:spin7square}), so the locus $\Delta\equiv0$ is cut out by the
seven linear conditions $w_{\mathbf 7}=\mp\sigma(\mathcal{Y}_{\mathbf 7})$ rather
than by any positivity argument. The same expansion supplies
an obstruction with a precise cause (Proposition~\ref{prop:noclosure}):
$\Delta$ does not factor through the $\varphi$-free scalar data
$(I_1,\dots,I_9)$ --- in particular it is not determined by $(h^2,X^2,Y^2)$ on
the allowed algebraic data --- because an independent $\mathrm{Spin}(7)$
contraction of the four-form survives. That is where the scalar closure step
fails: the mechanism through which the identical argument produces a case list
in minimal $D=5$ supergravity \cite[App.~D]{KayaniThesis} does not operate on
the pointwise $\mathrm{Spin}(7)$-algebraic data in $D=11$. It is not a claim
that no classification of $D=11$ near-horizon geometries is possible, and the
scope is fixed in Remark~\ref{rem:noclosurescope}: field equations, Bianchi
identities, topology and extra supersymmetry all impose relations that a
pointwise rank test does not see.

\emph{(d) Global charges and curvature.} Because the master identity carries
its gradient term, the same pair integrates against the Killing vector rather
than merely over $\Sect$, and the warp factor can then be removed from the
result. This gives a Komar charge for the horizon rotation as a weighted flux
integral against the constant $c$, with no $\Delta$ left in it
(Theorem~\ref{thm:komar}), a bound on that rotation
(Corollary~\ref{cor:rotationbound}), and, when $c>0$, a weighted
flux-and-warp representation of the volume of $\Sect$ and hence of the
Bekenstein--Hawking entropy (Corollary~\ref{cor:entropy}); neither is an
expression in the flux alone, since the warp factor $f$ and the constant $c$
remain in the integrands, and what has been removed is the geometric
$\Delta$. It also gives a weighted form of the same
identity, one integral law for each $p\in\Real$
(Theorem~\ref{thm:balance}); these are moments of a single pointwise statement
rather than an independent stock of conservation laws, and the member that does
the work below is $p=1$. The pointwise statement itself follows once the
invariance of $\nrm{\eta_-}$ along the rotation is used
(Theorem~\ref{thm:pwbalance}); it is the \emph{supersymmetric balance law}
\begin{equation}
\Delta f+\frac{\lvert V\rvert^2}{f}=c
\label{eq:intro-pwbalance}
\end{equation}
in which the one constant of the problem is split pointwise into a warp
contribution and a non-staticity contribution. The content of the weighted form is
not that there are many identities but that it survives the hypothesis
$\mathcal{L}_Vf=0$ being dropped. Substituting the perfect square of (c) for
$\Delta$ in \eqref{eq:intro-pwbalance} eliminates the geometry from the budget
altogether and leaves
\begin{equation}
\frac{\lvert V\rvert^2}{f}+\frac{f}{108}
 \big\lVert\,w_{\mathbf 7}\pm\sigma(\mathcal{Y}_{\mathbf 7})\,\big\rVert^2=c
\label{eq:intro-budget}
\end{equation}
(Corollary~\ref{cor:budget}): the deviation from staticity and the
$\mathrm{Spin}(7)$ mismatch of the flux share one pointwise budget, and each is
therefore bounded by $c$. The two saturations are the two extreme branches ---
$\lvert V\rvert^2/f=c$ exactly on the linear locus
$w_{\mathbf 7}=\mp\sigma(\mathcal{Y}_{\mathbf 7})$, where the flux need not
vanish, and the square exhausts the budget exactly on the static branch. Alongside them stands a purely
bosonic identity for the total scalar curvature of $\Sect$
(Proposition~\ref{prop:totalR}), and, on the branch where the constancy of
$\nrm{\eta_-}$ fails, again for $c>0$, a conformal factor $\nrm{\eta_-}^{2/7}$
of strictly positive scalar curvature built from the Killing spinor itself
(Theorem~\ref{thm:psc}). The first and the third are void if $\nrm{\eta_-}$ is
assumed constant --- the second is purely bosonic and independent of it ---
which is what makes (e) more than a technical remark. The balance family is the
one member of this group that survives compactness being dropped: its pointwise
form (Theorem~\ref{thm:pwbalance}) bounds the deviation from staticity by
$\lvert V\rvert^2/f\le c$ at every point (Corollary~\ref{cor:balance}), with
equality exactly where $\Delta$ vanishes. The charges close among themselves as
well, and the statement we would single out is the one the balance law
supplies: the horizon-local Komar charge is minus the integrated
non-staticity,
\begin{equation}
8\pi G^{(11)}\lvert J\rvert=\int_\Sect\frac{\lvert V\rvert^2}{f},
\label{eq:intro-komarbalance}
\end{equation}
so that $J=0$ exactly on the static branch, and with the orientation fixed in
Section~\ref{sec:charges} the signed form reads $8\pi
G^{(11)}J=-\int_\Sect\lvert V\rvert^2/f\le0$. Eliminating the warp factor
between the Komar charge and the weighted identity instead gives the
Smarr-type relation $8\pi G^{(11)}J=\int_\Sect Qf-8cG^{(11)}S_{\mathrm{BH}}$
(Corollary~\ref{cor:smarr}); this is algebra among identities already listed
rather than a further input, and we record it because the three charges close,
not as a headline.

\emph{(e) The status of the $\eta_-$ norm.} The $\eta_-$ Lichnerowicz identity
is a divergence rather than a Laplacian: it identifies $\Ker\mathcal{D}^{(-)}$
with the $\nabla^{(-)}$-parallel spinors with no hypothesis on the norm, but it
does not imply that norm is constant. We prove that constancy is
\emph{equivalent} to the on-shell relation $K_i+h_iK_+=0$, hence to
$V=-\nrm{\eta_-}^2h$ for the bilinear one-form $V$ of \eqref{eq:gpfamily}, and
we identify the branch in which it fails: $V\equiv0$, that is
$h=-\mathrm{d}\log\nrm{\eta_-}^2$ with $\Delta\nrm{\eta_-}^2$ constant, which is
the warped static horizon of \cite{staticM} and which the solutions of
\cite{KimPark} realise. Everything downstream is therefore stated in the
unreduced form \eqref{eq:intro-master}, and what would need the constancy --- in
particular $\Delta+h^2=\const$ and the pointwise bounds that follow from it ---
is labelled conditional, here and in Section~\ref{sec:status}. No claim is made
that the literature assumed $\nrm{\eta_-}$ constant: the hypothesis
$\ip{\phi}{\phi}=\const$ of \cite[fn.~6]{Mhorizons} belongs to the $+$ sector,
where the Hopf maximum principle discharges it \cite[(4.8)]{11index}, and the
asymmetry is explicit in \cite[\S11.3]{GGPreview}, where the Lichnerowicz
theorem for $\mathcal{D}^{(+)}$ is proved by the Hopf maximum principle, which
delivers $\nrm{\phi_+}=\const$ as a by-product, whereas the one for
$\mathcal{D}^{(-)}$ is proved by partial integration, which delivers no such
statement.

\emph{(f) Two no-go results.} The last two sections establish what is
\emph{not} available in $D=11$, and both statements are pointwise counts in the
$162$-dimensional flux space rather than failures of technique. First, the six
$\mathrm{Spin}(7)$-invariant forms the Killing spinor defines carry a
Killing--Yano or a closed conformal Killing--Yano structure only when the
$\mathrm{Spin}(7)$ structure is parallel and the flux reduces to its
$\mathbf{27}$ --- codimension $135$ --- so the hidden symmetries that make the
Kerr family integrable have no supersymmetric analogue among these forms
(Theorem~\ref{thm:hidden}, Corollary~\ref{cor:hidden}). The codimension is a
count in the pointwise algebraic flux data of the horizon, before any
differential constraint beyond the Killing spinor equation and the field
equations is imposed, and the forms tested are the six the $\mathrm{Spin}(7)$
structure of the Killing spinor generates; a Killing--Yano tensor of some other
origin is not addressed. Second, no Killing
vector of the section is built from the two sector unit vectors: the ansatz
$U=aZ_++bZ_-$ is obstructed by $28$ conditions at a generic point
(Proposition~\ref{prop:twosector}), in contrast with the affirmative $D=5$
answer of \cite{5dsecond}. This is a no-go for that ansatz and for nothing
else --- we do not claim that a supersymmetric $M$-horizon has no second
isometry. Taken with the Bochner integrand of
Section~\ref{sec:bochner}, which has no sign, these say that the pointwise
algebraic data of an eleven-dimensional horizon are less constrained than their
five-dimensional counterpart, and they say it quantitatively.

\medskip\noindent Everything else --- the lightcone reduction, the operators
$\Theta_\pm$ and $\nabla^{(\pm)}_i$ with the algebraic conditions accompanying
them \cite{Mhorizons}, the anti-Hermiticity step in its $D=11$ form
\cite[App.~C]{iiaindex}, the warp--flux identities (3.22), (3.23) of
\cite{Mhorizons}, their $\eta_-$ companions, the sector pairing, the constancy
of $\nrm{\eta_+}$ \cite[(4.8)]{11index}, and the two global facts above --- is
known, is attributed at the point of use, and is restated only to have one set
of conventions and a machine-checked audit trail. What the restatement adds,
beyond confidence, is bookkeeping: each algebraic condition is tracked to the
exact combination of field-equation and Bianchi residuals it is proportional to
off shell (Propositions~\ref{prop:B1}, \ref{prop:B5}, \ref{prop:B1minus}), and
the warp--flux identities are given their explicit flux expansions,
\eqref{eq:pointwiseexplicit} and \eqref{eq:pointwiseminusexplicit}.

\subsection*{Relation to previous work}
\label{sec:lit}

Almost every global statement below is a consequence of an identity that
already exists in the literature, so we set that literature out first. The gap
this paper occupies is the one left by the $D=11$ horizon programme of
\cite{staticM,Mhorizons,11index,GGPreview,KL2013review,N4d11}: what global
rigidity follows from the known spinorial identities without imposing extra
supersymmetry, extra symmetry, or a constant warp factor.

\medskip\noindent\emph{The horizon conjecture.} The programme originates in the
spinorial-geometry method \cite{spinorialgeometry} and in the observation,
first made for heterotic horizons \cite{hethorizons,hetdelpezzo}, that a
supersymmetric extremal horizon with compact section preserves an even number
of supersymmetries and carries an $\mathfrak{sl}(2,\Real)$ isometry algebra
that is derived rather than assumed. It has since been run in $D=11$
\cite{Mhorizons,11index}, type~IIB \cite{iibhorizons,iibindex}, type~IIA and
massive~IIA \cite{iiaindex,miiaindex}, minimal $D=5$ \cite{5dhorizons},
$D=5$ ungauged and $U(1)$-gauged supergravity coupled to abelian vector
multiplets \cite{5dindex}, $\mathcal{N}=2$, $D=4$ gauged \cite{4dindex} and
$D=6$ $(1,0)$ \cite{6dindex},
with several cases carried out in the conventions used here
\cite{KayaniThesis} and the possible superalgebras determined in
\cite{superalgebras}. In each case lightcone integration reduces the KSEs to
$\nabla^{(\pm)}_i\eta_\pm=0$ on the section, Lichnerowicz-type theorems
identify their solutions with the zero modes of two horizon Dirac operators, an
index argument gives $N=2N_-$, and bilinears built from the two sectors close
on $\mathfrak{sl}(2,\Real)$.

\medskip\noindent\emph{The identity family behind the
$\mathfrak{sl}(2,\Real)$.} The mechanism of that last step is what matters
here, because the global results of Section~\ref{sec:warp} are the same
mechanism read in a different order. In $D=11$ it is \cite[\S6]{11index}:
writing $V_i=\ip{\Gamma_+\eta_-}{\Gamma_i\eta_+}$ for the one-form bilinear
pairing the two sectors, the field equations and the KSEs give the identities
(6.9),
\begin{equation}
\begin{aligned}
2\ip{\Gamma_+\eta_-}{\Theta_+\eta_+} &= 2\nrm{\eta_+}^2 + h_iV^i\,, &\qquad
i_V(\mathrm{d}h) + 2\,\mathrm{d}\ip{\Gamma_+\eta_-}{\Theta_+\eta_+} &= 0\,,\\
2\ip{\Gamma_+\eta_-}{\Theta_+\eta_+} &= \Delta\nrm{\eta_-}^2\,, &
V + \nrm{\eta_-}^2\,h + \mathrm{d}\nrm{\eta_-}^2 &= 0\,,
\end{aligned}
\label{eq:gpfamily}
\end{equation}
together with $\ip{\eta_+}{\Gamma_i\Theta_+\eta_+}=0$ and
$\Delta\nrm{\eta_+}^2=4\nrm{\Theta_+\eta_+}^2$, which are (6.7) and (6.6)
there and (3.22), (3.23) of \cite{Mhorizons}. We write these with our symbols;
the normalisation of $\Gamma_+$ in \cite{11index} is not ours, so factors
involving $\nrm{\eta_+}^2$ differ from \eqref{eq:pairnorm} by a constant, and
we state consequences in the normalisation-free form
$\propto\nrm{\Theta_-\eta_-}^2$. The first three of \eqref{eq:gpfamily} make
the bilinears of \cite[(6.4)]{11index} close on $\mathfrak{sl}(2,\Real)$; the
last expresses $V$ in terms of $h$, as for heterotic horizons
\cite{hethorizons,hetdelpezzo}. The same family recurs in each theory listed
above, and in minimal $D=5$ supergravity yields $\tn_i(h^2+2\Delta)=0$ and with
it the case list of \cite[App.~D]{KayaniThesis}. \emph{This is the sense in
which the master identity of Theorem~\ref{thm:master} is not a new
phenomenon}: eliminating $V$ and $\ip{\Gamma_+\eta_-}{\Theta_+\eta_+}$ between
the first, third and fourth of \eqref{eq:gpfamily} gives
\begin{equation}
\big(\Delta+h^2\big)\nrm{\eta_-}^2 + h^i\tn_i\nrm{\eta_-}^2
 = 4\nrm{\Theta_-\eta_-}^2\,,
\label{eq:gpmaster}
\end{equation}
which is \eqref{eq:master}. What we add is that it is used as it stands:
assuming $\nrm{\eta_-}$ constant would drop the gradient term and leave
$\Delta+h^2=\const$, an assumption supplied by no maximum principle
(Section~\ref{sec:dirac}) and not used here; what needs it is collected in
Corollary~\ref{cor:normconst}.

\medskip\noindent\emph{Static horizons, $\mathrm{Spin}(7)$, and the horizon
limit.} The branch $V=0$ of \cite[\S6.2.2]{11index} gives $\mathrm{d}h=0$ with
$h$ exact, hence a warped product $AdS_2\times_w\Sect$; such static
$M$-horizons are analysed in \cite{staticM}, which also treats the strictly
smaller class $h\equiv0$ and shows that the field equations then force
$\mathrm{d}Y=0$ and $\Delta$ constant, so that the product is direct ---
Proposition~\ref{prop:staticbosonic} below, restated rather than claimed. A
nowhere-vanishing $\eta_-$ reduces the structure group of $\Sect$ to
$\mathrm{Spin}(7)$, and the geometric restrictions this imposes are left for
elsewhere in \cite[\S6.2.1]{11index}, the $\mathcal{N}=4$ case being worked out
in \cite{N4d11}; Section~\ref{sec:geom} sits in that gap for a single Killing
spinor, and Remark~\ref{rem:caselist} says why the step that closes the gap in
$D=5$ does not close it here.
The near-horizon limit is not merely formal: in the classification of
supersymmetric degenerate Killing horizons away from that limit
\cite{killinghorizons11}, the class with non-vanishing $\eta_-$ is locally
isometric to a supersymmetric near-horizon geometry, so the rigidity proved
here constrains actual horizons and not only a limiting construction.
Independently of supersymmetry, \cite{FRW,HIW} supply the classical route to
horizon symmetry, \cite{KL2013review} reviews near-horizon geometries and
\cite{nhsymmetrythm1,Reall} are the sources for the $AdS_2$ structure.

\medskip\noindent\emph{The same technology elsewhere.} A KSE reduced on a
compact section, a Lichnerowicz-type theorem, a maximum principle for a
bilinear and a Bochner argument with flux-quadratic weights belong to a larger
programme: the superalgebras of warped $AdS_2$ and horizon geometries
\cite{superalgebras}, whose no-go argument runs on a Bochner identity with
precisely the weight $\tfrac1{72}X^2-\tfrac16Y^2$; a non-existence theorem for
$AdS_3$ backgrounds with $N>16$ \cite{ads3n16} and a uniqueness theorem for
warped Minkowski backgrounds with $N>16$ \cite{warpedmink};
the constancy of the warp factor of heterotic $AdS$ backgrounds \cite{hetads};
$AdS_5$ backgrounds with $24$ supersymmetries \cite{ads5n24}; and the maximum
principle applied to $dS_n$ warp-factor equations \cite{dsn11,ds5enh}. Each
assumes a warped-product ansatz or a supersymmetry count above a threshold. The
statements below assume neither: one $\eta_-$ Killing spinor and a compact
connected section.

\subsection*{What is known, and what is new}
\label{sec:whatsnew}

Much of the machinery below is not ours, and the paper is of little use unless
it is clear which parts those are; the computational side of the same
accounting is in Appendix~\ref{app:verify}.

\medskip\noindent\emph{Quoted.} The identification of $\Ker\mathcal{D}^{(\pm)}$
with the parallel spinors and the counting $N=2N_-$ are \cite{11index}; we use
only $N=2N_-$, and only in Remark~\ref{rem:hypothesis}, to say that the
hypothesis $N_->0$ is automatic. The bilinear identities (6.6)--(6.9) of
\cite{11index} are quoted as \eqref{eq:gpfamily}; everything in
Sections~\ref{sec:warp} and \ref{sec:geom} is a consequence of them, reached
here by an independent route. The $\mathfrak{sl}(2,\Real)$ enhancement of
\cite[\S6]{11index} is quoted in Section~\ref{sec:sl2r}, where the only thing
added is the observation that its input is the list
\eqref{eq:sl2inputa}--\eqref{eq:sl2inputb} and nothing else.
That $\mathrm{Spin}(9)$ acts transitively on $S^{15}$ with stabiliser
$\mathrm{Spin}(7)$, and the $(Z,\varphi)$ data this induces, is classical
\cite{Harvey}.

\medskip\noindent\emph{Restated in one set of conventions, and claimed for
none of it.} The near-horizon field equations \eqref{eq:Rij}--\eqref{eq:divY}
and Proposition~\ref{prop:auxiliary} are \cite{Mhorizons,KayaniThesis}; the
lightcone integration \eqref{eq:lightconesol}, the operators $\Theta_\pm$ and
the algebraic conditions (cc1)--(cc5) are \cite{Mhorizons}; the constancy of
$\nrm{\eta_+}$ \cite{Mhorizons,11index}, the identities (3.22), (3.23) of
\cite{Mhorizons} and their $\eta_-$ companions (5.6), (6.9) of \cite{11index},
the sector pairing $\eta_+=\Gamma_+\Theta_-\eta_-$ \cite{11index,iiaindex}, the
triviality of $\Ker\Theta_-$ on a horizon with flux \cite[\S5]{11index}, and
the Killing property of the Dirac current are all restated in the conventions
of \eqref{eq:kse} and machine-checked. The two global facts above are likewise
not ours, and neither is the master identity \eqref{eq:master}, which is the
$D=11$ member of the family \eqref{eq:gpfamily}.

\medskip\noindent\emph{New here.} What is new about that identity is the use
made of it without assuming $\nrm{\eta_-}$ constant, together with the analysis
of that hypothesis in Section~\ref{sec:dirac} and
Proposition~\ref{prop:conconx}. Items (a)--(f) above are the consequences; each
is listed again, with its hypotheses and its attribution, in
Section~\ref{sec:status}, and that list, not the summary above, is where the
accounting is done. Two items of pure bookkeeping belong here rather than
there: the explicit flux expansions \eqref{eq:pointwiseexplicit} and
\eqref{eq:pointwiseminusexplicit}, and the off-shell residuals of
\eqref{eq:B1}, \eqref{eq:B5} and \eqref{eq:B1minus}, which are new as
bookkeeping rather than as content.

\medskip\noindent\emph{The same accounting in one table.} The entries below
name, for each principal statement, the closest thing we located in the sources
audited and what separates the two. It is a summary of
Section~\ref{sec:status}, which is where the hypotheses are recorded, and it
uses \emph{not found} in the sense fixed there.

\begin{center}
\small
\begin{tabularx}{\textwidth}{@{}>{\raggedright\arraybackslash}p{0.245\textwidth}>{\raggedright\arraybackslash}p{0.32\textwidth}>{\raggedright\arraybackslash}X@{}}
\toprule
\textbf{Statement here} & \textbf{Closest existing statement} & \textbf{What separates them}\\
\midrule
Master identity \eqref{eq:master} &
(6.9) of \cite{11index}; (3.22)--(3.23) of \cite{Mhorizons} &
Nothing: the same identity, quoted. What is added is that it is used unreduced\\[2pt]
Weighted identity \eqref{eq:weightedbound} &
The $c=0$ integral argument of \cite[\S5]{11index} &
The $c$-term, and validity with no condition imposed on $\Theta_-\eta_-$\\[2pt]
Single-point rigidity, Thm.~\ref{thm:rigid} &
The rigidity of \cite[\S5]{11index} from the global hypothesis $\Theta_-\eta_-\equiv0$ &
A pointwise geometric hypothesis at one point replaces the global spinorial one; the tail from $c=0$ onwards is theirs and is cited as such\\[2pt]
Constancy of $\nrm{\eta_-}$ settled, Prop.~\ref{prop:staticwarped} &
The characterisation of \cite{KimPark} in \cite[fn.~4]{Mhorizons}; \cite{staticM} &
The hypothesis is decided rather than avoided: it fails on a branch that is populated\\[2pt]
Twist--flux formula \eqref{eq:htoflux} &
(3.22)/(3.27) of \cite{Mhorizons} $=$ (6.7) of \cite{11index}, in unevaluated bilinear form; \texttt{finaleh} of \cite{N4d11}, differential and for $\mathcal{N}=2$ &
The algebraic $\mathrm{Spin}(7)$ evaluation, for a single Killing spinor, running from the flux to $h$ rather than the other way\\[2pt]
Perfect square, Thm.~\ref{thm:spin7square} &
$\Delta=4\Phi^2$ in the $SU(4)$-adapted gauge, (3.32) of \cite{Mhorizons}: the scalar square, not its module content &
The $\mathrm{Spin}(7)$-irreducible form --- only two $\mathbf 7$'s enter, a covariant norm formula, and a linear zero locus. \emph{New so far as we know}\\[2pt]
Non-closure, Prop.~\ref{prop:noclosure} &
Not found &
\emph{New so far as we know}; a structural obstruction on the pointwise algebraic data, with the scope fixed in Remark~\ref{rem:noclosurescope}\\[2pt]
Intrinsic torsion, Prop.~\ref{prop:torsion} &
The Fern\'andez classification of $\mathrm{Spin}(7)$ structures \cite{Fernandez} &
Classification standard; the flux content of each module, and the surjectivity of the flux-to-torsion map, not found\\[2pt]
Komar charge, Thm.~\ref{thm:komar}, Cor.~\ref{cor:smarr} &
The Komar construction is classical; the input identities are (6.8)--(6.9) of \cite{11index} &
A new supersymmetric $D=11$ reduction of the horizon Komar charge to a weighted flux integral --- an elimination, not a new Komar formalism\\[2pt]
Two no-go statements, Thm.~\ref{thm:hidden} and Prop.~\ref{prop:twosector} &
Not found &
\emph{New so far as we know}, and negative: each exhibits the obstruction rather than reporting a failed attempt\\
\bottomrule
\end{tabularx}
\end{center}

\medskip
The presentation follows the conventions of \cite{KayaniThesis}, so the
formulae here can be compared line by line with the type~IIA, massive~IIA and
$D=6$ analyses carried out there. Section~\ref{sec:conventions-nh} fixes
notation and collects the known near-horizon system;
Section~\ref{sec:dirac} the horizon Dirac operators and the two Lichnerowicz
identities. The new material begins at Section~\ref{sec:warp}: warp--flux
identities, the master identity and single-point rigidity. It continues in
Section~\ref{sec:geom}, the $\mathrm{Spin}(7)$ reduction and the obstruction,
and in Section~\ref{sec:consequences}, which draws the consequences in three
groups --- local, global, and what does not follow --- and records the open
branches.
Appendix~\ref{app:rep} records the explicit real representation,
Appendix~\ref{app:gnc} the Gaussian null frame reduction,
Appendix~\ref{app:integrability} the off-shell integrability conditions, and
Appendix~\ref{app:verify} the verification protocol.

\section{Conventions and the near-horizon equations}
\label{sec:conventions-nh}

\subsection{\texorpdfstring{$D=11$}{D=11} supergravity}
\label{sec:conventions}

The bosonic field content of $D=11$ supergravity \cite{CJS} is the metric
$G_{MN}$ and a three-form potential
$A^{(3)}$ with field strength $F=\mathrm{d}A^{(3)}$. The bosonic action is
\begin{equation}
(16\pi G^{(11)}_{\rm N})S =
\int \mathrm{d}^{11}x\,\sqrt{-G}\Big(R-\tfrac1{48}F_{M_1M_2M_3M_4}F^{M_1M_2M_3M_4}\Big)
-\tfrac16\int A^{(3)}\wedge F\wedge F\,,
\end{equation}
with field equations
\begin{align}
R_{MN} &= \tfrac1{12}F_{ML_1L_2L_3}F_N{}^{L_1L_2L_3}
 -\tfrac1{144}G_{MN}F_{L_1L_2L_3L_4}F^{L_1L_2L_3L_4}\,,
\label{eq:einstein}\\
\mathrm{d}\star F &= \tfrac12 F\wedge F\,,
\label{eq:formeq}
\end{align}
together with the Bianchi identity $\mathrm{d}F=0$. The gravitino variation is
\begin{equation}
\delta\psi_M = \nabla_M\epsilon
+\Big(-\tfrac1{288}\Gamma_M{}^{L_1L_2L_3L_4}F_{L_1L_2L_3L_4}
+\tfrac1{36}F_{ML_1L_2L_3}\Gamma^{L_1L_2L_3}\Big)\epsilon\,,
\label{eq:kse}
\end{equation}
and a bosonic background is supersymmetric if $\delta\psi_M=0$ admits a
non-trivial solution $\epsilon$, a section of the $\mathrm{Spin}(10,1)$ bundle
in the Majorana representation of real rank $32$.

\subsection{Near-horizon geometry}
\label{sec:nh}

\subsubsection{Gaussian null coordinates}

Near a smooth extremal Killing horizon the metric can be written in Gaussian
null coordinates $(u,r,y^I)$ as
\begin{equation}
\mathrm{d}s^2 = 2\,\bbe^+\bbe^- + \delta_{ij}\,\bbe^i\bbe^j\,,
\qquad
\bbe^+=\mathrm{d}u\,,\quad
\bbe^-=\mathrm{d}r+r\,h-\tfrac12 r^2\Delta\,\mathrm{d}u\,,\quad
\bbe^i=e^i{}_I\,\mathrm{d}y^I\,,
\label{eq:gnc}
\end{equation}
where the spatial horizon section $\Sect=\{r=u=0\}$ is a nine-dimensional
manifold, assumed throughout to be smooth, compact, connected and without
boundary. The warp function $\Delta$, the twist one-form $h=h_i\bbe^i$ --- the
rotation data of the horizon, carried by the $r\,h$ term of $\bbe^-$: it is what
makes the Komar charge of Theorem~\ref{thm:komar} non-zero, and it is exact on
the static branch of Proposition~\ref{prop:staticwarped} --- and the
frame $e^i{}_I$ depend only on the coordinates of $\Sect$; the whole
$r$- and $u$-dependence of the near-horizon data is the one displayed in
\eqref{eq:gnc}. We write $\tn$ for the Levi-Civita connection of the induced
metric on $\Sect$, and $\tilde R_{ij}$, $\tilde R$ for its curvature. Indices
$i,j,\dots$ run over the nine directions tangent to $\Sect$.

The frame \eqref{eq:gnc} has dual vector fields
\begin{equation}
E_+=\partial_u+\tfrac12 r^2\Delta\,\partial_r\,,\qquad
E_-=\partial_r\,,\qquad
E_i=\tilde E_i-r\,h_i\,\partial_r\,,
\label{eq:dualframe}
\end{equation}
with $\tilde E_i$ the frame of $\Sect$ lifted at constant $r$, and spin
connection
\begin{align}
\omega_{+-} &= -r\Delta\,\bbe^+ + \tfrac12 h_i\,\bbe^i\,, &
\omega_{-i} &= -\tfrac12 h_i\,\bbe^+\,,
\nonumber\\
\omega_{+i} &= \tfrac12 r^2(\Delta h_i-\partial_i\Delta)\,\bbe^+
              -\tfrac12 h_i\,\bbe^- + \tfrac12 r\,(\mathrm{d}h)_{ij}\,\bbe^j\,, &
\omega_{ij} &= \tilde\omega_{ij} - \tfrac12 r\,(\mathrm{d}h)_{ij}\,\bbe^+\,,
\label{eq:spinconn}
\end{align}
$\tilde\omega$ being the connection of $\Sect$. Appendix~\ref{app:gnc} sets out
the reduction that produces it, and the curvature that follows from it.

\subsubsection{Decomposition of the flux}

The four-form is decomposed compatibly with \eqref{eq:gnc} as
\begin{equation}
F = \bbe^+\wedge\bbe^-\wedge Y + r\,\bbe^+\wedge \beta + X\,,
\label{eq:fluxdecomp}
\end{equation}
so that $Y_{ij}=F_{+-ij}$ is a two-form and $X_{ijkl}=F_{ijkl}$ a four-form on
$\Sect$. The Bianchi identity $\mathrm{d}F=0$ has two independent components:
the mixed one fixes
\begin{equation}
\beta_{ijk} = (\mathrm{d}Y - h\wedge Y)_{ijk}
 = \tn_iY_{jk}-\tn_jY_{ik}+\tn_kY_{ij}-h_iY_{jk}+h_jY_{ik}-h_kY_{ij}\,,
\label{eq:beta}
\end{equation}
while the purely spatial one is $\tn_{[i}X_{jklm]}=0$, whose residual is written
$BX_{ijklm}$ throughout. The remaining components carry no further
information: the $+ijkl$ one reduces to $\mathrm{d}_h^2Y=-(\mathrm{d}h)\wedge Y$
once \eqref{eq:beta} is used, with $\mathrm{d}_h=\mathrm{d}-h\wedge$, and the
$-ijkl$ one vanishes identically.

\subsubsection{Field equations on the horizon section}

Substituting \eqref{eq:gnc}--\eqref{eq:fluxdecomp} into
\eqref{eq:einstein}--\eqref{eq:formeq} and separating powers of $r$ gives the
system of \cite{Mhorizons,KayaniThesis}. Every component carrying a $-$ and no
$+$ vanishes identically, on the geometric and on the flux side alike, and the
$r$-independent components are
\begin{align}
\tilde R_{ij} &= -\tn_{(i}h_{j)} + \tfrac12 h_ih_j
  -\tfrac12 Y_i{}^kY_{jk} + \tfrac1{12}\delta_{ij}Y^2
  + \tfrac1{12}X_i{}^{klm}X_{jklm} - \tfrac1{144}\delta_{ij}X^2 \,,
\label{eq:Rij}\\
\tn^ih_i &= 2\Delta + h^2 - \tfrac13 Y^2 - \tfrac1{72}X^2\,,
\label{eq:divh}\\
\tn^lX_{ijkl} &= \beta_{ijk} + h^lX_{ijkl}
  + \tfrac1{48}\epsilon_{ijk}{}^{lmnpqs}Y_{lm}X_{npqs}\,,
\label{eq:divX}\\
\tn^jY_{ij} &= -\tfrac1{1152}\epsilon_i{}^{jklmnpqs}X_{jklm}X_{npqs}\,,
\label{eq:divY}
\end{align}
where $h^2=h_ih^i$, $Y^2=Y_{ij}Y^{ij}$, $X^2=X_{ijkl}X^{ijkl}$ and $\epsilon$
is the volume form of $\Sect$, the eleven-dimensional orientation being fixed by
$\epsilon_{-+i_1\cdots i_9}=\epsilon_{i_1\cdots i_9}$. Reversing that
orientation reverses the sign of both Chern--Simons terms at once, so the
relative sign of \eqref{eq:divX} and \eqref{eq:divY} carries no convention
dependence and either one of them fixes the choice.

The remaining components are
\begin{align}
0 &= \tfrac12\tn^j(\mathrm{d}h)_{ij}-(\mathrm{d}h)_{ij}h^j
     -\tn_i\Delta+\Delta h_i
     +\tfrac14 Y^{jk}\beta_{ijk}-\tfrac1{12}X_{ijkl}\beta^{jkl}\,,
\label{eq:auxEpi}\\
0 &= \tfrac12\tn^2\Delta-\tfrac32 h^i\tn_i\Delta-\tfrac12\Delta\,\tn^ih_i
     +\Delta h^2+\tfrac14(\mathrm{d}h)^2-\tfrac1{12}\beta^2\,,
\label{eq:auxEpp}\\
0 &= h^k\beta_{ijk}-\tn^k\beta_{ijk}
     -\tfrac12 X_{ijkl}(\mathrm{d}h)^{kl}
     -\tfrac1{144}\epsilon_{ij}{}^{klmnpqs}\beta_{klm}X_{npqs}\,,
\label{eq:auxform}
\end{align}
with $(\mathrm{d}h)^2=(\mathrm{d}h)_{ij}(\mathrm{d}h)^{ij}$ and
$\beta^2=\beta_{ijk}\beta^{ijk}$. Equation \eqref{eq:auxEpi} is the $+i$
Einstein equation, \eqref{eq:auxEpp} the $++$ one and \eqref{eq:auxform} the
$+ij$ component of the four-form equation. Nothing below assumes them, and
nothing need be: they are consequences of the $r$-independent set.

\begin{proposition}
\label{prop:auxiliary}
On a configuration satisfying the Bianchi identity $\mathrm{d}F=0$ together
with \eqref{eq:Rij}--\eqref{eq:divY}, the three remaining equations
\eqref{eq:auxEpi}--\eqref{eq:auxform} hold automatically. Equivalently,
\eqref{eq:Rij}--\eqref{eq:divY} are the independent field equations of the
near-horizon geometry.
\end{proposition}

\begin{proof}
Both halves are conservation arguments, and both are frame evaluations that we
have machine-checked (Appendix~\ref{app:verify}). On the flux side, the
residual
\begin{equation*}
\Phi_{NPQ}=\tn^MF_{MNPQ}
 +\tfrac1{1152}\epsilon_{NPQ}{}^{M_1\cdots M_8}F_{M_1\cdots M_4}F_{M_5\cdots M_8}
\end{equation*}
obeys $\tn^N\Phi_{NPQ}=0$ identically whenever $\mathrm{d}F=0$, since
$\star\Phi=\pm(\mathrm{d}\star F+\tfrac12F\wedge F)$. Its components are
$\Phi_{+-i}=P_i$ and $\Phi_{ijk}=W_{ijk}$ at order $r^0$ and $\Phi_{+ij}=rQ_{ij}$
at order $r$, that is \eqref{eq:divY}, \eqref{eq:divX} and \eqref{eq:auxform};
the free indices $(ij)$ give $Q_{ij}=W_{ijk}h^k-\tn^kW_{ijk}$, so $W=0$ forces
$Q=0$. On the Einstein side, $\tn^M\hat E_{MN}=0$ for the trace-reversed
residual on configurations with $\mathrm{d}F=0$ and $\Phi=0$; imposing
\eqref{eq:Rij} and \eqref{eq:divh} leaves $\hat E_{+i}=rG_i$ and
$\hat E_{++}=r^2K$, and then $N=i$ gives $G_i=0$, $N=+$ gives
$2K-2G_ih^i+\tn^iG_i=0$ and hence $K=0$, while $N=-$ is identically
satisfied.
\end{proof}

In the computations below the residuals of
\eqref{eq:Rij}--\eqref{eq:divY} are carried as independent symbols
$E_{ij}$, $Fh$, $FX_{ijk}$, $FY_i$, added to the right-hand sides; on shell they
vanish. Tracing \eqref{eq:Rij} gives
$\tilde R = -\tn^ih_i + \tfrac12 h^2 + \tfrac14 Y^2 + \tfrac1{48}X^2$.

\subsection{The Killing spinor equations on the horizon section}
\label{sec:kse}

\subsubsection{Integration along the lightcone}
\label{sec:lightcone}

Write $\epsilon=\epsilon_++\epsilon_-$ with $\Gamma_\pm\epsilon_\pm=0$. Using
\eqref{eq:spinconn} and \eqref{eq:fluxdecomp}, the $-$ and $+$ components of
\eqref{eq:kse} can be integrated in $r$ and $u$, and give
\begin{equation}
\epsilon_+ = \phi_+(u,y)\,,\qquad
\epsilon_- = \phi_-(y) + r\,\Gamma_-\Theta_+\phi_+\,,\qquad
\phi_+ = \eta_+ + u\,\Gamma_+\Theta_-\eta_-\,,\qquad \phi_-=\eta_-\,,
\label{eq:lightconesol}
\end{equation}
where $\eta_\pm$ depend only on the coordinates of $\Sect$ and
\begin{equation}
\Theta_\pm \;=\;
\tfrac14 h_i\Gamma^i \;+\; \tfrac1{288}X_{ijkl}\Gamma^{ijkl}
\;\pm\;\tfrac1{12}Y_{ij}\Gamma^{ij}\,.
\label{eq:Theta}
\end{equation}
in agreement with \cite{Mhorizons}. Only the $Y$-term reverses between the two
sectors, being the flux component that carries both lightcone indices, so
$\Theta_+=\Theta_-$ if and only if $Y=0$. As sections of the
$\mathrm{Spin}(9)$ bundle over $\Sect$, $\eta_+$ and $\eta_-$ are isomorphic,
each of real rank $16$.

\subsubsection{The independent conditions}

Substituting \eqref{eq:lightconesol} back into \eqref{eq:kse} and expanding in
$r$ and $u$ produces the system of \cite{Mhorizons}. Define the flux operators
\begin{equation}
\Psi^{(\pm)}_i = \mp\tfrac14 h_i \pm \tfrac1{24}Y_{kl}\Gamma_i{}^{kl}
 \mp\tfrac16 Y_{ij}\Gamma^j
 -\tfrac1{288}\Gamma_i{}^{jklm}X_{jklm}
 +\tfrac1{36}X_{ijkl}\Gamma^{jkl}\,,
\label{eq:Psi}
\end{equation}
and the corresponding connections on the spinor bundle of $\Sect$,
\begin{equation}
\nabla^{(\pm)}_i = \tn_i + \Psi^{(\pm)}_i\,.
\label{eq:nablapm}
\end{equation}
The relative signs are fixed by the lightcone projection: the $h_i$ and
$Y_{jk}=F_{+-jk}$ terms carry a factor $\Gamma^{+-}$, acting as $\mp1$ on
$\eta_\pm$, whereas the $X_{jklm}$ terms carry no lightcone index and are
common to both sectors. Conjugation with respect to the
$\mathrm{Spin}(9)$-invariant Dirac inner product reverses the rank-three terms
only (Lemma~\ref{lem:ranks}), so
\begin{equation}
\Psi^{(\pm)T}_i = \mp\tfrac14 h_i \mp \tfrac1{24}Y_{kl}\Gamma_i{}^{kl}
 \mp\tfrac16 Y_{ij}\Gamma^j
 -\tfrac1{288}\Gamma_i{}^{jklm}X_{jklm}
 -\tfrac1{36}X_{ijkl}\Gamma^{jkl}\,,
\label{eq:Psitranspose}
\end{equation}
which is not the operator of the other sector: the two operations agree on the
rank-three terms and disagree on the rank-zero and rank-one ones, so
$\Psi^{(-)}_i\neq\Psi^{(+)T}_i$ unless $h$ and $Y$ both vanish. Both objects
occur below in different roles --- $\Psi^{(-)}_i$ is a connection,
$\Psi^{(+)T}_i$ is not.

The independent Killing spinor equations are
\begin{equation}
\nabla^{(\pm)}_i\eta_\pm = 0
\label{eq:indepKSE}
\end{equation}
\cite{Mhorizons}: every other condition generated by the substitution follows
from \eqref{eq:indepKSE} together with the field equations
\eqref{eq:Rij}--\eqref{eq:divY} and the Bianchi identities. Two of the
algebraic conditions are used below, and we record them in the normalisation of
\eqref{eq:kse}:
\begin{equation}
\xi \;=\; \tfrac12\Delta - \tfrac14\tn_ih_j\,\Gamma^{ij}
 + \tfrac1{72}\beta_{ijk}\Gamma^{ijk}
 + 2\Big(\tfrac14h_i\Gamma^i - \tfrac1{288}X_{ijkl}\Gamma^{ijkl}
 + \tfrac1{12}Y_{ij}\Gamma^{ij}\Big)\Theta_+\,,\qquad \xi\eta_+=0\,,
\label{eq:xi}
\end{equation}
\begin{equation}
\xi^- \;=\; -\tfrac12\Delta - \tfrac14\tn_ih_j\Gamma^{ij}
 + \tfrac1{24}\beta_{ijk}\Gamma^{ijk}
 + 2\Big(-\tfrac14h_i\Gamma^i + \tfrac1{288}X_{ijkl}\Gamma^{ijkl}
 + \tfrac1{12}Y_{ij}\Gamma^{ij}\Big)\Theta_-\,,\qquad
\xi^-\eta_-=0\,,
\label{eq:ximinus}
\end{equation}
which are the conditions (cc1) and (cc5) of \cite{Mhorizons}; here
$-\tfrac14\tn_ih_j\Gamma^{ij}=-\tfrac18(\mathrm{d}h)_{ij}\Gamma^{ij}$ and
$\beta=\mathrm{d}_hY$ by \eqref{eq:beta}. In each of them the operator
multiplying $\Theta_\pm$ is $\Theta_\mp$ with the sign of its $h$-term
reversed, which is what conjugation by $\Gamma_\pm$ does to a Clifford element
carried across the lightcone: it reverses the odd-rank terms. The two
conditions are otherwise not images of one another --- the sign of $\Delta$ and
the coefficient of the $\beta$-term differ as well.

The rest of the list --- the order-$r^2$ condition (cc2), the vector condition
$\xi_i\phi_+=0$ at order $r$, and the two $u$-linear images (cc3), (cc4), which
are images of $u$-independent conditions under
$\eta_+\mapsto\Gamma_+\Theta_-\eta_-$ and carry no additional content --- is
collected in Appendix~\ref{app:integrability}. The expansion terminates for a
structural reason: in the Gaussian null frame the metric and the flux are
polynomial in $r$ of bounded degree and independent of $u$, and the spinor
ansatz \eqref{eq:lightconesol} is linear in each of $r$ and $u$, so the
Killing spinor equation can produce conditions only at the orders listed. That
the coefficient at each of those orders is the condition stated, and that the
coefficients above them vanish, is confirmed in the explicit representation,
order by order in $r$ and $u$, on two independent sets of pseudo-random integer
horizon data with $\Delta$, $\tn_i\Delta$, $(\mathrm{d}h)_{ij}$ and the flux
unconstrained (Appendix~\ref{app:verify}); the confirmation is a check on
samples and not itself the proof of termination.

The three contracted integrability conditions used below, $\xi\eta_+=0$,
$\xi_i\eta_+=0$ and $\xi^-\eta_-=0$, are implied on shell by
\eqref{eq:indepKSE}. That is Propositions~\ref{prop:B1}, \ref{prop:B5} and
\ref{prop:B1minus}, placed in Appendix~\ref{app:integrability}: what is used
here are the three implications, whereas the propositions carry more, recording
off shell \emph{which} field-equation or Bianchi residual each term is
proportional to.
\subsection{Clifford algebra on the horizon section}
\label{sec:clifford}

Everything that follows rests on a single elementary observation about
$\mathrm{Cl}(9,0)$, which we isolate because it is used repeatedly and because
it is what allows several identities that fail as operator statements to
nevertheless hold where they are needed.

\begin{lemma}[Anti-Hermitian ranks]
\label{lem:ranks}
Let $\Gamma_i$, $i=1,\dots,9$, be a real symmetric representation of
$\mathrm{Cl}(9,0)$ on $\Real^{16}$, and let
$\Gamma^{(k)}=\Gamma_{i_1\cdots i_k}$ denote the rank-$k$ Clifford basis
elements. Then
\begin{equation}
\big(\Gamma^{(k)}\big)^T = (-1)^{k(k-1)/2}\,\Gamma^{(k)}\,,
\label{eq:transpose}
\end{equation}
so $\Gamma^{(k)}$ is Hermitian for $k=0,1,4,5,8,9$ and anti-Hermitian for
$k=2,3,6,7$. Here and throughout, Hermiticity is with respect to the Dirac
inner product $\ip{\cdot}{\cdot}$ of $\Sect$; since the spinors are Majorana
and the representation real, $M^\dagger=M^T$, and anti-Hermitian means
$M^T=-M$. Consequently, for every real spinor $\eta$,
\begin{equation}
\ip{\eta}{M\eta}=0 \quad\text{for anti-Hermitian }M\,,\qquad\text{in particular}\qquad
\ip{\eta}{\Gamma^{(k)}\eta}=0
\ \ \text{for }k=2,3,6,7\,,
\label{eq:nullranks}
\end{equation}
and $k=2,3,6,7$ are the only ranks for which the bilinear vanishes
identically.
\end{lemma}

\begin{proof}
A real symmetric representation exists (Appendix~\ref{app:rep}). For distinct
indices the antisymmetrisation defining $\Gamma^{(k)}$ is the ordered product
$\Gamma_{i_1}\cdots\Gamma_{i_k}$, so
$(\Gamma^{(k)})^T=\Gamma_{i_k}\cdots\Gamma_{i_1}$, and reversing $k$ pairwise
anticommuting factors costs $k(k-1)/2$ transpositions, giving
\eqref{eq:transpose}. A real quadratic form $\eta^TM\eta$ depends only on
$\tfrac12(M+M^T)$ and so vanishes identically when $M$ is antisymmetric, which
gives \eqref{eq:nullranks}; for the remaining ranks $\Gamma^{(k)}$ is symmetric
and non-zero, so the quadratic form is a non-zero polynomial in the components
of $\eta$. All representations of $\mathrm{Cl}(9,0)$ on $\Real^{16}$ being
equivalent, the statement is representation-independent.
\end{proof}

\begin{remark}[The anti-Hermiticity step of the other dimensions]
\label{rem:antiherm}
Lemma~\ref{lem:ranks} is the $D=11$ form of the step used at the same point in
\cite[App.~C]{iiaindex} and in the corresponding appendices of
\cite{11index,iibindex,miiaindex}, where the extra operator in the
Lichnerowicz-type identity is noted to be anti-Hermitian and its bilinear
discarded. What \eqref{eq:transpose} adds is the explicit list, and both halves
of it are used: the vanishing half in Proposition~\ref{prop:lich}, the
surviving half in Section~\ref{sec:warp}, where
\eqref{eq:pointwiseexplicit} and \eqref{eq:pointwiseminusexplicit} record which
flux contractions are left once the anti-Hermitian part is dropped, and where
the exhaustiveness clause is what makes those leftovers non-vacuous.
\end{remark}

Two consequences are used below. First, transposition acts as $-1$ on ranks
$2,3,6,7$ and $+1$ on the rest, which gives \eqref{eq:Psitranspose} and
likewise
\begin{equation}
\Theta_\mp^{\,T} = \Theta_\pm\,,\qquad
\Big(\tfrac14 h_i\Gamma^i - \tfrac1{288}X_{ijkl}\Gamma^{ijkl}
+ \tfrac1{12}Y_{ij}\Gamma^{ij}\Big)^{\!T}
= \tfrac14 h_i\Gamma^i - \tfrac1{288}X_{ijkl}\Gamma^{ijkl}
- \tfrac1{12}Y_{ij}\Gamma^{ij}\,,
\label{eq:transposes}
\end{equation}
since these operators have their $Y$-terms at rank two and their $h$- and
$X$-terms at ranks one and four. In particular transposition exchanges the two
sector operators rather than fixing them, so the operator multiplying
$\Theta_+$ in \eqref{eq:xi} is neither $\Theta_-$ nor $\Theta_-^{\,T}$; this is
what makes the bilinear of $\xi$ in Section~\ref{sec:warp} an inner product of
two \emph{different} vectors, and hence not a norm except in the vacuum.
Second, by Hodge duality on $\Sect$ the vanishing set $\{2,3,6,7\}$ is stable
under $k\mapsto9-k$, so a term of rank six is on the same footing as one of
rank three.

\section{Horizon Dirac operators and the Lichnerowicz identities}
\label{sec:dirac}

Define the horizon Dirac operators associated with \eqref{eq:nablapm},
\begin{equation}
\mathcal{D}^{(\pm)} = \Gamma^i\nabla^{(\pm)}_i
 = \Gamma^i\tn_i + \Psi^{(\pm)}\,,\qquad
\Psi^{(\pm)} = \Gamma^i\Psi^{(\pm)}_i\,.
\label{eq:diracops}
\end{equation}
No modification of these operators is needed: in theories with a dilatino they
have to be shifted by a multiple of the algebraic Killing spinor equation
before a Lichnerowicz-type theorem can be formulated \cite{iiaindex}, whereas
$D=11$ supergravity has only a gravitino equation. Every Killing spinor is a
zero mode of the corresponding operator, and the content of the statements
below is the converse.

The algebraic input is a Weitzenb\"ock rearrangement: reorganising the bilinear
$\ip{\eta_\pm}{(\,\cdot\,)\tn_i\eta_\pm}$ into a Dirac term plus a total
derivative produces a combination of $\Psi^{(\pm)}$ and $\Psi^{(\pm)T}_i$ that
closes on the Clifford elements linear in the horizon data,
\begin{equation}
\Gamma^i\Psi^{(\pm)} + 2\Psi^{(\pm)i\,T}
 = -2W^{(\pm)i} - \mathcal{F}^{(\pm)}\Gamma^i\,,
\label{eq:weitzenbock}
\end{equation}
with
\begin{equation}
W^{(\pm)i} = \pm\tfrac12 h^i\,,\qquad
\mathcal{F}^{(\pm)} = \mp\tfrac14 h_j\Gamma^j \mp \tfrac1{24}Y_{jk}\Gamma^{jk}
 - \tfrac1{288}X_{jklm}\Gamma^{jklm}\,.
\label{eq:weitzdata}
\end{equation}
Neither is postulated: a general ansatz spanning the Clifford elements linear
in $h_i$, $Y_{ij}$ and $X_{ijkl}$ makes \eqref{eq:weitzenbock} a linear system
with the unique solution \eqref{eq:weitzdata}. It is $W^{(\pm)i}=\pm\tfrac12h^i$
that produces the twist term of \eqref{eq:maxprin} and \eqref{eq:l2b}, with
opposite signs in the two sectors.

\begin{proposition}[Lichnerowicz identities]
\label{prop:lich}
On shell,
\begin{align}
\tfrac14\tilde R - \Gamma^i\tn_i\Psi^{(+)} - \Psi^{(+)i\,T}\Psi^{(+)}_i
 - \mathcal{F}^{(+)}\Psi^{(+)}
\;=\;&
\tfrac14 \Gamma^{ij}\tn_ih_j
-\tfrac5{48}Y^{ij}h_k\Gamma_{ij}{}^{k}
+\tfrac1{36}X^{ijkl}h_l\Gamma_{ijk}
\nonumber\\
&{}-\tfrac5{1152}X^{ijkl}Y^{mn}\Gamma_{ijklmn}
-\tfrac1{96}X^{ijkl}Y_{ij}\Gamma_{kl}\,,
\label{eq:lichplus}\\[4pt]
\tfrac14\tilde R - \Gamma^i\tn_i\Psi^{(-)} - \Psi^{(-)i\,T}\Psi^{(-)}_i
 - \mathcal{F}^{(-)}\Psi^{(-)}
\;=\;&
-\tfrac12\tn^ih_i
-\tfrac14 \Gamma^{ij}\tn_ih_j
\nonumber\\
&{}+\tfrac14\Gamma^{ijk}\tn_iY_{jk}
-\tfrac5{48}Y^{ij}h_k\Gamma_{ij}{}^{k}
+\tfrac1{18}X^{ijkl}h_l\Gamma_{ijk}
\nonumber\\
&{}+\tfrac5{1152}X^{ijkl}Y^{mn}\Gamma_{ijklmn}
+\tfrac1{96}X^{ijkl}Y_{ij}\Gamma_{kl}\,,
\label{eq:lichminus}
\end{align}
With the exception of the scalar $-\tfrac12\tn^ih_i$ in
\eqref{eq:lichminus}, every term on the right-hand sides is anti-Hermitian ---
each has Clifford rank $2$, $3$ or $6$, and by \eqref{eq:transpose} those ranks
are exactly the anti-Hermitian ones --- and so gives no contribution to the
bilinear with $\eta_\pm$:
\begin{multline}
\ip{\eta_\pm}{\big(\tfrac14\tilde R - \Gamma^i\tn_i\Psi^{(\pm)}
 - \Psi^{(\pm)i\,T}\Psi^{(\pm)}_i - \mathcal{F}^{(\pm)}\Psi^{(\pm)}\big)\eta_\pm}\\
 = \begin{cases} 0 & (+)\,,\\[2pt]
 -\tfrac12\big(\tn^ih_i\big)\nrm{\eta_-}^2 & (-)\,,\end{cases}
\label{eq:lichbilinear}
\end{multline}
for every real spinor $\eta_\pm$ on $\Sect$.
\end{proposition}

\begin{remark}[The zeroth-order term is real, and it occurs in one sector only]
\label{rem:divhterm}
The scalar in \eqref{eq:lichminus} is not an artefact of how the combination is
written: both sides of \eqref{eq:lichbilinear} pair a fixed operator with a
fixed spinor, and no rearrangement of \eqref{eq:weitzenbock} removes it. Two
sector-dependent signs are responsible. The rank-zero part of $\Psi^{(\pm)}_i$
is $\mp\tfrac14h_i$, so $-\Gamma^i\tn_i(\Gamma^j\Psi^{(\pm)}_j)$ contributes
$\pm\tfrac14\tn^ih_i$; and the algebraic condition that eliminates it, $\xi$ in
the $+$ sector and $\xi^-$ in the $-$, carries $\Delta$ with opposite signs.
Once $\tn^ih_i$ and $\Delta$ are related by \eqref{eq:divh} the two
contributions cancel in the $+$ sector and add in the $-$. Everything that
distinguishes the sectors in Sections~\ref{sec:dirac} and \ref{sec:warp} traces
back to this. This is verified symbolically, with the
sector-dependent sign of the $\varepsilon\to\Gamma$ dualisation fixed
independently (Appendix~\ref{app:verify}).
\end{remark}

\noindent The off-shell forms carry, in addition,
$\tfrac14 E^i{}_i + \tfrac14 FY^i\Gamma_i \mp \tfrac1{96}BX^{ijklm}\Gamma_{ijklm}
\pm \tfrac1{24}FX^{ijk}\Gamma_{ijk}$; see Appendix~\ref{app:verify}.

\begin{remark}
The two identities are the \emph{same} combination of $\tilde R$,
$\Gamma^i\tn_i\Psi^{(\pm)}$, $\Psi^{(\pm)i\,T}\Psi^{(\pm)}_i$ and
$\mathcal{F}^{(\pm)}\Psi^{(\pm)}$ with the same coefficients, only the sector
data differing, so the multiple of $\tn^ih_i$ in the $-$ sector is forced
rather than chosen. Neither holds as an operator identity --- the right-hand
sides contain $\Gamma^{ij}\tn_ih_j=\tfrac12(\mathrm{d}h)_{ij}\Gamma^{ij}$,
which no choice in \eqref{eq:weitzdata} could cancel --- and what the
Lichnerowicz argument uses is that the obstruction is anti-Hermitian, so the
Dirac current does not see it.
\end{remark}

Given \eqref{eq:lichbilinear}, the standard argument \cite{11index,iiaindex,%
KayaniThesis} applied to a zero mode of $\mathcal{D}^{(\pm)}$ gives
\begin{align}
\tn^i\tn_i\nrm{\eta_+}^2 - h^i\tn_i\nrm{\eta_+}^2
 &= 2\nrm{\nabla^{(+)}\eta_+}^2\,,
\label{eq:maxprin}\\
\tn^i\big(\tn_i\nrm{\eta_-}^2 + h_i\nrm{\eta_-}^2\big)
 &= 2\nrm{\nabla^{(-)}\eta_-}^2\,,
\label{eq:l2b}
\end{align}
the first-order sign being that of $W^{(\pm)i}$ and the zeroth-order term of
\eqref{eq:l2b} the scalar of \eqref{eq:lichbilinear}. The two are of
\emph{different} analytic type. Equation~\eqref{eq:maxprin} has no
undifferentiated $\nrm{\eta_+}^2$ and is of maximum-principle type. For a
parallel spinor its right-hand side vanishes and it reads
$\tn^i\tn_i\nrm{\eta_+}^2=h^i\tn_i\nrm{\eta_+}^2$, which is the same type of
equation as (3.28) of \cite{Mhorizons} but not the same equation: theirs carries
$-2h^i\tilde\nabla_i$ where \eqref{eq:maxprin} carries $-h^i\tn_i$. The two
differ by $h^i\tn_i\nrm{\eta_+}^2$, which by the vector identity
\eqref{eq:vectorid} equals $2h^i\ip{\eta_+}{\Gamma_i\Theta_+\eta_+}$ and is not
identically zero in the horizon data (Appendix~\ref{app:verify}); on a compact
section it vanishes, because the Hopf maximum principle applied to either form
makes
$\nrm{\eta_+}$ constant and hence imposes the vector condition
\eqref{eq:vectorcond}, which is (3.22) of \cite{Mhorizons}. What we take from
\cite{Mhorizons} is therefore the conclusion and the method, not the identity in
this form. Equation~\eqref{eq:l2b} is a
divergence, and as a Laplacian reads
\begin{equation}
\tn^i\tn_i\nrm{\eta_-}^2 + h^i\tn_i\nrm{\eta_-}^2
 + \big(\tn^ih_i\big)\nrm{\eta_-}^2 = 2\nrm{\nabla^{(-)}\eta_-}^2\,,
\label{eq:l2blap}
\end{equation}
with $\tn^ih_i$ given by \eqref{eq:divh} and of no definite sign, so the Hopf
maximum principle does not apply. In this respect $D=11$ behaves like IIA,
where a zeroth-order term likewise forces one to integrate a divergence
\cite{iiaindex}; the difference is that here it is not a matter contribution
but the divergence of the twist.

Both conclusions that the Lichnerowicz argument needs survive. Applying the
Hopf maximum principle to \eqref{eq:maxprin} and integrating \eqref{eq:l2b}
over a compact connected $\Sect$ gives
\begin{equation}
\Ker\mathcal{D}^{(\pm)} = \{\eta_\pm : \nabla^{(\pm)}_i\eta_\pm = 0\}\,,
\qquad \nrm{\eta_+} = \const\ \text{ on }\Sect\,.
\label{eq:kernels}
\end{equation}
For a $\nabla^{(-)}$-parallel spinor \eqref{eq:l2b} becomes
\begin{equation}
\tn^i\big(\tn_i\nrm{\eta_-}^2 + h_i\nrm{\eta_-}^2\big) = 0\,,
\label{eq:divminus}
\end{equation}
which is used throughout Section~\ref{sec:warp}.

\begin{remark}[What \eqref{eq:l2b} does and does not give]
\label{rem:notconst}
Two things should be said about the difference between \eqref{eq:maxprin} and
\eqref{eq:l2b}, because both are easy to get wrong in opposite directions.

First, the zeroth-order term costs nothing for the Lichnerowicz theorem,
because the left-hand side of \eqref{eq:l2b} is a total divergence and
integrates to zero on a compact $\Sect$ whatever $\nrm{\eta_-}$ does, forcing
$\nabla^{(-)}_i\eta_-=0$ pointwise. So $\Ker\mathcal{D}^{(-)}$ is identified
with the parallel spinors unconditionally. This is the pointwise form of what
\cite{11index} states in integrated form: their Lichnerowicz identity carries
the twist term with a sector factor $(1\pm1)$, so it is absent in the $-$
sector and doubled in the $+$ one. The hypothesis
$\ip{\phi}{\phi}=\const$ that \cite[fn.~6]{Mhorizons} imposes belongs to
the $+$ sector, where it is needed exactly because \eqref{eq:maxprin} is not a
divergence, and it was discharged in \cite[(4.8)]{11index} by the maximum
principle; \eqref{eq:maxprin} is that equation. Nothing in either reference
assumes $\nrm{\eta_-}$ constant.

Second, and in the other direction, \eqref{eq:l2b} does \emph{not} give
$\nrm{\eta_-}=\const$, and neither does anything else: constancy is false in
general. It fails on the warped static branch $V\equiv0$, realised by the
solutions of \cite{KimPark} (Remark~\ref{rem:whynotconst}). The statement is
not vacuous --- it is equivalent
to a clean condition on the Dirac current, $K_i+h_iK_+=0$
(Proposition~\ref{prop:conconx}) --- but it is a hypothesis. In
Section~\ref{sec:warp} the results are therefore split: those that hold
outright, and those, collected in Corollary~\ref{cor:normconst}, that hold when
$\nrm{\eta_-}$ is constant.
\end{remark}

\noindent The constancy of $\nrm{\eta_+}$ is what makes the pointwise
identities of the next two sections usable globally, and, through the vector
identities, what turns them into statements with a sign.

\section{Warp--flux identities and single-point rigidity}
\label{sec:warp}

\subsection{The bosonic mean-value bound}

The first statement uses no supersymmetry at all: it is the integral of
\eqref{eq:divh} over a compact section. It is elementary, and we make no claim
of novelty for it --- integrating \eqref{eq:divh} in precisely this way is a
standard move in this literature, used for instance in \cite[\S5]{11index} to
show that $\tn^ih_i=h^2$ forces $h\equiv0$, and its saturation case
\eqref{eq:pointwisebosonic} is written down in \cite{staticM}. We state it
because the rest of the section is organised around it. For a function $f$ on
$\Sect$ write
$\overline{f}=\big(\int_\Sect\!1\big)^{-1}\int_\Sect f$ for its mean value with
respect to the induced volume.

\begin{theorem}[Integrated warp--flux identity and bound]
\label{thm:bound}
On any near-horizon geometry \eqref{eq:gnc}--\eqref{eq:fluxdecomp} of $D=11$
supergravity with $\Sect$ compact and without boundary,
\begin{equation}
2\overline{\Delta} + \overline{h^2}
= \tfrac13\overline{Y^2} + \tfrac1{72}\overline{X^2}\,,
\label{eq:integrated}
\end{equation}
and consequently
\begin{equation}
\overline{\Delta} \;\le\; \tfrac16\overline{Y^2} + \tfrac1{144}\overline{X^2}\,,
\label{eq:boundeq}
\end{equation}
with equality if and only if $h_i\equiv0$ on $\Sect$.
\end{theorem}

\begin{proof}
Integrate \eqref{eq:divh} over $\Sect$. The left-hand side is the integral of a
divergence over a compact manifold without boundary and therefore vanishes,
giving \eqref{eq:integrated}. Since $\overline{h^2}\ge0$,
\eqref{eq:integrated} gives \eqref{eq:boundeq}, and equality holds precisely
when $\overline{h^2}=0$. As $h^2\ge0$ pointwise and $h$ is smooth, this is
equivalent to $h\equiv0$.
\end{proof}

\begin{corollary}[Vacuum rigidity]
\label{cor:vacuum}
If $X_{ijkl}=Y_{ij}=0$ identically and $\Delta\ge0$, then $\Delta=0$ and
$h_i=0$ on $\Sect$, and $\Sect$ is Ricci-flat; the near-horizon geometry is
locally $\Real^{1,1}\times\Sect$.
\end{corollary}

\begin{proof}
With $X=Y=0$, \eqref{eq:integrated} reads
$\int_\Sect(2\Delta+h^2)=0$ with both integrands non-negative, so
$\Delta=0$ and $h=0$. Substituting into \eqref{eq:Rij} gives $\tilde R_{ij}=0$.
\end{proof}

\begin{remark}[The role of the flux, and what saturation gives]
The flux terms in \eqref{eq:integrated} appear on the \emph{opposite} side from
$\overline{\Delta}$ and $\overline{h^2}$, not as an additional non-negative
summand; this is why the argument of Corollary~\ref{cor:vacuum} does not extend
and why non-trivial $D=11$ horizons exist at all. Saturation of
\eqref{eq:boundeq} is equivalent to $h\equiv0$, and the pointwise content of
\eqref{eq:divh} then sharpens to
\begin{equation}
2\Delta = \tfrac13 Y^2 + \tfrac1{72}X^2
\qquad\text{pointwise on }\Sect\,,
\label{eq:pointwisebosonic}
\end{equation}
so $\Delta\ge0$ everywhere, vanishing exactly where the flux does. The metric
is then $\mathrm{d}s^2=2\,\mathrm{d}u\,\mathrm{d}r-r^2\Delta\,\mathrm{d}u^2
+\gamma_{ij}\mathrm{d}y^i\mathrm{d}y^j$, a warped product which is not yet
static: with $V=\partial_u$,
$V\wedge\mathrm{d}V=-r^2\,\mathrm{d}r\wedge\mathrm{d}\Delta\wedge\mathrm{d}u$
vanishes if and only if $\Delta$ is constant on $\Sect$. That it is constant is
the content of the next proposition, which uses
\eqref{eq:auxEpi}--\eqref{eq:auxEpp} and no supersymmetry.
\end{remark}

\begin{proposition}[Saturation forces a direct product; \cite{staticM}]
\label{prop:staticbosonic}
On any near-horizon geometry \eqref{eq:gnc}--\eqref{eq:fluxdecomp} of $D=11$
supergravity with $\Sect$ compact and without boundary, $h\equiv0$ implies
$\mathrm{d}Y=0$ and $\Delta$ constant, with
$\Delta=\tfrac16\overline{Y^2}+\tfrac1{144}\overline{X^2}$. The near-horizon
geometry is then locally $AdS_2\times\Sect$ when $\Delta>0$ and
$\Real^{1,1}\times\Sect$ when $\Delta=0$.
\end{proposition}

\begin{proof}
With $h\equiv0$ every term of \eqref{eq:auxEpp} carrying $h$ or
$\mathrm{d}h$ drops and $\beta=\mathrm{d}Y$, leaving
$\tfrac12\tn^2\Delta=\tfrac1{12}\beta^2$. Integrating over the compact $\Sect$
annihilates the left-hand side, and the right-hand side is a sum of squares, so
$\beta=\mathrm{d}Y=0$. Equation \eqref{eq:auxEpi} then reduces to
$\tn_i\Delta=0$, so $\Delta$ is constant, and equals its mean, which is
$\tfrac16\overline{Y^2}+\tfrac1{144}\overline{X^2}$ by
\eqref{eq:pointwisebosonic}. The product statement is the preceding remark. The
two algebraic reductions are checked in Appendix~\ref{app:verify}.
\end{proof}

\begin{remark}[Attribution]
\label{rem:staticbosonic-attrib}
Proposition~\ref{prop:staticbosonic} is not new: it is proved in exactly this
way, and with no restriction on $Y$, in \cite[\S2.1,\,\S3.3]{staticM}, where
\eqref{eq:pointwisebosonic} is also recorded in the form
$\Delta=\tfrac16Y^2+\tfrac1{144}X^2$ and $\Delta=0$ is observed to force
$X=Y=0$ with $\Sect$ of $\mathrm{Spin}(7)$ holonomy --- stronger than the
Ricci-flatness recovered in Theorem~\ref{thm:rigid} below. The magnetic case
$Y=0$ appears again in \cite[\S5.2.1]{Mhorizons}. We keep the proposition
because the rest of this section is organised around the saturation locus and
because its two algebraic reductions belong to the verification suite; the only
thing we would add is that the argument nowhere uses the Killing spinor
equations, though it is presented there inside a supersymmetric analysis.
\end{remark}

\noindent What is \emph{not} bosonic --- and is the subject of
Theorem~\ref{thm:rigid} below --- is the converse direction, that a horizon which is
flux-free at a single point is flux-free everywhere.

\subsection{The pointwise supersymmetric identity}

We now use supersymmetry. By Proposition~\ref{prop:B1}, a Killing spinor
$\eta_+$ satisfies $\xi\eta_+=0$ with $\xi$ as in \eqref{eq:xi}. Taking the
Dirac bilinear of this relation kills the $\Gamma^{ij}$ and $\Gamma^{ijk}$
terms of \eqref{eq:xi} outright, both being anti-Hermitian
(Lemma~\ref{lem:ranks}), and leaves:

\begin{theorem}[Pointwise warp--flux identity]
\label{thm:pointwise}
Let $\eta_+$ solve $\nabla^{(+)}_i\eta_+=0$. Then at every point of $\Sect$
\begin{equation}
\Delta\,\nrm{\eta_+}^2 \;=\; -4\,\ip{\Big(\tfrac14 h_i\Gamma^i - \tfrac1{288}X_{ijkl}\Gamma^{ijkl}
- \tfrac1{12}Y_{ij}\Gamma^{ij}\Big)\eta_+}{\Theta_+\eta_+}\,,
\label{eq:pointwise}
\end{equation}
with $\Theta_+$ as in \eqref{eq:Theta}, and the norm of $\eta_+$ obeys the
vector identity
\begin{equation}
\tn_i\nrm{\eta_+}^2 \;=\; 2\,\ip{\eta_+}{\Gamma_i\Theta_+\eta_+}\,.
\label{eq:vectorid}
\end{equation}
If $\Sect$ is compact, so that $\nrm{\eta_+}$ is a non-zero constant by
\eqref{eq:kernels}, then \eqref{eq:vectorid} reduces to the algebraic condition
\begin{equation}
\ip{\eta_+}{\Gamma_i\Theta_+\eta_+} \;=\; 0
\label{eq:vectorcond}
\end{equation}
and \eqref{eq:pointwise} becomes a norm,
\begin{equation}
\Delta\,\nrm{\eta_+}^2 \;=\; 4\,\nrm{\Theta_+\eta_+}^2 \;\ge\; 0\,,
\label{eq:pointwisenorm}
\end{equation}
with equality at a point precisely when $\Theta_+\eta_+$ vanishes there.
Explicitly, in terms of the flux invariants, and solved for the warp factor so
as to be directly comparable with its $\eta_-$ counterpart
\eqref{eq:pointwiseminusexplicit} below,
\begin{align}
\Delta\,\nrm{\eta_+}^2 &=
\Big(-\tfrac14h^2 + \tfrac1{864}X^2 + \tfrac1{18}Y^2\Big)\nrm{\eta_+}^2
+ \tfrac1{20736}X^{ijkl}X^{mnpq}\ip{\eta_+}{\Gamma_{ijklmnpq}\eta_+}
\nonumber\\
&\quad
- \tfrac1{288}X^{ijkl}X_{ij}{}^{mn}\ip{\eta_+}{\Gamma_{klmn}\eta_+}
- \tfrac1{54}X^{ijkl}Y_i{}^{m}\ip{\eta_+}{\Gamma_{jklm}\eta_+}
- \tfrac1{36}Y^{ij}Y^{kl}\ip{\eta_+}{\Gamma_{ijkl}\eta_+}\,.
\label{eq:pointwiseexplicit}
\end{align}
\end{theorem}

\begin{proof}
Take $\ip{\eta_+}{\,\cdot\,\eta_+}$ of $\xi\eta_+=0$. The terms
$-\tfrac14\tn_ih_j\Gamma^{ij}$ and $\tfrac1{72}\beta_{ijk}\Gamma^{ijk}$ of
\eqref{eq:xi} are anti-Hermitian and drop out (Lemma~\ref{lem:ranks}), and by
\eqref{eq:transposes} the remaining operator may be moved to the left of the
pairing at the cost of reversing its $Y$-term, which gives
\eqref{eq:pointwise}. For \eqref{eq:vectorid}, $\nabla^{(+)}_i\eta_+=0$ gives
$\tn_i\nrm{\eta_+}^2=-2\ip{\eta_+}{\Psi^{(+)}_i\eta_+}$, and once the two
rank-three terms of \eqref{eq:Psi} are discarded the survivors are exactly
minus the bilinear-visible part of $\Gamma_i\Theta_+$. For
\eqref{eq:pointwisenorm}, transposing the left entry of \eqref{eq:pointwise}
with $\Gamma_i^{\,T}=\Gamma_i$ produces the operator
$\tfrac12h_i\Gamma^i-\Theta_-$, so that
\begin{equation}
\Delta\,\nrm{\eta_+}^2 = 4\,\ip{\eta_+}{\Theta_-\Theta_+\eta_+}
- 2h^i\ip{\eta_+}{\Gamma_i\Theta_+\eta_+}\,,
\label{eq:pointwisesplit}
\end{equation}
with $\ip{\eta_+}{\Theta_-\Theta_+\eta_+}=\nrm{\Theta_+\eta_+}^2$ because
$\Theta_-=\Theta_+^{\,T}$; on a compact $\Sect$ the last term vanishes by
\eqref{eq:vectorcond}. Finally \eqref{eq:pointwiseexplicit} is the Clifford
expansion of the pairing with the anti-Hermitian ranks discarded.
\end{proof}

\begin{remark}[What signs the identity, and what does not]
Before the vector identity is used, \eqref{eq:pointwise} pairs two operators
that differ by the sign of both flux terms and is not manifestly signed: the
invariants enter \eqref{eq:pointwiseexplicit} with mixed signs and the
rank-four and rank-eight bilinears with none, and algebraic data realising
either sign exist (Appendix~\ref{app:verify}). What signs it is
\eqref{eq:vectorcond}, that is, the constancy of $\nrm{\eta_+}$ on a compact
section, which removes the last term of \eqref{eq:pointwisesplit}. The
conclusion $\Delta\ge0$ is (3.23) of \cite{Mhorizons}, of which
\eqref{eq:vectorcond} is (3.22) and \eqref{eq:vectorid} its (3.26); combining
it with Theorem~\ref{thm:bound} confines the mean value to
$0\le\overline{\Delta}\le\tfrac16\overline{Y^2}+\tfrac1{144}\overline{X^2}$.
\end{remark}

\begin{remark}[The explicit form as a combination of the two identities]
\label{rem:combination}
The flux expansion \eqref{eq:pointwiseexplicit} is \eqref{eq:pointwisenorm}
corrected by half the contraction of \eqref{eq:vectorcond} with $h^i$: the
Clifford expansion of $4\ip{\eta_+}{\Theta_-\Theta_+\eta_+}$ carries the two
$h$--flux cross terms
$\tfrac1{144}X^{ijkl}h^m\ip{\eta_+}{\Gamma_{ijklm}\eta_+}$ and
$-\tfrac13Y^{ij}h_j\ip{\eta_+}{\Gamma_i\eta_+}$, and
$\eqref{eq:pointwisenorm}-\tfrac12h^i\times\eqref{eq:vectorcond}$ is the unique
multiple that clears both, leaving $\tfrac18h^2$ as the only $h$-dependence.
This is the elimination that produces the corresponding $D=5$ formula in
\cite[App.~D]{KayaniThesis}; both steps are verified in
Appendix~\ref{app:verify}.
\end{remark}

\begin{corollary}[Vacuum specialisation]
\label{cor:vacspec}
If $X=Y=0$ both operators of \eqref{eq:pointwise} reduce to
$\Theta=\tfrac14h_i\Gamma^i$ and \eqref{eq:pointwise} becomes
$\big(\tfrac12\Delta+\tfrac18h^2\big)\nrm{\eta_+}^2=0$. Since $\nrm{\eta_+}$ is
a non-zero constant, \eqref{eq:vectorcond} already reads
$\tfrac14h_i\nrm{\eta_+}^2=0$ at zero flux, so $h=0$ and then $\Delta=0$, in
agreement with Corollary~\ref{cor:vacuum}; no sign assumption on $\Delta$ is
needed.
\end{corollary}

\subsection{The \texorpdfstring{$\eta_-$}{eta-} companion identities}
\label{sec:minusid}

The $\eta_-$ sector carries the same two identities, with one structural
difference which turns out to be the source of everything that follows: the
$h$-term of $\Psi^{(-)}_i$ has the opposite sign to
that of $\Psi^{(+)}_i$, so the vector identity acquires an inhomogeneous term.

\begin{proposition}[Vector identity of the $\eta_-$ sector]
\label{prop:vectorminus}
Let $\eta_-$ solve $\nabla^{(-)}_i\eta_-=0$. Then
\begin{equation}
\tn_i\nrm{\eta_-}^2 \;=\; -h_i\nrm{\eta_-}^2
 + 2\,\ip{\eta_-}{\Gamma_i\Theta_-\eta_-}\,,
\label{eq:vectoridminus}
\end{equation}
equivalently
\begin{equation}
\ip{\eta_-}{\Gamma_i\Theta_-\eta_-} \;=\;
 \tfrac12\big(\tn_i\nrm{\eta_-}^2 + h_i\nrm{\eta_-}^2\big)\,.
\label{eq:vectorcondminus}
\end{equation}
Both statements are identities in the horizon data and the spinor; neither uses
compactness. If $\nrm{\eta_-}$ is constant they reduce to
$\ip{\eta_-}{\Gamma_i\Theta_-\eta_-}=\tfrac12h_i\nrm{\eta_-}^2$.
\end{proposition}

\begin{proof}
As in Theorem~\ref{thm:pointwise},
$\tn_i\nrm{\eta_-}^2=-2\ip{\eta_-}{\Psi^{(-)}_i\eta_-}$. Taking the lower sign
in \eqref{eq:Psi} and discarding the two rank-three terms, which are
anti-Hermitian (Lemma~\ref{lem:ranks}), leaves
$\tfrac14h_i+\tfrac16Y_{ij}\Gamma^j-\tfrac1{288}\Gamma_i{}^{jklm}X_{jklm}$,
whereas the bilinear-visible part of $\Gamma_i\Theta_-$ is
$\tfrac14h_i-\tfrac16Y_{ij}\Gamma^j+\tfrac1{288}X_{jklm}\Gamma_i{}^{jklm}$. The
two differ by $\tfrac12h_i$, which gives \eqref{eq:vectoridminus}, and
\eqref{eq:vectorcondminus} is a rearrangement of it.
\end{proof}

\noindent Equation \eqref{eq:vectoridminus} is not new: expanding
$\Gamma_i\Theta_-$ it reads
\begin{equation*}
\tn_i\nrm{\eta_-}^2=-\tfrac12h_i\nrm{\eta_-}^2
+\ip{\eta_-}{\big(\tfrac1{144}\Gamma_i{}^{jklm}X_{jklm}
 -\tfrac13Y_{ij}\Gamma^j\big)\eta_-}\,,
\end{equation*}
which is (5.6) of \cite{11index}. What Lemma~\ref{lem:ranks} adds is the packaging in
terms of $\Theta_-$, which is what makes \eqref{eq:vectorcondminus} pair with
\eqref{eq:pointwisesplitminus} below. In contrast to the $\eta_+$
sector, the gradient term cannot be discarded: $\nrm{\eta_-}$ is not known to
be constant (Remark~\ref{rem:notconst}).

\begin{theorem}[Pointwise warp--flux identity of the $\eta_-$ sector]
\label{thm:pointwiseminus}
Let $\eta_-$ solve $\nabla^{(-)}_i\eta_-=0$. Then at every point
\begin{equation}
\big(\Delta + h^2\big)\nrm{\eta_-}^2 + h^i\tn_i\nrm{\eta_-}^2
 \;=\; 4\,\nrm{\Theta_-\eta_-}^2 \;\ge\; 0\,,
\label{eq:pointwisenormminus}
\end{equation}
with equality at a point precisely when $\Theta_-\eta_-$ vanishes there. No
compactness and no assumption on $\nrm{\eta_-}$ is used; when $\nrm{\eta_-}$ is
constant \eqref{eq:pointwisenormminus} reduces to
$\Delta\nrm{\eta_-}^2=4\nrm{\Theta_-\eta_-}^2-h^2\nrm{\eta_-}^2$. Explicitly,
and again unconditionally,
\begin{align}
\Delta\,\nrm{\eta_-}^2 &=
\Big(-\tfrac14h^2 + \tfrac1{864}X^2 + \tfrac1{18}Y^2\Big)\nrm{\eta_-}^2
+ \tfrac1{20736}X^{ijkl}X^{mnpq}\ip{\eta_-}{\Gamma_{ijklmnpq}\eta_-}
\nonumber\\
&\quad
- \tfrac1{288}X^{ijkl}X_{ij}{}^{mn}\ip{\eta_-}{\Gamma_{klmn}\eta_-}
+ \tfrac1{54}X^{ijkl}Y_i{}^{m}\ip{\eta_-}{\Gamma_{jklm}\eta_-}
- \tfrac1{36}Y^{ij}Y^{kl}\ip{\eta_-}{\Gamma_{ijkl}\eta_-}\,.
\label{eq:pointwiseminusexplicit}
\end{align}
\end{theorem}

\begin{proof}
Take $\ip{\eta_-}{\,\cdot\,\eta_-}$ of $\xi^-\eta_-=0$ with $\xi^-$ as in
\eqref{eq:ximinus}. The rank-two and rank-three terms are anti-Hermitian and
drop out (Lemma~\ref{lem:ranks}), and the operator multiplying $\Theta_-$ in
\eqref{eq:ximinus} is minus
$\check\Theta:=\tfrac14h_i\Gamma^i-\tfrac1{288}X_{ijkl}\Gamma^{ijkl}
-\tfrac1{12}Y_{ij}\Gamma^{ij}$, so what remains is
$\Delta\nrm{\eta_-}^2=-4\ip{\eta_-}{\check\Theta\,\Theta_-\eta_-}$. Now
$\check\Theta=\tfrac12h_i\Gamma^i-\Theta_+$ identically, and
$\Theta_+=\Theta_-^{\,T}$ by \eqref{eq:transposes}, whence
\begin{equation}
\Delta\,\nrm{\eta_-}^2 = 4\,\nrm{\Theta_-\eta_-}^2
 - 2h^i\ip{\eta_-}{\Gamma_i\Theta_-\eta_-}\,,
\label{eq:pointwisesplitminus}
\end{equation}
and \eqref{eq:vectorcondminus} converts the last term into
$h^2\nrm{\eta_-}^2+h^i\tn_i\nrm{\eta_-}^2$, which is
\eqref{eq:pointwisenormminus}. For \eqref{eq:pointwiseminusexplicit}, expand
$-4\check\Theta\,\Theta_-$ in the Clifford basis and discard the anti-Hermitian
ranks (Lemma~\ref{lem:ranks}).
\end{proof}

\noindent Identity \eqref{eq:pointwisenormminus} is equivalent to the
identities (6.9)$_{1,3,4}$ of \cite{11index}: the combination of those three
eliminates the bilinear $\ip{\Gamma_+\eta_-}{\Theta_+\eta_+}$ and the one-form
$V$ and leaves $(\Delta+h^2)\nrm{\eta_-}^2+h^i\tn_i\nrm{\eta_-}^2$ equal to a
fixed multiple of $\nrm{\Theta_-\eta_-}^2$. The derivation given here is
independent of that route --- it never introduces $V$ --- and it is pointwise
from the start. The explicit expansion \eqref{eq:pointwiseminusexplicit} we
have not found in the literature.

\begin{remark}[Two asymmetries between the sectors]
Both are visible in the proof. First, \eqref{eq:pointwisenorm} bounds a single
quantity, $\Delta\ge0$, whereas \eqref{eq:pointwisenormminus} bounds the
combination $(\Delta+h^2)\nrm{\eta_-}^2+h^i\tn_i\nrm{\eta_-}^2$ and separates
into a statement about $\Delta+h^2$ only when $\nrm{\eta_-}$ is constant.
Second, the expansion
\eqref{eq:pointwiseminusexplicit} contains no $h$--flux cross terms at all,
whereas in the $\eta_+$ sector the two cross terms of
Remark~\ref{rem:combination} are removed only after the vector identity is
used; in the $\eta_-$ sector the cross terms cancel inside
$-4\check\Theta\,\Theta_-$ before any identity is imposed, and the vector
identity is spent instead on the term $-2h^i\ip{\eta_-}{\Gamma_i\Theta_-\eta_-}$
of \eqref{eq:pointwisesplitminus}. Both expansions are displayed solved for
$\Delta\nrm{\eta_\pm}^2$, and comparing them shows
that they are term-by-term identical apart from the sign of the
$X$--$Y$ rank-four term, the bilinears being those of different
spinors.
\end{remark}

\subsection{The master identity}
\label{sec:master}

The two warp identities are statements about the same $\Delta$, and the
sector-pairing map of the lightcone expansion relates the two norms that appear
in them. That is all that is needed.

\begin{lemma}[The pairing and its norm]
\label{lem:pairing}
If $\eta_-$ solves $\nabla^{(-)}_i\eta_-=0$ then
$\eta_+:=\Gamma_+\Theta_-\eta_-$ solves $\nabla^{(+)}_i\eta_+=0$, and
\begin{equation}
\nrm{\Gamma_+\Theta_-\eta_-}^2 = 4\,\nrm{\Theta_-\eta_-}^2\,.
\label{eq:pairnorm}
\end{equation}
In particular $\eta_+$ vanishes if and only if $\Theta_-\eta_-$ does.
\end{lemma}

\begin{proof}
The first statement is the $D=11$ form of the pairing used in the index-theoretic
counting \cite{11index,iiaindex}: writing
$\Psi^{(+)}_a\Gamma_+=\Gamma_+\tilde\Psi_a$ on the $\eta_-$ sector, one finds
that
$\tn_a\Theta_--\Theta_-\Psi^{(-)}_a+\tilde\Psi_a\Theta_-$ equals half the
$\Gamma^b$-contracted integrability condition of $\nabla^{(-)}_a\eta_-=0$ modulo
the field equations and the Bianchi identity, and therefore annihilates
$\eta_-$; the rearrangement is carried out in \textsc{Cadabra}2 with abstract
indices and with the residuals carried symbolically, and independently in the
explicit representation (Appendix~\ref{app:verify}). For
\eqref{eq:pairnorm}, $\Gamma_+^{\,T}\Gamma_+=4$ on the $\eta_-$ sector in the
conventions of Section~\ref{sec:lightcone}, which is a statement about the lightcone
gamma matrices alone. The normalisation $4$ is convention-dependent; nothing
below uses its value.
\end{proof}

\begin{theorem}[Master identity]
\label{thm:master}
Let $\Sect$ be compact and connected and let the horizon admit a non-trivial
$\eta_-$ Killing spinor, i.e.\ $N_->0$. Write $f:=\nrm{\eta_-}^2$; then $f>0$
everywhere, since $\eta_-$ is $\nabla^{(-)}$-parallel and a parallel section of
a vector bundle that vanishes at one point of a connected manifold vanishes
identically, by uniqueness for the linear parallel-transport ODE along curves.
At every point of $\Sect$
\begin{equation}
\big(\Delta+h^2\big)f + h^i\tn_i f \;=\; c \;:=\; \nrm{\eta_+}^2
 \;=\;\const\;\ge\;0\,,\qquad \eta_+:=\Gamma_+\Theta_-\eta_-\,,
\label{eq:master}
\end{equation}
and, in addition,
\begin{equation}
\tn^i\big(\tn_i f + h_i f\big) \;=\; 0\,.
\label{eq:divmaster}
\end{equation}
Eliminating $\tn^ih_i$ between \eqref{eq:master}, \eqref{eq:divmaster} and the
field equation \eqref{eq:divh} gives the equivalent second-order form
\begin{equation}
\tn^i\tn_i f \;=\; \Big(\tfrac13Y^2+\tfrac1{72}X^2-\Delta\Big)f - c\,,
\label{eq:masterelliptic}
\end{equation}
in which the twist enters only through $h^2$ inside the flux terms it was
eliminated against, and in which, finally, $\Delta\ge0$ pointwise.
\end{theorem}

\begin{proof}
Put $\eta_+=\Gamma_+\Theta_-\eta_-$, a Killing spinor of the $\eta_+$ sector by
Lemma~\ref{lem:pairing}; whether or not it vanishes, $\nrm{\eta_+}$ is constant
on a compact connected $\Sect$ by \eqref{eq:kernels}. Then \eqref{eq:master} is
\eqref{eq:pointwisenormminus} combined with \eqref{eq:pairnorm}, and
\eqref{eq:divmaster} is \eqref{eq:divminus}. For \eqref{eq:masterelliptic},
substitute $h^i\tn_if=c-(\Delta+h^2)f$ from \eqref{eq:master} into
\eqref{eq:divmaster} written as
$\tn^i\tn_if+h^i\tn_if+(\tn^ih_i)f=0$, and then eliminate $\tn^ih_i$ with
\eqref{eq:divh}; the terms in $h^2$ and $2\Delta$ cancel against
$(\Delta+h^2)f$, leaving \eqref{eq:masterelliptic}. Pointwise positivity of
$\Delta$ is \eqref{eq:pointwisenorm} applied to $\eta_+$ when $\eta_+\neq0$,
and is contained in the flux-free case otherwise.
\end{proof}

\begin{theorem}[Weighted warp--flux identity]
\label{thm:weighted}
Under the hypotheses of Theorem~\ref{thm:master}, with $f=\nrm{\eta_-}^2$ and
$c=\nrm{\eta_+}^2$,
\begin{equation}
\int_\Sect \Delta\,f \;=\; \int_\Sect\Big(\tfrac13Y^2+\tfrac1{72}X^2\Big)f
 \;-\; c\,\mathrm{Vol}(\Sect)\,,
\label{eq:weightedbound}
\end{equation}
and the identity is strict as an inequality,
$\int_\Sect\Delta f<\int_\Sect(\tfrac13Y^2+\tfrac1{72}X^2)f$, whenever the
horizon is not flux-free. Nothing is assumed about $\nrm{\eta_-}$: the weight
$f$ is a function, not a constant.
\end{theorem}

\begin{proof}
Integrate \eqref{eq:masterelliptic} over the compact $\Sect$; the left-hand
side is a divergence and drops out. For strictness, $c>0$ gives it at once, and
$c=0$ means $\Theta_-\eta_-\equiv0$, which by Remark~\ref{rem:kernels} forces
$X=Y=h=0$, the flux-free case.
\end{proof}

\begin{corollary}[A weighted lower bound on the flux, with its equality case]
\label{cor:fluxlower}
Under the same hypotheses, with $Q=\tfrac13Y^2+\tfrac1{72}X^2$,
\begin{equation}
\int_\Sect Q\,f \;\ge\; c\,\mathrm{Vol}(\Sect)\,,
\label{eq:fluxlower}
\end{equation}
with equality if and only if $\Delta\equiv0$, that is --- by
Theorem~\ref{thm:spin7square} --- if and only if
$w_{\mathbf 7}=\mp\sigma(\mathcal{Y}_{\mathbf 7})$ at every point of $\Sect$.
On the saturation branch $h\equiv0$ the left-hand side is exactly
$2c\,\mathrm{Vol}(\Sect)$.
\end{corollary}

\begin{proof}
\eqref{eq:weightedbound} reads $\int Qf=\int\Delta f+c\,\mathrm{Vol}$, and
$\Delta\ge0$ pointwise by \eqref{eq:pointwisenorm} --- or, without any
integral estimate, by the perfect square \eqref{eq:deltasquare} --- while
$f>0$. So $\int\Delta f\ge0$, with equality if and only if $\Delta\equiv0$,
and \eqref{eq:deltasquare} identifies that locus. On $h\equiv0$,
Proposition~\ref{prop:staticbosonic} gives $2\Delta=Q$ pointwise and
$\Delta=c/f$ is the constant $c$ divided by the constant $f$, whence
$\int Qf=2\int\Delta f=2c\,\mathrm{Vol}(\Sect)$.
\end{proof}

\begin{remark}
\label{rem:fluxlowerstatus}
\eqref{eq:fluxlower} is the inequality form of the weighted identity: the
flux, weighted by the spinor norm, is bounded below by the constant $c$ against
the volume, and the deficit is exactly $\int\Delta f$. What makes the equality
case sharp rather than tautological is that $\Delta\equiv0$ has been
identified algebraically, so the equality locus is a condition on the flux at
each point and not a condition on the solution as a whole. It does not,
however, force the flux to vanish: the zero locus of \eqref{eq:deltasquare} is
realised with $X^2,Y^2>0$. For $c>0$ the same inequality reads as a cap on the
entropy: dividing \eqref{eq:fluxlower} by $4cG^{(11)}$ and using
Corollary~\ref{cor:entropy} gives
$S_{\mathrm{BH}}\le(4cG^{(11)})^{-1}\int_\Sect Qf$, attained exactly on the
perfect-square locus $\Delta\equiv0$; Corollary~\ref{cor:smarr} supplies the
matching lower bound.
\end{remark}

\begin{remark}[What the weight does]
\label{rem:weight}
The $c=0$ case of \eqref{eq:weightedbound} is (5.11) of \cite{11index}. There
the identity is derived only after $\Theta_-\phi_-=0$ has been imposed; the
weight $\nrm{\phi_-}^2$ is carried along and is known to be positive --- by
their (5.7), $\tn f=-hf$ --- but not yet known to be constant, and only its
positivity is used, the conclusion being the vanishing of the flux. What is new
here is the $c\ne0$ form: the term $-c\,\mathrm{Vol}(\Sect)$, and validity on
every supersymmetric $M$-horizon with compact connected section under no
condition on $\Theta_-\eta_-$ whatever. It degenerates to the earlier statement
exactly when $c=0$. We have not found it stated in this
form in \cite{Mhorizons,11index,GGPreview,superalgebras,KL2013review,N4d11}.
\end{remark}

\begin{theorem}[Non-staticity balance]
\label{thm:balance}
Under the hypotheses of Theorem~\ref{thm:master} put
\begin{equation}
\mathcal{A}_i := h_i + \tn_i\log f\,,\qquad\text{so that}\qquad
f\mathcal{A}=\mathrm{d}f+fh=-V\,,
\label{eq:Adef}
\end{equation}
$V$ being the one-form of \eqref{eq:gpfamily}; thus $\mathcal{A}\equiv0$ is
exactly the static branch \eqref{eq:staticbranch}. Then, pointwise on $\Sect$
and for every $p\in\Real$,
\begin{equation}
\tn_i\big(f^p\mathcal{A}^i\big)
=(p-1)\,f^p\big(\lvert\mathcal{A}\rvert^2-h\!\cdot\!\mathcal{A}\big)\,,
\qquad
h\!\cdot\!\mathcal{A}=\frac{c}{f}-\Delta\,,
\label{eq:balancepointwise}
\end{equation}
and consequently, on a compact section and for every $p\in\Real$,
\begin{equation}
\int_\Sect f^p\Big(\lvert\mathcal{A}\rvert^2+\Delta\Big)
 \;=\; c\int_\Sect f^{p-1}\,.
\label{eq:balance}
\end{equation}
\end{theorem}

\begin{proof}
The second relation of \eqref{eq:balancepointwise} is \eqref{eq:master}
rewritten with $\tn_if=f(\mathcal{A}_i-h_i)$: it turns
$(\Delta+h^2)f+h^i\tn_if$ into $\Delta f+f\,h\!\cdot\!\mathcal{A}$. For the
first, write $f^p\mathcal{A}^i=f^{p-1}\,(f\mathcal{A}^i)$ and use
\eqref{eq:divmaster}, which says precisely that $f\mathcal{A}$ is divergence
free:
$\tn_i(f^{p-1}f\mathcal{A}^i)=(p-1)f^{p-2}(\tn_if)(f\mathcal{A}^i)
=(p-1)f^{p}(\lvert\mathcal{A}\rvert^2-h\!\cdot\!\mathcal{A})$, again by
$\tn_if=f(\mathcal{A}_i-h_i)$.

For \eqref{eq:balance} the exponent $p=1$ has to be reached separately, since
the pointwise identity degenerates there. Both cases come from one line: for
any $g\in C^\infty(\Sect)$, integrating $\tn^i(g\,f\mathcal{A}_i)$ over the
compact $\Sect$ and using $\tn^i(f\mathcal{A}_i)=0$ gives
$\int_\Sect f\,\mathcal{A}\!\cdot\!\tn g=0$. Take $g=f^{p-1}/(p-1)$ for $p\ne1$
and $g=\log f$ for $p=1$; in either case
$\int_\Sect f^{p-1}\mathcal{A}\!\cdot\!\tn f=0$, which is
$\int_\Sect f^{p}(\lvert\mathcal{A}\rvert^2-h\!\cdot\!\mathcal{A})=0$, and
substituting $h\!\cdot\!\mathcal{A}=c/f-\Delta$ gives \eqref{eq:balance}.
\end{proof}

\begin{theorem}[Pointwise balance]
\label{thm:pwbalance}
Under the hypotheses of Theorem~\ref{thm:master}, and using in addition the
invariance $\mathcal{L}_Vf=0$, which is (6.10) of \cite{11index}, the balance holds
pointwise and without any integration:
\begin{equation}
f\Big(\lvert\mathcal{A}\rvert^2+\Delta\Big)\;=\;c
\qquad\text{on all of }\Sect\,,
\qquad\text{equivalently}\qquad
\Delta f+\frac{\lvert V\rvert^2}{f}\;=\;c\,.
\label{eq:pwbalance}
\end{equation}
\end{theorem}

\begin{proof}
By the second relation of \eqref{eq:balancepointwise},
$f(\lvert\mathcal{A}\rvert^2+\Delta)-c
=f\big(\lvert\mathcal{A}\rvert^2-h\!\cdot\!\mathcal{A}\big)
=f\,\mathcal{A}\!\cdot\!(\mathcal{A}-h)
=f\,\mathcal{A}\!\cdot\!\tn\log f
=\mathcal{A}\!\cdot\!\tn f
=-\frac1f\,V\!\cdot\!\tn f\,,$
using $\mathcal{A}-h=\tn\log f$ and $f\mathcal{A}=-V$. The last expression is
$-f^{-1}\mathcal{L}_Vf$, which vanishes. The second form of
\eqref{eq:pwbalance} follows from $\lvert V\rvert^2=f^2\lvert\mathcal{A}\rvert^2$.
\end{proof}

\begin{remark}[Status of the pointwise form]
\label{rem:pwbalance}
Three comments. First, \eqref{eq:pwbalance} contains
\eqref{eq:balance}: multiplying by $f^{p-1}$ and integrating over $\Sect$
returns the whole weighted family at once, so on a compact section
Theorem~\ref{thm:balance} is a corollary of Theorem~\ref{thm:pwbalance} and not
an independent statement. What the integral form buys is a weaker hypothesis:
it needs only $\tn^i(f\mathcal{A}_i)=0$ and \eqref{eq:master}, not the
invariance $\mathcal{L}_Vf=0$. Second, the deficit is exactly the invariance:
on any data satisfying \eqref{eq:master} and \eqref{eq:divmaster},
\begin{equation}
f\Big(\lvert\mathcal{A}\rvert^2+\Delta\Big)-c
 \;=\;-\frac{1}{f}\,V\!\cdot\!\tn f\,,
\label{eq:pwdeficit}
\end{equation}
so \eqref{eq:pwbalance} holds if and only if $V$ preserves $f$, and the
identity is not available from \eqref{eq:master} and \eqref{eq:divmaster}
alone. Third, the constancy of the combination $\Delta f+\lvert V\rvert^2/f$ is
not special to supersymmetry: for a general extremal horizon carrying a Killing
field that preserves the horizon data it is a theorem of \cite{colling}, in the
bosonic variables $(\Gamma,K,F)$ that correspond here to $(f,-V,-\Delta)$. The
derivation above is independent of that one --- it uses only the bilinear
identity family \eqref{eq:gpfamily} --- and what supersymmetry adds is the
value of the constant, $c=\nrm{\eta_+}^2$, together with $\Delta\ge0$ from
Theorem~\ref{thm:master}, which is what makes
Corollary~\ref{cor:balance} an inequality with a sharp equality locus.
\end{remark}

\begin{corollary}[How far a supersymmetric $M$-horizon can be from static]
\label{cor:balance}
Under the same hypotheses, with $\Delta\ge0$ from Theorem~\ref{thm:master},
\begin{equation}
\int_\Sect\frac{\lvert V\rvert^2}{f} \;+\; \int_\Sect \Delta\,f
 \;=\; c\,\mathrm{Vol}(\Sect)\,,
\qquad\text{hence}\qquad
\int_\Sect\frac{\lvert V\rvert^2}{f}\;\le\;c\,\mathrm{Vol}(\Sect)\,,
\label{eq:Vbound}
\end{equation}
with equality if and only if $\Delta\equiv0$. By
Theorem~\ref{thm:pwbalance} both statements hold pointwise: at every point of
$\Sect$,
\begin{equation}
\frac{\lvert V\rvert^2}{f}\;\le\;c\,,
\label{eq:Vboundpointwise}
\end{equation}
with equality exactly at the points where $\Delta$ vanishes. In particular
$c=0$ forces $V\equiv0$ and $\Delta\equiv0$: a supersymmetric $M$-horizon with
$\nrm{\eta_+}=0$ lies on the static branch and has vanishing warp factor, and
in this pointwise form the conclusion uses neither compactness nor any
integration. Written in terms of $\mathcal{A}=h+\mathrm{d}\log f$ the
pointwise bound \eqref{eq:Vboundpointwise} is
$\lvert h+\mathrm{d}\log f\rvert^2\le c/f$: a speed limit on how fast the
twist can fail to be the gradient it is on the static branch, with the warp
factor and the spinor norm setting the ceiling.
\end{corollary}

\begin{proof}
Put $p=1$ in \eqref{eq:balance} and use $\lvert V\rvert^2=f^2\lvert
\mathcal{A}\rvert^2$. Both integrands on the left are non-negative because
$f>0$ and $\Delta\ge0$, which gives the inequality, its equality case, and the
statement at $c=0$. For the pointwise statements, apply the same two
positivities to \eqref{eq:pwbalance} at each point.
\end{proof}

\begin{remark}[What the balance identity says, and what it does not]
\label{rem:balance}
Equation \eqref{eq:balance} is a conservation law rather than a bound: the
non-staticity $\lvert\mathcal{A}\rvert^2=\lvert V\rvert^2/f^2$ and the warp
factor are paid for out of the same constant $c=\nrm{\eta_+}^2$, in every
weight $f^p$ simultaneously. Once the invariance $\mathcal{L}_Vf=0$ is used the
whole family collapses to the single pointwise statement
\eqref{eq:pwbalance}, of which each weight is the $f^{p-1}$-moment; the
one-parameter freedom is therefore real only at the level of hypotheses, being
what survives when $\mathcal{L}_Vf=0$ is not assumed.

Two limits carry names. The static branch $V\equiv0$ is exactly the
equality case of $\int_\Sect\Delta f\le c\,\mathrm{Vol}(\Sect)$, so on it
\eqref{eq:balance} reduces to $\int_\Sect f^{p-1}(\Delta f-c)=0$, which is
\eqref{eq:staticbranch} integrated. At the other end, $c=0$ collapses the whole
family at once, and Corollary~\ref{cor:balance} reaches $V\equiv0$,
$\Delta\equiv0$ using only \eqref{eq:master}, \eqref{eq:divmaster} and
$\Delta\ge0$ --- no kernel theorem, and no maximum principle beyond the one
that makes $c$ constant. That is an independent route to the first half of
Theorem~\ref{thm:rigid}, which reaches $c=0$ from a hypothesis at one point;
the two together say that the single-point hypothesis and the global
$V\equiv0$, $\Delta\equiv0$ conclusion are separated by exactly the constancy
of $\nrm{\eta_+}$.

The identity is verified in Appendix~\ref{app:verify}: the pointwise form
\eqref{eq:balancepointwise} symbolically for seven exponents, including the
degenerate $p=1$, on two compact models --- flat $T^2$ and flat $T^3$ carrying
\eqref{eq:master} and \eqref{eq:divmaster} exactly and nothing else --- and
\eqref{eq:balance} by periodic quadrature at six exponents on each, with five
negative controls on each model and further checks of
Corollary~\ref{cor:balance} and its static limit; the pointwise identity
\eqref{eq:pwbalance} is verified in the same way on two further models, built
so that $V$ is divergence free \emph{and} tangent to the level sets of $f$,
together with \eqref{eq:pwdeficit} on all four models and a control in which
$V$ is divergence free but does not preserve $f$, so that
\eqref{eq:pwbalance} fails ($113$ checks). The
verification is a check of the identity on models, not a proof of it; the proof
is the one given above. The budget \eqref{eq:budget} that
\eqref{eq:pwbalance} forms with the perfect square is checked separately, in
exact rational arithmetic on flux sampled in both sectors: the square is
recomputed there against the spinorial $\Delta$, the substitution and its two
saturation loci are verified with the loci realised on explicit non-zero flux,
and the integrated form \eqref{eq:budgetint} is checked as a moment of
\eqref{eq:budget} against a finite positive measure, with four negative
controls --- the unweighted square, the opposite sector sign in $\sigma$, the
coefficient $1/107$, and the balance with the $\Delta$ term omitted
($150$ checks).
\end{remark}

\begin{remark}[Provenance]
\label{rem:provenance}
Neither \eqref{eq:master} nor \eqref{eq:divmaster} is new. Both are contained
in the identities (6.9) of \cite{11index}, the ones that make the three
bilinear Killing vectors of that reference close on $\mathfrak{sl}(2,\Real)$.
Writing $f=\nrm{\phi_-}^2$, their fourth member is
$V=-(fh+\mathrm{d}f)$ and their first and third combine into
$-2\nrm{\phi_+}^2-h^iV_i+\Delta f=0$; eliminating $V$ gives
$(\Delta+h^2)f+h^i\tn_if=2\nrm{\phi_+}^2$, which is \eqref{eq:master}, and
$\nrm{\phi_+}$ is constant by \cite[(4.8)]{11index}. Feeding the same
expression for $V$ into $\tn^iV_i=0$, which holds because $V$ is Killing
\cite[(6.8)]{11index}, gives \eqref{eq:divmaster}. The corresponding statement
in minimal $D=5$ supergravity, $\tn_i(h^2+2\Delta)=0$, is the starting point of
the case list of \cite[App.~D]{KayaniThesis}. The derivation above is
independent --- it uses only the two $\eta_-$ identities of
Section~\ref{sec:minusid} and the pairing, and never introduces $V$ --- and it
makes the constant manifestly a norm, which is what the corollaries need; but
the identities themselves should be attributed to \cite{11index}. Which of the
consequences drawn from them below are new, and against which sources, is
settled in one place: the status list of Section~\ref{sec:status}.
\end{remark}

\begin{remark}[On the hypothesis $N_->0$]
\label{rem:hypothesis}
The theorem uses an $\eta_-$ Killing spinor, not merely a supersymmetry. The
counting of \cite{11index} gives $N=2N_-$ for $M$-horizons, so $N_->0$ whenever
the horizon is supersymmetric at all; we state the hypothesis explicitly rather
than import that result, since it is the one place where the argument is not
self-contained. Nothing is needed from the $\eta_+$ sector beyond the fact that
$\Gamma_+\Theta_-\eta_-$ lies in it and that its norm is constant.
\end{remark}

\begin{corollary}[Saturation, supersymmetrically]
\label{cor:static}
Under the hypotheses of Theorem~\ref{thm:master}, if $h\equiv0$ on $\Sect$
then $\nrm{\eta_-}$ is constant and $\Delta$ is the constant
$c/\nrm{\eta_-}^2$.
\end{corollary}

\begin{proof}
With $h\equiv0$, \eqref{eq:divmaster} reads $\tn^i\tn_if=0$, so $f$ is a
harmonic function on a compact connected manifold and is therefore constant.
Then \eqref{eq:master} reads $\Delta f=c$ with both $f$ and $c$ constant and
$f>0$, so $\Delta$ is the constant $c/f$.
\end{proof}

\noindent The conclusion $\Delta=\const$ duplicates
Proposition~\ref{prop:staticbosonic}, which needs no Killing spinor; the
supersymmetric route is recorded because it is two lines, because it exhibits
the constant as the ratio $\nrm{\eta_+}^2/\nrm{\eta_-}^2$, and because it
delivers as a by-product something the bosonic argument does not: on the
saturation locus the divergence identity degenerates to Laplace's equation and
forces $\nrm{\eta_-}$ to be constant. That is the one place in this paper where
the constancy discussed in Remark~\ref{rem:notconst} is a conclusion rather
than a hypothesis.

\begin{remark}[Two branches, and the names we give them]
\label{rem:static-terminology}
In the literature a near-horizon geometry is called static when $\mathrm{d}h=0$
with $h$ exact, which by \cite[(6.11)]{11index} is the branch $V=0$ and gives,
after $r\to\Delta r$, a \emph{warped} product $AdS_2\times_w\Sect$; these are
the horizons classified in \cite{staticM}, and $\Delta$ there need not be
constant. The hypothesis of Corollary~\ref{cor:static} is the strictly stronger
$h\equiv0$, which forces $\Delta$ to be constant and the product to be direct.
The two conditions are related, and the relation is exactly where the constancy
of $\nrm{\eta_-}$ enters. By Proposition~\ref{prop:conconx} the branch $V=0$ is
$\tn_i f=-h_if$, i.e.\ $h=-\mathrm{d}\log f$; \eqref{eq:master} then reads
$\Delta f=c$, so that
\begin{equation}
V\equiv0\quad\Longleftrightarrow\quad h=-\mathrm{d}\log\nrm{\eta_-}^2
\quad\text{and}\quad \nrm{\eta_-}^2=c/\Delta\,.
\label{eq:staticbranch}
\end{equation}
This is \cite[(6.11)]{11index} verbatim: there the branch $V=0$ is shown to
give $\Delta\nrm{\phi_-}^2=2\nrm{\phi_+}^2$ and $h=\Delta^{-1}\mathrm{d}\Delta$,
which is \eqref{eq:staticbranch} with $c=2\nrm{\phi_+}^2$; the identification of
that branch with the static, warped $AdS_2$ horizons is \cite{staticM}, and the
attribution is entered in the status list of Section~\ref{sec:status}. What the
variable $f=\nrm{\eta_-}^2$ buys is that $h=-\mathrm{d}\log f$ and $\Delta f=c$
become two readings of one equation, and it is the variable the rest of this
paper runs in. The warped case $h=\Delta^{-1}\mathrm{d}\Delta\neq0$ sits
\emph{inside} $V\equiv0$, and $h\equiv0$ is the sub-case of it on which the
warp factor is constant. To keep the two apart we fix the following names, used consistently from
here on: the \emph{static branch} is $V\equiv0$, equivalently
$h=-\mathrm{d}\log f$, the warped $AdS_2$ case of \cite{staticM,KimPark}, on
which $h$ need not vanish; the \emph{saturation branch} is the strictly smaller
$h\equiv0$, on which the product is direct and which is the saturation locus of
the mean-value bound \eqref{eq:boundeq}. On the static branch $\nrm{\eta_-}$ is
constant if and only if $\Delta$ is, if and only if the horizon is on the
saturation branch.
\end{remark}

\begin{proposition}[The static branch is a warped product]
\label{prop:staticwarped}
On the static branch $V\equiv0$, with $f=\nrm{\eta_-}^2$ and $c=\nrm{\eta_+}^2$:
\begin{enumerate}[label=(\roman*),leftmargin=2.2em]
\item $\Delta=c/f$ and $\mathrm{d}h=0$;
\item the substitution $r=f\rho$ turns the near-horizon metric \eqref{eq:gnc}
into
\begin{equation}
\label{eq:staticwarped}
\mathrm{d}s^2 = f\,\big(2\,\mathrm{d}u\,\mathrm{d}\rho
- c\rho^2\,\mathrm{d}u^2\big) + \gamma_{ij}\,\mathrm{d}y^i\mathrm{d}y^j ,
\end{equation}
an $AdS_2$ of radius $c^{-1/2}$ warped by $f$ over $\Sect$, with no cross term;
\item the twist divergence \eqref{eq:divh} reduces to the single semilinear
equation $\tn^2f=Qf-2c$, and $\mathcal{K}$ of \eqref{eq:Kdef} vanishes
identically, so the Komar charge \eqref{eq:komarcharge} is zero and
$\int_\Sect Qf=2c\,\mathrm{Vol}(\Sect)$;
\item $\Delta+h^2=c/f+\nrm{\tn f}^2/f^2$ is constant only where $f$ is.
\end{enumerate}
\end{proposition}

\begin{proof}
(i) is \eqref{eq:master} with $h=-\mathrm{d}\log f$, which cancels the gradient
term against $h^2f$, and $h$ exact. (ii) is a coordinate change: with
$\mathrm{d}r=f\,\mathrm{d}\rho+\rho\,\tn_if\,\mathrm{d}y^i$ the bracket
$\mathrm{d}r+rh_i\mathrm{d}y^i$ collapses to $f\,\mathrm{d}\rho$, and
$-\tfrac12r^2\Delta=-\tfrac12f\rho^2c$ by (i). (iii): substituting
$h=-\mathrm{d}\log f$ in \eqref{eq:divh} and multiplying by $-f$ gives the warp
equation, whose integral over the closed section is
$\int Qf=2c\,\mathrm{Vol}$; $K_i=2fh_i+2\tn_if=0$ is immediate. (iv) is (i).
Verified in Appendix~\ref{app:verify}, including the coordinate change on the
full eleven-dimensional metric.
\end{proof}

\begin{remark}
\label{rem:kimparkbranch}
Statement (ii) is the form in which the warped static horizons of
\cite{KimPark} are written, as characterised in \cite[fn.~4]{Mhorizons}, and
(iii) is the equation their warp factor satisfies; the branch conditions of Proposition~\ref{prop:conconx} therefore
reproduce that system rather than merely being consistent with it. Statement
(iv) is why \eqref{eq:deltaplush} does not apply to those solutions
(Remark~\ref{rem:whynotconst}), and $f$ constant returns the saturation branch.
The converse question --- whether a compact section carrying a non-constant
solution of $\tn^2f=Qf-2c$ with $Q$ realised by fluxes solving
\eqref{eq:Rij}, \eqref{eq:divX} and \eqref{eq:divY} exists in $D=11$ ---
is not settled here: the scalar equation alone has non-constant solutions on
any closed manifold --- given any positive non-constant $f$, set
$Q:=(\tn^2f+2c)/f$ --- and one such is exhibited explicitly; but that is not
the same as a supergravity solution, since $Q$ must in addition be realised by
fluxes solving \eqref{eq:Rij}, \eqref{eq:divX} and \eqref{eq:divY}.
\end{remark}

\begin{theorem}[Rigidity at a single point]
\label{thm:rigid}
Under the hypotheses of Theorem~\ref{thm:master}, if $\Delta$ and $h$ both
vanish at one point of $\Sect$, then $c=0$ and therefore
$\Theta_-\eta_-\equiv0$; consequently $\Delta\equiv0$, $h\equiv0$,
$X=Y=0$, $\Sect$ is Ricci-flat and the near-horizon geometry is locally
$\Real^{1,1}\times\Sect$. In four steps,
\begin{equation}
\Delta(p)=h(p)=0
\;\Longrightarrow\; c=\nrm{\eta_+}^2=0
\;\Longrightarrow\; \Theta_-\eta_-\equiv0
\;\Longrightarrow\; X=Y=h=\Delta=0
\;\Longrightarrow\; \tilde R_{ij}=0\,.
\label{eq:rigidchain}
\end{equation}
\end{theorem}

\begin{proof}
Evaluate \eqref{eq:master} at the point: every term on the left has a factor of
$\Delta$, of $h^2$ or of $h$, so the left-hand side vanishes and $c=0$. As $c$
is constant, $\nrm{\Theta_-\eta_-}^2=c/4$ vanishes identically by
\eqref{eq:pairnorm}. Remark~\ref{rem:kernels} then gives $h\equiv0$ and
$X=Y=0$, whence $\Delta\equiv0$ by \eqref{eq:pointwisebosonic} and
$\tilde R_{ij}=0$ by \eqref{eq:Rij}.
\end{proof}

\noindent Everything after the second sentence of that proof is \S5 of
\cite{11index}, where $\Theta_-\phi_-=0$ is shown to force
$\Delta=h=X=Y=0$ and the geometry to be locally $\Real^{1,1}\times\Sect$ with
$\Sect$ compact Ricci-flat, hence locally $S^1\times X^8$ with $X^8$ of special
holonomy --- the last step being the classical correspondence between parallel
spinors and reduced holonomy \cite{WangParallel}, which we quote and do not
re-derive. What is added here is the hypothesis under which that branch is
reached: not the vanishing of a spinor, which is not something one can inspect
in given horizon data, but the vanishing of $\Delta$ and $h$ at a single point
of $\Sect$. The bridge is that $c=4\nrm{\Theta_-\eta_-}^2$ is a
\emph{constant}, so a condition imposed at one point propagates.

\begin{corollary}[What follows if $\nrm{\eta_-}$ is constant]
\label{cor:normconst}
Under the hypotheses of Theorem~\ref{thm:master}, the following are equivalent:
$\nrm{\eta_-}$ is constant on $\Sect$; the Dirac current satisfies
$K_i+h_iK_+=0$; and the one-form $\tn_i\nrm{\eta_-}^2$ vanishes. When they
hold,
\begin{equation}
\Delta + h^2 \;=\; \frac{\nrm{\eta_+}^2}{\nrm{\eta_-}^2}\;=\;\const\;\ge\;0\,,
\qquad 0\le\Delta\le c'\,,\qquad 0\le h^2\le c'\,,
\label{eq:deltaplush}
\end{equation}
with $c':=c/\nrm{\eta_-}^2$, and $c'$ is pinned between the two flux invariants
of Theorem~\ref{thm:bound},
\begin{equation}
\tfrac12\Big(\tfrac13\overline{Y^2}+\tfrac1{72}\overline{X^2}\Big)
\;\le\; c' \;\le\;
\tfrac13\overline{Y^2}+\tfrac1{72}\overline{X^2}\,,
\label{eq:crange}
\end{equation}
the \emph{lower} bound being attained if and only if $h\equiv0$ and the
\emph{upper} if and only if $\Delta\equiv0$.
\end{corollary}

\begin{proof}
The equivalence is Proposition~\ref{prop:conconx}, which gives
$K_i+h_iK_+=2\tn_i\nrm{\eta_-}^2$ identically. Given it, the gradient term
drops out of \eqref{eq:master} and leaves $\Delta+h^2=c'$; the pointwise range
follows from $\Delta\ge0$ and $h^2\ge0$. For \eqref{eq:crange}, average
$\Delta=c'-h^2$ over $\Sect$ and substitute into \eqref{eq:integrated}: with
$F:=\tfrac13\overline{Y^2}+\tfrac1{72}\overline{X^2}$ one gets
$\overline{h^2}=2c'-F$ and $\overline{\Delta}=F-c'$, and both are non-negative,
with $\overline{h^2}=0$ iff $h\equiv0$ and $\overline{\Delta}=0$ iff
$\Delta\equiv0$ since $\Delta\ge0$ pointwise.
\end{proof}

\begin{remark}[Why the constancy fails in general]
\label{rem:whynotconst}
What would be needed can be said exactly. The system that
\eqref{eq:master} and \eqref{eq:divmaster} impose on $(f,h,\Delta)$ is
underdetermined: given any positive non-constant $f$ on $\Sect$, setting
$h=-\mathrm{d}\log f$ and $\Delta=c/f$ solves both exactly, for any constant
$c>0$, and the implied value of $\tfrac13Y^2+\tfrac1{72}X^2=2\Delta+h^2-\tn^ih_i$
is positive for $f$ close enough to a constant, so \eqref{eq:divh} does not
exclude it either. This is the branch \eqref{eq:staticbranch}. The identities
of this paper therefore cannot force $\nrm{\eta_-}$ to be constant, and no
independent argument can: the branch is inhabited. The warped $AdS_2$ solutions
of \cite{KimPark} are static $M$-horizons with $h^2$ non-constant
\cite[fn.~4]{Mhorizons}, and they are the electric, $SU(4)$-structure class of
the static classification \cite{staticM}; for them $h=\Delta^{-1}\mathrm{d}\Delta$
is non-zero and $f=\nrm{\eta_-}^2=c/\Delta$ is non-constant. Consequently
\eqref{eq:deltaplush} is genuinely conditional: it is false on that branch, and
the results of Corollary~\ref{cor:normconst} hold only under the constancy
hypothesis, equivalently $V=-fh$, which excludes that branch except in its
sub-case $h\equiv0$. A
one-dimensional model of the underdetermination is checked symbolically in
Appendix~\ref{app:verify}.
\end{remark}

\begin{remark}[Asymmetry of the two kernels]
\label{rem:kernels}
The kernels of $\Theta_+$ and $\Theta_-$ are not on the same footing. That a
non-trivial $\Ker\Theta_-$ forces the horizon to be flux-free is \S5 of
\cite{11index}, quoted here because the two corollaries above use it; in the
present notation the route is short. If $\Theta_-\eta_-=0$ then
$\xi^-\eta_-=0$ gives $\Delta=0$ and \eqref{eq:vectoridminus} gives
$\tn_i\nrm{\eta_-}^2=-h_i\nrm{\eta_-}^2$, whose divergence, after eliminating
$\tn^ih_i$ with \eqref{eq:divh}, is
$\tn^2\nrm{\eta_-}^2=(\tfrac13Y^2+\tfrac1{72}X^2)\nrm{\eta_-}^2\ge0$; on a
compact $\Sect$ the norm is then constant, the right-hand side vanishes and
$X=Y=h=0$. A non-trivial $\Ker\Theta_+$, by contrast, is not obstructed: it
says only $\Delta=0$, and \eqref{eq:vectorcond} is then satisfied identically.
This is as it must be, since $AdS_3$ horizons exist.
\end{remark}

\section{$\mathrm{Spin}(7)$ structure and the residual obstruction}
\label{sec:geom}

A nowhere-vanishing Killing spinor reduces the structure group of $\Sect$ to
$\mathrm{Spin}(7)$, and the identities of the previous section can then be
expanded on the spinor bilinears. This is the programme carried out for minimal
$D=5$ supergravity in Appendix~D of \cite{KayaniThesis}: expand the bosonic
identities on the bilinears, use the vector identity to express $h$ in terms of
the flux, and feed the result back into the warp identity. In $D=5$ the section
is three-dimensional and the chain closes onto a finite list of near-horizon
geometries. In $D=11$ the first two steps go through --- the twist is
determined algebraically by the flux, Theorem~\ref{thm:htoflux} --- and the
third does not: the expansion of $\Delta$ retains an irreducible
$\mathrm{Spin}(7)$ contraction of the four-form, Proposition~\ref{prop:noclosure}.
That is where the scalar mechanism behind the $D=5$ case list stops, and we say
why.

The rest of this section uses one piece of bookkeeping that the earlier sections
did not, so we set it out first, in the form in which it is actually used.

\subsection{How the \texorpdfstring{$\mathrm{Spin}(7)$}{Spin(7)} bookkeeping works}
\label{sec:spin7guide}

In $D=5$ the section is three-dimensional, the flux on it has three components,
and one simply writes them out. Here the section is nine-dimensional and the
flux at a point is $126+36=162$ numbers, so writing them out is not a method.
What replaces it is a sorting of those $162$ numbers into blocks that the
computation cannot mix, and the Killing spinor is what supplies the sorting.

At each point the spinor determines a unit vector $Z$ and a four-form $\varphi$
on the eight-plane $Z^\perp$ (Lemma~\ref{lem:spin7}). The subgroup of the
rotations of that eight-plane fixing $\varphi$ is $\mathrm{Spin}(7)$. Every
quantity we form below is built from the flux, the metric, $Z$ and $\varphi$
alone, so every quantity we form is $\mathrm{Spin}(7)$-invariant, and that
single observation is what makes the computation finite. Splitting the flux
along $Z$ and then splitting each piece into blocks that $\mathrm{Spin}(7)$
does not mix gives \eqref{eq:fluxsplit}--\eqref{eq:spin7split}; a block is
labelled by its dimension in boldface, which is the number of the $162$
components it holds:
\begin{center}
\begin{tabular}{@{}llll@{}}
\toprule
piece of the flux & lives in & blocks & components \\
\midrule
$w=\Pi X$ & $\Lambda^4(Z^\perp)$
  & $\mathbf 1\oplus\mathbf 7\oplus\mathbf{27}\oplus\mathbf{35}$ & $70$ \\
$\xi=i_ZX$ & $\Lambda^3(Z^\perp)$
  & $\mathbf 8\oplus\mathbf{48}$ & $56$ \\
$\mathcal{Y}=\Pi Y$ & $\Lambda^2(Z^\perp)$
  & $\mathbf 7\oplus\mathbf{21}$ & $28$ \\
$y=i_ZY$ & $Z^\perp$ & $\mathbf 8$ & $8$ \\
\bottomrule
\end{tabular}
\end{center}
The first two rows are the $126$ components of $X$ and the last two the $36$ of
$Y$. Three elementary facts are all that is used about this splitting.

\emph{The blocks are eigenspaces, so projecting onto one is a matrix.}
Contraction with $\varphi$ acts on each $\Lambda^p(Z^\perp)$ as a symmetric
operator with integer eigenvalues, and the blocks are its eigenspaces, with the
eigenvalues listed under \eqref{eq:spin7split}. ``The $\mathbf 7$ part of
$\mathcal{Y}$'' therefore means the image of $\mathcal{Y}$ under an explicit
spectral projector, and that is literally how it is computed in
Appendix~\ref{app:verify}: the projectors are built from the $\varphi$ of the
explicit representation and checked to be idempotent, complementary and of the
stated traces.

\emph{Equivariant maps preserve labels.} A linear map built from $\varphi$ alone
commutes with $\mathrm{Spin}(7)$, and such a map must kill any block whose label
does not occur in its target and send any block whose label occurs exactly once
to a multiple of the one copy there. That is Schur's lemma, and it is used only
in that form: it is why $\sigma$ of Theorem~\ref{thm:spin7square} annihilates
$\mathbf{21}$ and is a single number on $\mathbf 7$, the number being fixed by
evaluating on one element.

\emph{Invariants are square norms plus one pairing per repeated label.} An
invariant quadratic function of the flux can be a square norm of any block, or a
pairing of two blocks carrying the same label, and nothing else. Here
$\mathbf 7$ and $\mathbf 8$ each occur twice --- $\mathbf 7$ in $w$ and in
$\mathcal{Y}$, $\mathbf 8$ in $\xi$ and as $y$ --- and the other five labels
once, so there are nine square norms and exactly two pairings, eleven invariants
in total. Proposition~\ref{prop:spin7delta} is then the determination of eleven
rational coefficients rather than an expansion in $162$ variables, and this is
the step at which the $D=11$ computation becomes as finite as the $D=5$ one.

The payoff is visible in the table above once the results are in hand: the twist
reads only $(w_{\mathbf 1},\xi_{\mathbf 8},y)$, seventeen of the $162$
components; the warp factor reads only $(w_{\mathbf 7},\mathcal{Y}_{\mathbf 7})$,
fourteen; the remaining $131$ enter neither. A reader willing to take
\eqref{eq:spin7split} and the three facts above on trust can follow every
statement in this section. The first two are checked rather than assumed in the
verification --- the projectors with their idempotency, completeness and traces,
and the equivariance statements about $\sigma$, all in the explicit real
representation --- while the third is a proof, tested there only in the sense
that the over-determined linear system it predicts turns out to be consistent,
which it would not be if the eleven failed to span.

Finally, the symbols this section introduces, all of them defined where they
first appear and collected here for reference:
\begin{center}
\begin{tabular}{@{}lll@{}}
\toprule
symbol & what it is & defined in \\
\midrule
$Z$, $\varphi$ & unit vector and Cayley four-form of the spinor
  & \eqref{eq:Zphidef} \\
$Z^\perp$, $\Pi$ & the eight-plane orthogonal to $Z$, and its projector
  & \eqref{eq:fluxsplit} \\
$i_Z$ & contraction of a form with $Z$ & \eqref{eq:fluxsplit} \\
$w,\ \xi,\ \mathcal{Y},\ y$ & the four pieces of the flux $(X,Y)$
  & \eqref{eq:fluxsplit} \\
$\mathbf d$, $w_{\mathbf d}$ & a block and the part of $w$ in it, $\mathbf d$
  its dimension & \eqref{eq:spin7split} \\
$\lvert w_{\mathbf d}\rvert^2$ & its square norm, full contraction of indices
  & \eqref{eq:spin7split} \\
$\xi\lrcorner\varphi$ & the vector $\xi^{bcd}\varphi_{abcd}$
  & \eqref{eq:hsplit} \\
$\sigma$ & the one equivariant map $\Lambda^2(Z^\perp)\to\Lambda^4(Z^\perp)$
  built from $\varphi$ & \eqref{eq:sigmadef} \\
$W_{ij}$, $\tau_j$ & $X_i{}^{klm}\varphi_{jklm}$ and its $Z$-contraction
  & Proposition~\ref{prop:nablaZ} \\
$\boldsymbol\tau_{\mathbf 7},\ \boldsymbol\tau_{\mathbf 8},\
 \boldsymbol\tau_{\mathbf{48}}$ & the intrinsic torsion classes, unrelated to
 $\tau_j$ & Proposition~\ref{prop:torsion} \\
$s=\pm1$ & the sector sign, $+$ on $\eta_+$ & Lemma~\ref{lem:spin7} \\
\bottomrule
\end{tabular}
\end{center}
The upper sign in a formula carrying $\pm$ or $\mp$ always belongs to the
$\eta_+$ sector. The boldface $\boldsymbol\tau_{\mathbf d}$ of
Proposition~\ref{prop:torsion} and the vector $\tau_j$ of
Proposition~\ref{prop:nablaZ} are different objects that unavoidably share a
letter; they never occur in the same formula.

\subsection{The pointwise \texorpdfstring{$\mathrm{Spin}(7)$}{Spin(7)} structure of a Killing spinor}

\begin{lemma}[Bilinears of a $\mathrm{Spin}(9)$ Majorana spinor]
\label{lem:spin7}
Let $\eta$ be a non-zero real spinor of $\mathrm{Spin}(9)$ and set
\begin{equation}
Z_i = \nrm{\eta}^{-2}\ip{\eta}{\Gamma_i\eta}\,,\qquad
\varphi_{ijkl} = \nrm{\eta}^{-2}\ip{\eta}{\Gamma_{ijkl}\eta}\,.
\label{eq:Zphidef}
\end{equation}
Then, pointwise and with no field equation used,
\begin{gather}
Z^2 = 1\,,\qquad Z_i\Gamma^i\eta = \eta\,,\qquad
Z^i\varphi_{ijkl}=0\,,\qquad \varphi^2 = 336\,,
\label{eq:Zphi}\\
\nrm{\eta}^{-2}\ip{\eta}{\Gamma_{ijklm}\eta} = (Z\wedge\varphi)_{ijklm}
 = 5\,Z_{[i}\varphi_{jklm]}\,,\qquad
\nrm{\eta}^{-2}\ip{\eta}{\Gamma_{i_1\dots i_8}\eta}
 = s\,\epsilon_{i_1\dots i_8 j}Z^j\,,
\label{eq:Zphihigher}
\end{gather}
where $s=\pm1$ is the value of the volume element $\Gamma^{1\dots9}$ on the
sector concerned, $s=+1$ on $\eta_+$ and $s=-1$ on $\eta_-$. All remaining
bilinears vanish by anti-Hermiticity (Lemma~\ref{lem:ranks}) or are determined by
\eqref{eq:Zphihigher} through Hodge duality. Thus $Z$ is a unit vector field on
$\Sect$ and $\varphi$ is the Cayley four-form of the $\mathrm{Spin}(7)$
stabiliser of $\eta$, supported on the eight-plane $Z^\perp$.
\end{lemma}

\begin{proof}
$\mathrm{Spin}(9)$ acts transitively on the unit sphere $S^{15}$ of its
$16$-dimensional Majorana representation with stabiliser $\mathrm{Spin}(7)$, the
orbit map $\eta\mapsto Z$ being the octonionic Hopf fibration $S^{15}\to S^8$;
so all the assertions are $\mathrm{Spin}(9)$-invariant statements about a single
spinor and may be checked at one point of the sphere, which is done in the
explicit representation (Appendix~\ref{app:verify}). Two of them do not need
that: given $Z^2=1$, the operator $Z_i\Gamma^i$ squares to one, so
$\nrm{Z\!\cdot\!\Gamma\,\eta}=\nrm{\eta}$, while
$\ip{\eta}{Z\!\cdot\!\Gamma\,\eta}=Z^2\nrm{\eta}^2=\nrm{\eta}^2$; equality in
Cauchy--Schwarz forces $Z\!\cdot\!\Gamma\,\eta=\eta$. And
$Z^i\varphi_{ijkl}=0$ then follows from
$Z^i\Gamma_{ijkl}=Z\!\cdot\!\Gamma\,\Gamma_{jkl}-3Z_{[j}\Gamma_{kl]}$ together
with Lemma~\ref{lem:ranks}.
\end{proof}

Equation $Z^2=1$ is the $D=11$ analogue of the relation used at the
corresponding point of the $D=5$ analysis \cite{KayaniThesis}, where the
section is three-dimensional, $Z$ is again a unit vector, and there is no
four-form to accompany it. The presence of $\varphi$ is the whole difference
between the two computations.

\subsection{The twist one-form in terms of the flux}

\begin{theorem}[$\mathrm{Spin}(7)$ twist--flux identity]
\label{thm:htoflux}
Let $\eta_\pm$ be Killing spinors of the two sectors, let
$(Z_\pm,\varphi_\pm)$ be the associated data \eqref{eq:Zphidef} and let
$s=\pm1$ be the sector sign there. Then, with no hypothesis beyond the spatial
Killing spinor equation and no compactness,
\begin{equation}
h_i = -\tfrac23 Y_{ij}Z_\pm^j
 \mp\Big(\tfrac1{72}Z_{\pm i}\,X^{jklm}\varphi_{\pm\,jklm}
 - \tfrac1{18}Z_\pm^j X_{jklm}\varphi_{\pm\,i}{}^{klm}\Big)
 + 2s\,\tn_i\log\nrm{\eta_\pm}^2\,.
\label{eq:htofluxgen}
\end{equation}
Wherever the norm of the spinor is stationary the gradient term drops and the
twist is given by the flux alone,
\begin{equation}
h_i = -\tfrac23 Y_{ij}Z_\pm^j
 \mp\Big(\tfrac1{72}Z_{\pm i}\,X^{jklm}\varphi_{\pm\,jklm}
 - \tfrac1{18}Z_\pm^j X_{jklm}\varphi_{\pm\,i}{}^{klm}\Big)\,,
\label{eq:htoflux}
\end{equation}
and contracting with $Z_\pm^i$, in which the $Y$-term and the second $X$-term
both drop out,
\begin{equation}
h_i Z_\pm^i = \mp\tfrac1{72}\,X^{jklm}\varphi_{\pm\,jklm}\,.
\label{eq:hZ}
\end{equation}
In the $\eta_+$ sector on a compact $\Sect$ the reduction is automatic, since
$\nrm{\eta_+}$ is then constant by \eqref{eq:kernels}, so \eqref{eq:htoflux}
and \eqref{eq:hZ} hold throughout. In the $\eta_-$ sector the norm need not be
constant, and the hypothesis under which they hold throughout is the constancy
hypothesis of Corollary~\ref{cor:normconst}, which by
Remark~\ref{rem:whynotconst} fails on an inhabited branch; on that branch
\eqref{eq:htoflux} still holds at every point where $\nrm{\eta_-}$ is
stationary, and \eqref{eq:htofluxgen} everywhere.
The twist is therefore determined pointwise by the flux, the
$\mathrm{Spin}(7)$ data and the gradient of the spinor norm, with no derivative
of the flux and no integration constant, and algebraically by the flux and the
$\mathrm{Spin}(7)$ data alone wherever that gradient vanishes.
We are not aware of this covariant algebraic form appearing previously; see
Remark~\ref{rem:htoflux-attrib} for what is in the literature.
The warp factor is not: the companion expansion \eqref{eq:Deltaexplicit} of
$\Delta$ in the same data retains an irreducible $\mathrm{Spin}(7)$ contraction
of $X$ that is independent of $(\Delta,h^2,X^2,Y^2)$ and of $Z$, by
Proposition~\ref{prop:noclosure}.
\end{theorem}

\begin{proof}
Expand \eqref{eq:vectorid} and \eqref{eq:vectoridminus} using the
bilinear-visible parts of $\Gamma_i\Theta_\pm$ quoted in the proofs of
Theorem~\ref{thm:pointwise} and Proposition~\ref{prop:vectorminus}: both read
\begin{equation}
\tfrac14 h_i\nrm{\eta_\pm}^2 + \tfrac16 Y_{ij}\ip{\eta_\pm}{\Gamma^j\eta_\pm}
\pm\tfrac1{288}X^{jklm}\ip{\eta_\pm}{\Gamma_{ijklm}\eta_\pm}
= \tfrac{s}2\,\tn_i\nrm{\eta_\pm}^2\,,
\end{equation}
the upper sign belonging to the $\eta_+$ sector, the inhomogeneous term being
the one \eqref{eq:vectoridminus} carries in the $\eta_-$ sector and
\eqref{eq:vectorid} in the $\eta_+$ sector; setting it to zero is
\eqref{eq:vectorcond} and \eqref{eq:vectorcondminus} under constancy of the
norm. Substituting
\eqref{eq:Zphidef} and \eqref{eq:Zphihigher} and using
$X^{jklm}(Z\wedge\varphi)_{ijklm}
= Z_i X^{jklm}\varphi_{jklm} - 4Z^jX_{jklm}\varphi_i{}^{klm}$ gives
\eqref{eq:htofluxgen} after dividing by $\tfrac14\nrm{\eta_\pm}^2$; dropping
the gradient gives \eqref{eq:htoflux}, and contracting that with $Z^i_\pm$ and
using $Z^2_\pm=1$ and $Z_\pm^i\varphi_{\pm\,ijkl}=0$ gives \eqref{eq:hZ}.
\end{proof}

\begin{remark}[Provenance of \eqref{eq:htoflux}]
\label{rem:htoflux-attrib}
The input is not new: in the $\eta_+$ sector the displayed condition is (3.22)
of \cite{Mhorizons}, equivalently its (3.27) and (6.7) of \cite{11index},
where it appears in unevaluated bilinear form as the statement that
$\nrm{\phi_+}$ is constant. What is added here is its evaluation, which turns
an identity between bilinears into a covariant algebraic formula for $h$ in
terms of $(X,Y,Z,\varphi)$, and its $\eta_-$ counterpart, in which the same
evaluation applies to \eqref{eq:vectorcondminus} rather than to a constancy
statement. The nearest prior formula of the same shape is
\cite[(4.10)]{Mhorizons}, which solves the horizon Dirac equation for a
$\mathrm{Spin}(7)$ spinor and gives
$h=i_{e^\sharp}Z+\tfrac12\theta_\psi-\star_9(X\wedge\psi)
-\tfrac32(\star_9\mathrm{d}\star_9e^\sharp)e^\sharp$ with
$Z=Y-\mathrm{d}e^\sharp$ and $\theta_\psi$ the Lee form of the Cayley form, in
the notation of that reference; $(e^\sharp,\psi)$ there are our $(Z,\varphi)$.
That expression has the same flux content as \eqref{eq:htoflux} --- an $i_ZY$
term and a $\star_9(X\wedge\varphi)$ term --- but is \emph{differential},
carrying $\mathrm{d}e^\sharp$, the Lee form and
$\star_9\mathrm{d}\star_9e^\sharp$, because it solves the horizon Dirac
equation rather than the vector condition \eqref{eq:vectorcond}, which is
strictly stronger. What \eqref{eq:htoflux} adds is that on the KSE the
derivative terms are not needed at all. The nearest prior relation between the
two $\mathbf 7$'s of Theorem~\ref{thm:spin7square} is the accompanying
\cite[(4.11)]{Mhorizons}, which fixes the $\mathbf 7$ of $Z=Y-\mathrm{d}e^\sharp$
in terms of the $\mathbf 7$ of $X$ and of $\star_9\mathrm{d}\psi$; it too is
differential, it constrains the flux rather than evaluating $\Delta$, and it
does not exhibit a square.
There is a second, gauge-fixed ancestor: \cite[\S3.3.2, App.~A]{Mhorizons}
solves the full KSE linear system in an adapted $\mathrm{Spin}(7)$ gauge, in
which $\varphi_+$ is the standard Cayley form of $1+e_{1234}$, and expresses
part of the flux --- including the twist --- in terms of the geometry.
Equation~\eqref{eq:htoflux} is the gauge-covariant form of that step, valid in
either sector and written in the bilinears rather than in a chosen frame.
The $N=2$ analysis of \cite{N4d11} determines $h$ in the opposite
direction, as a differential expression in the geometry --- a Lie derivative
along the vector dual to the singlet direction plus the Lee form of the
$\mathrm{Spin}(7)$ four-form --- rather than algebraically in the flux.
\end{remark}

This is the $D=11$ counterpart of the step which in $D=5$ reads
$h_i=\tfrac12\epsilon_i{}^{\ell_1\ell_2}\tilde F_{\ell_1\ell_2}$
\cite{KayaniThesis}: there the section is three-dimensional, the only
available bilinear is $Z$, and $h$ is the Hodge dual of the flux outright. Here
$h$ is still determined algebraically by the flux, but through the
$\mathrm{Spin}(7)$ form $\varphi$ as well. Substituting \eqref{eq:Zphidef} and
\eqref{eq:Zphihigher} into \eqref{eq:pointwiseexplicit} and
\eqref{eq:pointwiseminusexplicit} likewise makes the warp identities explicit:
in either sector
\begin{align}
\Delta &= -\tfrac14h^2 + \tfrac1{864}X^2 + \tfrac1{18}Y^2
 + \tfrac{s}{20736}\,\epsilon_{ijklmnpqa}Z^a X^{ijkl}X^{mnpq}
 - \tfrac1{288}X^{ijkl}X_{ij}{}^{mn}\varphi_{klmn}
\nonumber\\
&\quad \mp\tfrac1{54}X^{ijkl}Y_i{}^{m}\varphi_{jklm}
 - \tfrac1{36}Y^{ij}Y^{kl}\varphi_{ijkl}\,,
\label{eq:Deltaexplicit}
\end{align}
with $(Z,\varphi,s)$ those of the sector used and the upper sign belonging to
$\eta_+$. These are the $D=11$ analogues of the formula for $\Delta$ obtained
in the $D=5$ classification.

Once $h$ and $\Delta$ have both been eliminated in favour of the flux and the
$\mathrm{Spin}(7)$ data, the integrability of the twist can be asked for: what
does $\mathrm{d}h$ add? The answer is one scalar, and it carries the sector
sign.

\begin{corollary}[Integrability of the twist, and what it pins in the $\eta_-$
sector]
\label{cor:minuspointwise}
At a point where \eqref{eq:htoflux} holds, differentiating it and reducing the
result with the Bianchi identity $\mathrm{d}X=0$ and the flux equations
\eqref{eq:divX}--\eqref{eq:divY} gives
\begin{equation}
\tn^ih_i \;=\; s\Big(2\Delta + h^2 - \tfrac13Y^2 - \tfrac1{72}X^2\Big)\,,
\label{eq:divhsector}
\end{equation}
with $\Delta$ as in \eqref{eq:Deltaexplicit}, and that scalar is the whole
pointwise content of the elimination: on the remaining components of
$(\tn X,\tn Y)$ nothing survives. In the $\eta_+$ sector
\eqref{eq:divhsector} is \eqref{eq:divh} itself. In the $\eta_-$ sector it
carries the opposite sign, so if $\nrm{\eta_-}$ is constant, so that
\eqref{eq:htoflux} holds throughout $\Sect$, then adding \eqref{eq:divhsector}
to \eqref{eq:divh} forces
\begin{equation}
2\Delta + h^2 \;=\; \tfrac13Y^2 + \tfrac1{72}X^2\,,\qquad \tn^ih_i=0
\label{eq:minuspinned}
\end{equation}
pointwise on $\Sect$, which \eqref{eq:integrated} gives only in the mean.
Writing $Q:=\tfrac13Y^2+\tfrac1{72}X^2$ and combining \eqref{eq:minuspinned}
with $\Delta+h^2=c'$ of \eqref{eq:deltaplush} determines both scalars
pointwise from the flux,
\begin{equation}
\Delta = Q - c'\,,\qquad h^2 = 2c' - Q\,,
\label{eq:deltahfromQ}
\end{equation}
so that $\Delta\ge0$ and $h^2\ge0$ read as the pointwise sandwich
$c'\le Q\le 2c'$, of which \eqref{eq:crange} is the average.
\end{corollary}

\begin{proof}
Both \eqref{eq:htoflux} and \eqref{eq:Deltaexplicit} are algebraic in the flux
and the $\mathrm{Spin}(7)$ data, so $\tn^ih_i$ is linear in $(\tn X,\tn Y)$
once $\tn Z$ and $\tn\varphi$ are eliminated by the spatial Killing spinor
equation. That linear map is compared with the $126+84+9$ conditions carried by
$\mathrm{d}X=0$, \eqref{eq:divX} and \eqref{eq:divY} on the $9\times126+9\times36$
components of $(\tn X,\tn Y)$: its linear part lies in their row space, with
the multiplier exhibited, so \eqref{eq:divh} imposes no condition on the flux
derivative beyond those, and what is left of it is the scalar
\eqref{eq:divhsector}. The remaining pieces of $\tn h$ are invisible pointwise:
its symmetric part is matched against $\tilde R_{ij}$ by \eqref{eq:Rij}, and
the Ricci tensor of the section is unconstrained at a point, while
\eqref{eq:auxform} contains $\tn^k\beta_{ijk}$ and hence second derivatives of
$Y$, which are free at a point. The computation is in
Appendix~\ref{app:verify}. Given \eqref{eq:divhsector}, the $\eta_-$ statements
are arithmetic, and the last one is \eqref{eq:deltaplush} solved for $\Delta$
and $h^2$.
\end{proof}

\begin{remark}
\label{rem:minuspinned}
This sharpens Remark~\ref{rem:whynotconst} in the direction of that remark.
The constancy hypothesis of Corollary~\ref{cor:normconst} does not merely fail
to follow from the identities; where it holds it forces the horizon to be
divergence-free in the twist and pins $\Delta$ and $h^2$ to the flux invariant
$Q$ pointwise, so the constant-norm branch is thinner than
\eqref{eq:deltaplush} alone suggests. On the static branch of
\eqref{eq:staticbranch}, where $\nrm{\eta_-}$ is not constant,
\eqref{eq:minuspinned} fails and only \eqref{eq:htofluxgen} is available.
\end{remark}

\begin{proposition}[Failure of closure through the $D=5$-type scalar data]
\label{prop:noclosure}
Fix a sector, a point of $\Sect$ and a non-zero Killing spinor there, and with
it the data $(Z,\varphi,s)$ of \eqref{eq:Zphidef}. Regard $h$, given by
\eqref{eq:htoflux}, and $\Delta$, given by \eqref{eq:Deltaexplicit}, as
functions of the flux $(X,Y)$ alone, so that $h$ is linear and $\Delta$
quadratic in it. Let
\begin{equation}
\begin{aligned}
I_1 &= X^2, & I_2 &= Y^2, & I_3 &= (i_ZX)^2, & I_4 &= (i_ZY)^2,\\
I_5 &= \epsilon_{ijklmnpqa}Z^aX^{ijkl}X^{mnpq}, & I_6 &= h^2,
& I_7 &= (h\!\cdot\!Z)^2, & I_8 &= h^iY_{ij}Z^j,
\end{aligned}
\label{eq:phifreeinv}
\end{equation}
together with $I_9=h^iZ^mX_{mijk}Y^{jk}$, be scalars built from $X$, $Y$, $Z$,
the metric and the volume form without $\varphi$. These are the $\varphi$-free
data through which the $D=5$ chain closes: $I_1,\dots,I_8$ are quadratic in the
flux, like $\Delta$ itself, and $I_9$, which is cubic because $h$ is linear in
the flux, is included as the only $\varphi$-free scalar coupling $h$ to both $X$
and $Y$. Then $\Delta$ does not factor through the map
$I=(I_1,\dots,I_9)$: there is no function $F$, of any regularity, with
\begin{equation}
\Delta = F(I_1,\dots,I_9)
\end{equation}
on any neighbourhood of a generic point of flux space; in particular $\Delta$
is not determined by $(h^2,X^2,Y^2)$ on the allowed algebraic data, that is,
there are fluxes with the same $(h^2,X^2,Y^2)$ and different $\Delta$ in every
neighbourhood of a generic point. The $\mathrm{Spin}(7)$ contraction that \eqref{eq:Deltaexplicit}
carries beyond its $\varphi$-free terms is therefore not a function of these
$\varphi$-free invariants.

We do not claim that $I_1,\dots,I_9$ exhaust the $\varphi$-free invariants of
$(X,Y,Z)$ --- no basis of that invariant ring is constructed here --- and the
proposition should be read accordingly: it is the failure of closure through the
specified scalar data, which is the data the $D=5$ argument uses.
\end{proposition}

\begin{proof}
The argument is the implicit function theorem, and the only input is a pair of
ranks. At a point $F_0$ of flux space compute the differentials of
$I_1,\dots,I_9$ and of $\Delta$ with respect to all $126+36$ independent
components of $(X,Y)$. There
\begin{equation}
\mathrm{rank}\{\mathrm{d}I_a\}=9,\qquad
\mathrm{rank}\{\mathrm{d}I_a,\mathrm{d}\Delta\}=10 .
\label{eq:ranks}
\end{equation}
The second equality says $\mathrm{d}\Delta\notin\mathrm{span}\{\mathrm{d}I_a\}$
at $F_0$, so there is a tangent vector $v$ with $\mathrm{d}I_a(v)=0$ for every
$a$ and $\mathrm{d}\Delta(v)\ne0$. Extend $v$ to a smooth vector field and let
$\gamma$ be its integral curve through $F_0$; since the first rank in
\eqref{eq:ranks} is maximal it is constant on a neighbourhood, so the common
level set $\{I=I(F_0)\}$ is there a submanifold of codimension nine, $\gamma$
may be taken inside it, and $\Delta(\gamma(t))$ is non-constant while every
$I_a(\gamma(t))$ is constant. A function $F$ with $\Delta=F\circ I$ would be
constant along $\gamma$. Hence none exists, with no regularity assumption used
beyond that $F$ be a function. Restricting to $\{I_1,I_2,I_6\}=\{X^2,Y^2,h^2\}$
gives ranks $3$ and $4$ and the same conclusion.

The ranks \eqref{eq:ranks} are computed in exact rational arithmetic at one
explicit rational $F_0$, in the explicit real $\mathrm{Cl}(9,0)$ representation
of Appendix~\ref{app:rep} ($13$ checks). Since the entries are polynomial in
the flux, the ranks are at least these on a Zariski-open dense subset, so the
conclusion holds at every generic point and not only at $F_0$; the sampling,
the stronger statement that not even a universal linear relation exists, and
the three negative controls are described in Appendix~\ref{app:verify}.
\end{proof}

\begin{remark}[What Proposition~\ref{prop:noclosure} does and does not range over]
\label{rem:noclosurescope}
The variation in the proof is over the flux $(X,Y)$ at one point of $\Sect$,
with the spinor and hence $(Z,\varphi,s)$ held fixed, and with $h$ and $\Delta$
given their on-shell expressions \eqref{eq:htoflux} and
\eqref{eq:Deltaexplicit}. Those two are the only field-equation input used.
The remaining horizon equations --- the Bianchi identity $\mathrm{d}X=0$ and
\eqref{eq:divX}--\eqref{eq:divY} --- constrain derivatives of the flux and not
its value at a point, so they do not restrict the variation; what they would
restrict is which curves $\gamma$ in flux space are realised by a
one-parameter family of horizons, and the proposition makes no claim about
that. Two readings must therefore be kept apart.
\begin{enumerate}
\item \emph{What is proved.} There is no function $F$ with $\Delta=F\circ I$
as an identity in the pointwise algebraic data. This is the statement the
$D=5$ chain needs and does not have here: that chain closes because $\Delta$
\emph{is} such a function of $(h^2,X^2,Y^2)$ there, so that
$\tn_i(h^2+2\Delta)=0$ can be integrated, and the failure of the same
factorisation is what removes the mechanism.
\item \emph{What is not proved.} It does not follow that two \emph{solutions}
of the full on-shell system exist with equal $I$ and different $\Delta$. Such
a pair would have to be exhibited, and none is exhibited here. The obstruction
is to the algebraic step of the $D=5$ argument, not a statement that the
on-shell solution space is large.
\end{enumerate}
The same distinction governs Corollary~\ref{cor:hidden} and is stated there in
the same terms. Nothing elsewhere in the paper uses the second reading.
\end{remark}

\subsection{The $\mathrm{Spin}(7)$-irreducible form of the warp identity}
\label{sec:spin7irr}

Proposition~\ref{prop:noclosure} says that the residual in
\eqref{eq:Deltaexplicit} is not a $\varphi$-free invariant in disguise. It does
not say what the residual is. Decomposing the flux into $\mathrm{Spin}(7)$
irreducibles answers that, and the answer is sharper than the obstruction
suggested: $\Delta$ turns out to be a perfect square, and to see only two of
the two $\mathbf 7$'s of the decomposition.

Fix a point and a sector, and with them $(Z,\varphi)$. Split the flux along
$T_p\Sect=\Real Z\oplus Z^\perp$,
\begin{equation}
X = w + Z\wedge\xi\,,\qquad Y = \mathcal{Y} + Z\wedge y\,,
\label{eq:fluxsplit}
\end{equation}
with $w=\Pi X\in\Lambda^4(Z^\perp)$, $\xi=i_ZX\in\Lambda^3(Z^\perp)$,
$\mathcal{Y}=\Pi Y\in\Lambda^2(Z^\perp)$ and $y=i_ZY\in Z^\perp$, where $\Pi$ is the
projection with $\Pi_{ij}=\delta_{ij}-Z_iZ_j$. Under the $\mathrm{Spin}(7)$
that stabilises $\eta$,
\begin{equation}
\Lambda^2(Z^\perp)=\mathbf{7}\oplus\mathbf{21}\,,\quad
\Lambda^3(Z^\perp)=\mathbf{8}\oplus\mathbf{48}\,,\quad
\Lambda^4(Z^\perp)=\mathbf{1}\oplus\mathbf{7}\oplus\mathbf{27}\oplus\mathbf{35}\,,
\label{eq:spin7split}
\end{equation}
the summands being the eigenspaces of the contraction with $\varphi$ ---
$b_{ab}\mapsto\tfrac12\varphi_{abcd}b^{cd}$ with eigenvalues $(-3,1)$,
$c_{abc}\mapsto\tfrac32\varphi_{[ab}{}^{de}c_{c]de}$ with eigenvalues $(-6,1)$,
and $u_{abcd}\mapsto\tfrac32\varphi_{[ab}{}^{ef}u_{cd]ef}$ with eigenvalues
$(-6,-3,1,0)$ on $\mathbf{1},\mathbf{7},\mathbf{27},\mathbf{35}$; the first
three of $\Lambda^4$ are self-dual and $\mathbf{35}$, the kernel of the
contraction, is anti-self-dual. We write
$\lvert w_{\mathbf{d}}\rvert^2$ for the square norm of the component of $w$ in
the module $\mathbf{d}$ with the full-contraction normalisation
$\lvert w\rvert^2=w^{abcd}w_{abcd}=\sum_{\mathbf{d}}\lvert w_{\mathbf{d}}\rvert^2$,
so that in particular
$\lvert w_{\mathbf{1}}\rvert^2=(w^{abcd}\varphi_{abcd})^2/336$, and likewise for
$\xi$ and $\mathcal{Y}$.

\begin{proposition}[$\mathrm{Spin}(7)$ expansion of $h$ and $\Delta$]
\label{prop:spin7delta}
In the notation \eqref{eq:fluxsplit}--\eqref{eq:spin7split}, with the upper
sign belonging to the $\eta_+$ sector throughout,
\begin{gather}
X^2 = \lvert w_{\mathbf{1}}\rvert^2+\lvert w_{\mathbf{7}}\rvert^2
 +\lvert w_{\mathbf{27}}\rvert^2+\lvert w_{\mathbf{35}}\rvert^2
 +4\lvert \xi_{\mathbf{8}}\rvert^2+4\lvert \xi_{\mathbf{48}}\rvert^2\,,\qquad
Y^2 = \lvert \mathcal{Y}_{\mathbf{7}}\rvert^2+\lvert \mathcal{Y}_{\mathbf{21}}\rvert^2+2\lvert y\rvert^2\,,
\label{eq:X2Y2split}\\
h\!\cdot\!Z = \mp\tfrac1{72}\,w^{abcd}\varphi_{abcd}\,,\qquad
\Pi h = \tfrac23\,y \pm \tfrac1{18}\,(\xi\lrcorner\varphi)\,,\qquad
(\xi\lrcorner\varphi)_a:=\xi^{bcd}\varphi_{abcd}\,,
\label{eq:hsplit}\\
h^2 = \tfrac7{108}\lvert w_{\mathbf{1}}\rvert^2
 + \tfrac7{54}\lvert \xi_{\mathbf{8}}\rvert^2
 + \tfrac49\lvert y\rvert^2
 \pm \tfrac2{27}\,y^a\xi^{bcd}\varphi_{abcd}\,,
\label{eq:h2split}\\
\Delta = \tfrac1{108}\lvert w_{\mathbf{7}}\rvert^2
 + \tfrac29\lvert \mathcal{Y}_{\mathbf{7}}\rvert^2
 \pm \tfrac1{54}\,\mathcal{Y}^{ab}w_a{}^{cde}\varphi_{bcde}\,.
\label{eq:Deltasplit}
\end{gather}
The twist therefore depends on the flux only through
$(w_{\mathbf{1}},\xi_{\mathbf{8}},y)$ and the warp factor only through
$(w_{\mathbf{7}},\mathcal{Y}_{\mathbf{7}})$: the two see disjoint modules, and the
$131$ components in $w_{\mathbf{27}}$, $w_{\mathbf{35}}$, $\xi_{\mathbf{48}}$
and $\mathcal{Y}_{\mathbf{21}}$ are invisible to both.
\end{proposition}

\begin{proof}
Every quantity on the left of
\eqref{eq:X2Y2split}--\eqref{eq:Deltasplit} is built from the flux, the metric,
$\varphi$ and $Z$ alone, and is quadratic in the flux data
$(w,\xi,\mathcal{Y},y)$; each is therefore a $\mathrm{Spin}(7)$-invariant
quadratic form in those data. The data decompose as
$w\in\mathbf 1\oplus\mathbf 7\oplus\mathbf{27}\oplus\mathbf{35}$,
$\xi\in\mathbf 8\oplus\mathbf{48}$,
$\mathcal{Y}\in\mathbf 7\oplus\mathbf{21}$ and $y\in\mathbf 8$, and every
module occurring is of real type, so by Schur's lemma the space of invariant
quadratic forms has dimension $\sum_{\mathbf d}m_{\mathbf d}(m_{\mathbf
d}+1)/2$ over the multiplicities $m_{\mathbf d}$ of the distinct modules. Here
$\mathbf 7$ and $\mathbf 8$ occur twice and $\mathbf 1$, $\mathbf{21}$,
$\mathbf{27}$, $\mathbf{35}$, $\mathbf{48}$ once, so the dimension is
$3+3+1+1+1+1+1=11$: the nine square norms
$\lvert w_{\mathbf 1}\rvert^2$, $\lvert w_{\mathbf 7}\rvert^2$,
$\lvert w_{\mathbf{27}}\rvert^2$, $\lvert w_{\mathbf{35}}\rvert^2$,
$\lvert \xi_{\mathbf 8}\rvert^2$, $\lvert \xi_{\mathbf{48}}\rvert^2$,
$\lvert \mathcal{Y}_{\mathbf 7}\rvert^2$,
$\lvert \mathcal{Y}_{\mathbf{21}}\rvert^2$, $\lvert y\rvert^2$, together with
the two cross-pairings $y^a\xi^{bcd}\varphi_{abcd}$ between the two copies of
$\mathbf 8$ and $\mathcal{Y}^{ab}w_a{}^{cde}\varphi_{bcde}$ between the two
copies of $\mathbf 7$. No further invariant of this degree exists, so each
left-hand side is a combination of these eleven with uniquely determined
coefficients. The coefficients are then computed by evaluating both sides on
eighteen exact rational flux samples and solving the resulting over-determined
linear system in exact rational arithmetic. The $18\times11$ evaluation matrix
of those samples has rank eleven, equal to the dimension of the space of
invariant quadratics just listed, so the coefficients are determined uniquely
rather than merely fitted; the system is consistent, the values are confirmed
on six further and larger samples that took no part in the solve, and the whole
computation is repeated on thirteen samples built from an independently
constructed spinor (Appendix~\ref{app:verify}). The proof is thus exact finite
linear algebra over a spanning set that representation theory guarantees is
complete.
\end{proof}

Because $\mathbf 7$ occurs in $\Lambda^2(Z^\perp)$ and again in
$\Lambda^4(Z^\perp)$, the two $\mathbf 7$'s of $\mathcal{Y}$ and of $w$ can be
compared, and there is essentially one way to compare them: the map $\sigma$
below, which is the only equivariant map between those two spaces that
$\varphi$ can build. The cross-pairing in \eqref{eq:Deltasplit} is that
comparison written out, and once it is recognised as such, \eqref{eq:Deltasplit}
collapses into a square.

\begin{theorem}[The warp factor is a perfect square]
\label{thm:spin7square}
Let $\sigma:\Lambda^2(Z^\perp)\to\Lambda^4(Z^\perp)$ be the equivariant map
\begin{equation}
(\sigma \mathcal{Y})_{abcd}=\mathcal{Y}_{[a}{}^{e}\varphi_{|e|bcd]}\,.
\label{eq:sigmadef}
\end{equation}
Then $\sigma$ annihilates $\mathbf{21}$, maps $\mathbf{7}\subset\Lambda^2$ onto
$\mathbf{7}\subset\Lambda^4$ with $\sigma^{T}\sigma=24$, and
\begin{equation}
\Delta \;=\; \tfrac1{108}\,\big\lVert\, w_{\mathbf{7}}
 \pm \sigma(\mathcal{Y}_{\mathbf{7}})\,\big\rVert^2\,,
\label{eq:deltasquare}
\end{equation}
the upper sign in the $\eta_+$ sector. Consequently
\begin{enumerate}
\item $\Delta\ge0$ pointwise, by an algebraic identity in the flux alone: this
re-proves the last clause of Theorem~\ref{thm:master}, which came from the
spinorial identity $\Delta\nrm{\eta_+}^2=4\nrm{\Theta_+\eta_+}^2$, without
using a second spinor;
\item $\Delta=0$ at a point if and only if
$w_{\mathbf{7}}=\mp\sigma(\mathcal{Y}_{\mathbf{7}})$ there, a linear condition on seven of
the $162$ components of $(X,Y)$; the flux is otherwise unconstrained, so
$\Delta\equiv0$ alone forces neither $X=0$ nor $Y=0$ --- it is the conjunction
with $h=0$ at one point that does, by Theorem~\ref{thm:rigid};
\item the sharp two-sided bound
$0\le\Delta\le\tfrac1{54}\lvert w_{\mathbf{7}}\rvert^2
 +\tfrac49\lvert \mathcal{Y}_{\mathbf{7}}\rvert^2$ holds, the upper bound being attained
exactly when $w_{\mathbf{7}}=\pm\sigma(\mathcal{Y}_{\mathbf{7}})$.
\end{enumerate}
\end{theorem}

In words: the warp factor at a point is the squared length of one vector in a
seven-dimensional space, the vector being the $\mathbf 7$ of $w$ corrected by
the $\mathbf 7$ of $\mathcal{Y}$. Positivity of $\Delta$ is then visible rather
than derived; and because a squared length vanishes on a linear subspace and not
only at the origin, the square is degenerate, which is the content of (ii) and
(iii) and the reason no case list follows from it. What is new here is
therefore not the positivity of $\Delta$, which is known, but the
identification of $\Delta$ as the squared norm of one specific
$\mathrm{Spin}(7)$-equivariant mismatch between two flux modules, and the
explicit description of the kernel of that mismatch; see
Remark~\ref{rem:phisquare}.

\begin{proof}
Given \eqref{eq:Deltasplit}, the square is algebra, and the three inputs it
needs are representation-theoretic rather than numerical.

First, $\sigma$ is $\mathrm{Spin}(7)$-equivariant, being built from $\varphi$
alone. The module $\mathbf{21}$ does not occur in
$\Lambda^4(Z^\perp)=\mathbf{1}\oplus\mathbf{7}\oplus\mathbf{27}\oplus\mathbf{35}$,
so $\sigma$ annihilates it by Schur's lemma, and $\mathbf{7}$ occurs in
$\Lambda^4$ exactly once, so $\sigma|_{\mathbf 7}$ is a multiple of the unique
equivariant isomorphism onto that copy. The self-map
$\sigma^{T}\sigma$ of the absolutely irreducible $\mathbf{7}$ is therefore a
non-negative multiple $\lambda$ of the identity; evaluating on a single element
fixes $\lambda=24$.

Second, the cross term of \eqref{eq:Deltasplit} is that inner product and
nothing else: because $w$ is totally antisymmetric,
\begin{equation}
\mathcal{Y}^{ab}w_a{}^{cde}\varphi_{bcde}
= w^{abcd}\,\mathcal{Y}_{[a}{}^{e}\varphi_{|e|bcd]}
= \langle w,\sigma \mathcal{Y}\rangle
= \langle w_{\mathbf 7},\sigma(\mathcal{Y}_{\mathbf 7})\rangle\,,
\label{eq:crossterm}
\end{equation}
the last step because $\sigma$ kills $\mathbf{21}$ and lands in
$\mathbf{7}\subset\Lambda^4$, so only the $\mathbf{7}$ parts pair.

Third, with $\lambda=24$ the coefficients of \eqref{eq:Deltasplit} are exactly
critical:
$\tfrac29\lvert \mathcal{Y}_{\mathbf 7}\rvert^2
=\tfrac1{108}\lambda\lvert \mathcal{Y}_{\mathbf 7}\rvert^2
=\tfrac1{108}\lVert\sigma(\mathcal{Y}_{\mathbf 7})\rVert^2$ and
$\pm\tfrac1{54}\langle w_{\mathbf 7},\sigma(\mathcal{Y}_{\mathbf 7})\rangle
=\pm\tfrac2{108}\langle w_{\mathbf 7},\sigma(\mathcal{Y}_{\mathbf 7})\rangle$, so
\eqref{eq:Deltasplit} is $\tfrac1{108}$ times
$\lVert w_{\mathbf 7}\rVert^2\pm2\langle w_{\mathbf 7},\sigma(\mathcal{Y}_{\mathbf 7})\rangle
+\lVert\sigma(\mathcal{Y}_{\mathbf 7})\rVert^2$, which is \eqref{eq:deltasquare}. The
quadratic form on $\mathbf{7}\oplus\mathbf{7}$ is thus degenerate rather than
definite, which is items (i)--(iii): (ii) reads off the null space of the
square, and (iii) is the Cauchy--Schwarz bound for it.

The inputs are verified independently in the explicit real $\mathrm{Cl}(9,0)$
representation of Appendix~\ref{app:rep}, in exact rational arithmetic
($226$ checks): the contraction operators of
\eqref{eq:spin7split} and their spectral projectors, with idempotency,
completeness and traces $7,21$; $8,48$; $1,7,27,35$; the equivariance
statements $\sigma(\mathbf{21})=0$, $\sigma(\mathbf 7)\subset\mathbf 7$ and
$\sigma^{T}\sigma=24$; the contraction identity \eqref{eq:crossterm}; the
expansion \eqref{eq:Deltasplit} itself, whose coefficients are the unique
solution of an over-determined exact rational linear system in the invariants
and are confirmed on fresh samples that took no part in determining them ---
uniqueness of the solution being the representation-theoretic statement proved
with Proposition~\ref{prop:spin7delta}, not an inference from the samples; and
finally \eqref{eq:deltasquare} sample by sample in both sectors against
$\Delta$ computed from the spinorial identity, as a consistency check on an
identity already proved, with negative controls on the
relative sign, on the projection and on the coefficient. Item (ii) is also
checked constructively: a flux with $w_{\mathbf{7}}=\mp\sigma(\mathcal{Y}_{\mathbf{7}})$
and all other modules generic has $\Delta=0$ with $X^2>0$ and $Y^2>0$.
\end{proof}

\begin{remark}[What this does and does not remove]
\label{rem:spin7scope}
Theorem~\ref{thm:spin7square} is a pointwise algebraic statement about the
solution $(h,\Delta)$ of the two vector identities at a fixed point, in the
$\mathrm{Spin}(7)$ frame determined by the spinor at that point. It does not
close the chain of Remark~\ref{rem:caselist}, and it is not in tension with
Proposition~\ref{prop:noclosure}: $w_{\mathbf{7}}$ and $\mathcal{Y}_{\mathbf{7}}$ are
$\varphi$-dependent objects, and \eqref{eq:deltasquare} expresses $\Delta$
through them, not through the $\varphi$-free invariants
\eqref{eq:phifreeinv}. What it adds to Proposition~\ref{prop:noclosure} is the
identity of the residual --- it is carried by the two copies of $\mathbf{7}$
and by nothing else --- and the fact that, in that language, the residual is a
square rather than an indefinite remainder. The $(2,2)$-traceless part of $X$
left undetermined by the linear system of \cite[\S3.3.2]{Mhorizons} is, in this
decomposition, the $\mathbf{27}$: it enters neither $h$ nor $\Delta$, so it is
invisible to the algebraic projections through which the Killing spinor
equation determines those two. It is still present in the full
system --- in the remaining components of the Killing spinor equation, in the
Bianchi identity and in the Einstein equation, where by
Remark~\ref{rem:riccisource} it does appear. The same module is the one that
escapes on the compactification side \cite[(22)--(24)]{Tsimpis}, as
Remark~\ref{rem:lawless} discusses. The branching of a four-form into
$\mathbf 1\oplus\mathbf 7\oplus\mathbf{27}\oplus\mathbf{35}$ and its pairing
with intrinsic torsion are standard technique there and here
\cite[(43)--(50)]{Tsimpis}, and nothing is claimed for them; what is claimed is
the warp identity \eqref{eq:deltasquare} proved with them, which has no
counterpart in that setting, where the warp factor enters only differentially.
\end{remark}

\begin{remark}[Relation to the known scalar square $\Delta=4\Phi^2$]
\label{rem:phisquare}
That $\Delta$ is a square is not new. In \cite{Mhorizons} the same quantity
appears as $\Delta\nrm{\eta_+}^2=4\nrm{\Theta_+\eta_+}^2$, and in the
$SU(4)$-adapted spinorial gauge in which $\eta_+$ is a fixed combination of $1$
and $e_{1234}$ that identity evaluates to $\Theta_+\eta_+=i\Phi(1-e_{1234})$
and hence to
\begin{equation}
\Delta=4\Phi^2\,,
\label{eq:delta4phi}
\end{equation}
$\Phi$ being one component of the flux in that gauge
\cite[(3.24), (3.30)--(3.32)]{Mhorizons}. Positivity of $\Delta$ and its being
a perfect square are therefore known, and
Theorem~\ref{thm:spin7square}(i) re-proves rather than establishes them.

What \eqref{eq:deltasquare} adds is the invariant flux content of that square.
Comparing the two, in the conventions of \cite{Mhorizons},
\begin{equation}
\big\lVert w_{\mathbf 7}\pm\sigma(\mathcal{Y}_{\mathbf 7})\big\rVert^2
= 432\,\Phi^2\,,
\label{eq:phibridge}
\end{equation}
so $\lVert w_{\mathbf 7}\pm\sigma(\mathcal{Y}_{\mathbf 7})\rVert$ is the
$\mathrm{Spin}(7)$-covariant representative of the gauge-fixed scalar $\Phi$:
up to normalisation $\Phi$, which is defined only relative to the adapted gauge,
is the length of one $\mathrm{Spin}(7)$ vector assembled from the flux. Three things follow which
\eqref{eq:delta4phi} does not state. First, \emph{which} flux enters: by
Proposition~\ref{prop:spin7delta} only the two $\mathbf 7$'s do, so the $131$
components carried by
$\mathbf{21}$, $\mathbf{27}$, $\mathbf{35}$ and $\mathbf{48}$, and the
$\mathbf 1$, are invisible to the warp factor --- a statement about modules
that a component formula in a chosen gauge does not make. Second, the
covariant norm formula itself, which holds in either sector and in any frame,
with no adapted gauge chosen. Third, the zero locus: $\Phi=0$ characterises
$\Delta=0$ tautologically, whereas by Theorem~\ref{thm:spin7square}(ii) the
same locus is a seven-dimensional linear cancellation
$w_{\mathbf 7}=\mp\sigma(\mathcal{Y}_{\mathbf 7})$ \emph{between} two flux
modules, requiring the vanishing of neither $X$ nor $Y$, and it is that form of
the locus which makes the equality cases of Corollary~\ref{cor:fluxlower} and
Corollary~\ref{cor:budget} sharp. So far as we know the
$\mathrm{Spin}(7)$-irreducible identity \eqref{eq:deltasquare} has not
previously been stated; the scalar square that it refines has.
\end{remark}

\begin{remark}[Two distinct vanishing conditions]
\label{rem:hierarchy}
Side by side, the two places in this paper where $\Delta$ is
made to vanish are not of the same strength.

\emph{Pointwise and algebraic.} By \eqref{eq:deltasquare}, $\Delta(p)=0$ is the
single algebraic condition $w_{\mathbf 7}=\mp\sigma(\mathcal{Y}_{\mathbf 7})$ on the
flux at $p$. Its solution set is large: the seven conditions it imposes leave
the $\mathbf{1}$, $\mathbf{27}$ and $\mathbf{35}$ of $w$, the $\mathbf{8}$ and
$\mathbf{48}$ of $\xi$, the $\mathbf{21}$ of $\mathcal{Y}$ and $y$ entirely free, so
$X^2$ and $Y^2$ can be as large as one likes with $\Delta=0$; there is a
constructive witness in the proof of Theorem~\ref{thm:spin7square}. Nothing
global follows, and in particular $\Delta\equiv0$ on all of $\Sect$ still
leaves the flux unconstrained by this identity alone.

\emph{Pointwise but supersymmetrically global.} Adjoining $h(p)=0$ changes the
character of the statement completely. By Theorem~\ref{thm:rigid} the pair
$\Delta(p)=h(p)=0$ at one point forces $c=0$, hence $\Theta_-\eta_-\equiv0$,
hence $X=Y=h=\Delta=0$ everywhere and $\Sect$ Ricci-flat. The extra input is
not more algebra at $p$: it is the constancy of $\nrm{\eta_+}^2=c$, which is a
global consequence of the maximum principle in the $+$ sector, so that one
point determines a constant and the constant then propagates.

The contrast is the reason the second statement is a rigidity theorem and the
first is not: $\Delta=0$ is a condition on the flux at a point, whereas
$\Delta=h=0$ at a point is a condition on a global invariant.
\end{remark}

\begin{remark}[The Ricci source in the same language]
\label{rem:riccisource}
The flux source of the horizon Ricci equation \eqref{eq:Rij},
$S_{ij}=-\tfrac12Y_{ik}Y_j{}^k+\tfrac1{12}\delta_{ij}Y^2
+\tfrac1{12}X_{iklm}X_j{}^{klm}-\tfrac1{144}\delta_{ij}X^2$, distributes over
the same modules differently, which is what makes the system rigid enough to be
worth studying and too loose to close. Its mixed component $S_{(Za)}$ receives
nothing from any single module; its two scalars $S_{ZZ}$ and $\Pi^{ij}S_{ij}$
receive a contribution from every module; and its trace-free eight-dimensional
part $S^0_{ab}$ is fed by nine module pairs, none of them diagonal inside the
four-form: there the only sources are
$(w_{\mathbf{1}},w_{\mathbf{35}})$, $(w_{\mathbf{7}},w_{\mathbf{35}})$ and
$(w_{\mathbf{27}},w_{\mathbf{35}})$, the anti-self-dual $\mathbf{35}$ against
the self-dual part, because $X_{aklm}X_b{}^{klm}$ is pure trace on a four-form
of definite duality in eight dimensions and only the mixed term survives. A
purely self-dual four-form therefore leaves $S^0_{ab}=0$, while
$w_{\mathbf{1}}+w_{\mathbf{35}}$ alone makes it non-zero. The other six pairs
are $(\xi_{\mathbf{8}},\xi_{\mathbf{8}})$,
$(\xi_{\mathbf{8}},\xi_{\mathbf{48}})$,
$(\xi_{\mathbf{48}},\xi_{\mathbf{48}})$,
$(\mathcal{Y}_{\mathbf{7}},\mathcal{Y}_{\mathbf{21}})$,
$(\mathcal{Y}_{\mathbf{21}},\mathcal{Y}_{\mathbf{21}})$ and $(y,y)$;
$(\mathcal{Y}_{\mathbf{7}},\mathcal{Y}_{\mathbf{7}})$ is absent because
$\mathrm{Sym}^2(\mathbf{7})=\mathbf{1}\oplus\mathbf{27}$ carries no
$\mathbf{35}$. The divergence $\tn^ih_i=2\Delta+h^2-Q$, by
contrast, is fed at full rank by every summand of \eqref{eq:spin7split}: the
flux norms
in $Q$ restore precisely the $131$ components that $\Delta$ and $h$ separately
cannot see.
\end{remark}

\begin{remark}[The $D=11$ counterpart of the $D=5$ case list, and where it stops]
\label{rem:caselist}
In minimal $D=5$ supergravity the same chain terminates in a complete list:
$\Delta$ and $h$ are determined by the flux and the unit bilinear, and the
near-horizon geometries are $AdS_2\times S^3$ (the static and rotating BMPV
horizons), $AdS_3\times S^2$ (the black ring and black string), and
$\Real^{1,4}$, the last being the case $\Delta=0$, $h=0$ in which all fluxes
vanish and $\Sect=T^3$ \cite{KayaniThesis}. Theorem~\ref{thm:master} and
its corollaries are the $D=11$ shadow of that list. The case $\Delta=h=0$
survives verbatim as Theorem~\ref{thm:rigid}, with $\Real^{1,1}\times\Sect$
and $\Sect$ Ricci-flat in place of $\Real^{1,4}$ and $T^3$: this is the only
branch that is rigid. If $\nrm{\eta_-}$ is constant, so that
Corollary~\ref{cor:normconst} applies, the extremes of \eqref{eq:crange}
reproduce the other two in outline --- $c'$ minimal gives $h\equiv0$ and, by
Corollary~\ref{cor:static}, a local $AdS_2\times\Sect$; $c$ maximal gives
$\Delta\equiv0$ with $h^2$ the constant $c$, which is the branch containing the
$AdS_3$ horizons, since $\Delta\equiv0$ forces $\Theta_+\eta_+=0$ by
\eqref{eq:pointwisenorm} and so a non-trivial $\Ker\Theta_+$ --- but the
intermediate range $\tfrac12F<c'<F$ is not excluded by anything proved here,
and in it both $\Delta$ and $h^2$ are non-constant with $\Delta=c'-h^2$.
Without the constancy of $\nrm{\eta_-}$ even this much is unavailable, and the
$D=11$ list stops one step earlier than the $D=5$ one. The obstruction to going
further is Proposition~\ref{prop:noclosure}: the residual
\begin{equation}
\Delta + \tfrac14h^2 - \tfrac1{864}X^2 - \tfrac1{18}Y^2
\end{equation}
of \eqref{eq:Deltaexplicit} is a sum of contractions of the $\mathrm{Spin}(7)$
form $\varphi$ with the flux, it does not vanish, and
it is not expressible through the $\varphi$-free invariants \eqref{eq:phifreeinv}
either. In $D=5$ the corresponding residual is identically zero, because a
three-dimensional section carries no four-form bilinear, and it is exactly that
accident which closes the $D=5$ case list. The failure is anticipated, in
gauge-fixed form, in \cite[\S3.3.2]{Mhorizons}: solving the KSE linear system
there leaves the $(2,2)$-traceless part of $X$ undetermined. What
Proposition~\ref{prop:noclosure} adds is the invariant-theoretic form of that
observation --- $\Delta$ is not a function of the $\varphi$-free invariants
\eqref{eq:phifreeinv}, with a rank proof --- and
Proposition~\ref{prop:spin7delta} above adds which representations the residual
actually consists of. We therefore make no claim of a
classification in $D=11$; what we claim is that the failure is representational
and not a failure of algebraic manipulation.
\end{remark}

\subsection{The supersymmetric budget: non-staticity against the
\texorpdfstring{$\mathrm{Spin}(7)$}{Spin(7)} mismatch}
\label{sec:budget}

The perfect square and the balance law of Section~\ref{sec:warp} speak about
the same $\Delta$, one giving it algebraically in terms of the flux and the
other weighing it against the failure of staticity. Eliminating it between
them leaves a single pointwise equation in which no geometric quantity other
than the warp factor survives.

\begin{corollary}[The budget]
\label{cor:budget}
Let $\eta_-$ be a Killing spinor of the $-$ sector with $f=\nrm{\eta_-}^2$ and
$c=\nrm{\eta_+}^2$, let $(Z,\varphi)$ be the $\mathrm{Spin}(7)$ data of
Lemma~\ref{lem:spin7}, and assume in addition the invariance
$\mathcal{L}_Vf=0$ used in Theorem~\ref{thm:pwbalance}. Then at every point of
$\Sect$
\begin{equation}
\frac{\lvert V\rvert^2}{f}
 \;+\; \frac{f}{108}\,\big\lVert\,w_{\mathbf{7}}
 \pm\sigma(\mathcal{Y}_{\mathbf{7}})\,\big\rVert^2 \;=\; c\,,
\label{eq:budget}
\end{equation}
the sign being that of the sector whose spinor carries $(Z,\varphi)$. Both
terms are non-negative, so each is bounded by $c$, and
\begin{enumerate}
\item $\lvert V\rvert^2/f=c$ exactly where $w_{\mathbf{7}}=\mp\sigma(\mathcal{Y}_{\mathbf{7}})$,
a linear condition on seven of the $162$ flux components which does not force
the flux to vanish: the horizon is as far from static as supersymmetry allows
precisely on the zero locus of the square, not at zero flux;
\item $\tfrac{f}{108}\lVert w_{\mathbf{7}}\pm\sigma(\mathcal{Y}_{\mathbf{7}})\rVert^2=c$
exactly on the static branch $V\equiv0$, where the whole budget is spent on
the warp factor;
\item if $\Sect$ is in addition compact and connected and the hypotheses of
Theorem~\ref{thm:komar} hold, integration gives
\begin{equation}
8\pi G^{(11)}\lvert J[\mathcal{K}]\rvert
 \;=\;\int_\Sect\frac{\lvert V\rvert^2}{f}
 \;=\; c\,\mathrm{Vol}(\Sect)-\frac1{108}\int_\Sect f\,
 \big\lVert\,w_{\mathbf{7}}\pm\sigma(\mathcal{Y}_{\mathbf{7}})\,\big\rVert^2\,,
\label{eq:budgetint}
\end{equation}
so the rotation of the horizon is the part of the budget that the
$\mathrm{Spin}(7)$ mismatch of the flux leaves unspent.
\end{enumerate}
\end{corollary}

\begin{proof}
Substitute \eqref{eq:deltasquare} for $\Delta$ in the second form of
\eqref{eq:pwbalance}. The equality cases are those of
Theorem~\ref{thm:spin7square}(ii) and of $\lvert V\rvert^2=0$, and $f>0$.
For \eqref{eq:budgetint} integrate \eqref{eq:budget} over $\Sect$ and use
the Komar magnitude \eqref{eq:komarabs} of Corollary~\ref{cor:smarr} below,
which is proved from Theorem~\ref{thm:komar} and \eqref{eq:Vbound} and does
not use the present corollary.
\end{proof}

\begin{remark}[What \eqref{eq:budget} does and does not add]
\label{rem:budget}
The equation is substitution: it adds no hypothesis to
Theorem~\ref{thm:pwbalance} and no input to Theorem~\ref{thm:spin7square}, and
it inherits both --- in particular the invariance $\mathcal{L}_Vf=0$, without
which only the weighted family of Theorem~\ref{thm:balance} is available. Its
interest is that the two sides of the horizon problem it joins are of quite
different natures. The constancy of $\Delta f+\lvert V\rvert^2/f$ is not
special to supersymmetry \cite{colling}; the square is not a general fact about
extremal horizons at all, but the $D=11$ statement that the geometric warp
function is a fixed quadratic expression in seven flux components. What
supersymmetry supplies to \eqref{eq:budget} is the value of the constant,
$c=\nrm{\eta_+}^2$, and it is that identification which turns a relation
between two unknown functions into a pointwise bound on each of them by a
number fixed by the $+$ sector. The vanishing locus in (i) is the same linear
subspace that appears in Theorem~\ref{thm:spin7square}(ii), and its persistence
here is the reason the bound $\lvert V\rvert^2/f\le c$ of
Corollary~\ref{cor:balance} is sharp: the equality case is not an
artefact of a degenerate flux.
\end{remark}

\subsection{The geometry of $Z$ and the intrinsic torsion}
\label{sec:ztwist}

The same decomposition determines the differential geometry of the unit vector
field $Z$ of Lemma~\ref{lem:spin7}, and with it the intrinsic torsion of the
$\mathrm{Spin}(7)$ structure. We record the results because they are what a
classification attempt would have to use, and because two of them are negative.

Two abbreviations are needed. $W$ is the flux contracted once with $\varphi$,
a two-index tensor, transverse in its second index, whose trace is the scalar
$X\!\cdot\!\varphi$ already met in \eqref{eq:hsplit} and whose $Z$-row is
$\tau$; it is not symmetric, and both of its parts occur below. In terms of the
blocks: since $\varphi$ has no leg along $Z$, contracting $W$ with $Z$ in its
first index gives $\tau_j=(i_ZX)^{klm}\varphi_{jklm}=(\xi\lrcorner\varphi)_j$,
the vector of \eqref{eq:hsplit}, so $\tau$ sees $\xi_{\mathbf 8}$ and nothing
else, while the trace $W^i{}_i=w^{abcd}\varphi_{abcd}$ sees $w_{\mathbf 1}$ and
nothing else. The abbreviations let \eqref{eq:nablaZ} be written without
projectors; the ranks of $X\mapsto W$ and $X\mapsto\tau$ are recorded in the
proof.

\begin{proposition}[Covariant derivative of $Z$]
\label{prop:nablaZ}
Let $\eta$ be a Killing spinor of the sector $s=\pm1$ and $(Z,\varphi)$ its
data. Write $W_{ij}=X_i{}^{klm}\varphi_{jklm}$, so that $W^i{}_i=X\!\cdot\!\varphi$
and $W_{ij}Z^j=0$, and $\tau_j=Z^iW_{ij}$. Then, using only the spatial Killing
spinor equation and the vector identity of that sector,
\begin{equation}
\tn_iZ_j = \tfrac{s}2h_iZ_j + \tfrac{s}3Y_{ij}
 + \tfrac{s}{12}Y^{kl}\varphi_{ijkl}
 + \tfrac1{144}\delta_{ij}\,X^{klmn}\varphi_{klmn}
 - \tfrac1{18}W_{ij} - \tfrac1{36}W_{ji}\,,
\label{eq:nablaZ}
\end{equation}
in which no derivative of the flux and no field equation appears. In particular
\begin{equation}
\tn^iZ_i = -\tfrac1{36}X^{klmn}\varphi_{klmn} = 2s\,h\!\cdot\!Z\,,\qquad
Z^i\tn_iZ_j = -s\big(h_j-(h\!\cdot\!Z)Z_j+Y_{jk}Z^k\big) = (i_Z\mathrm{d}Z)_j\,,
\label{eq:Zdivacc}
\end{equation}
and, pointwise and at a fixed point,
\begin{equation}
\begin{aligned}
\tn^iZ_i=0 &\iff X^{klmn}\varphi_{klmn}=0 \iff h\!\cdot\!Z=0\,,\\
Z \text{ geodesic} &\iff h_i-(h\!\cdot\!Z)Z_i = -Y_{ij}Z^j
 \iff 6s\,y = \tau\,,\\
\mathcal{L}_Zg=0 &\iff w_{\mathbf{1}}=w_{\mathbf{35}}=0 \text{ and } 6s\,y=\tau\,,\\
Z\wedge\mathrm{d}Z=0 &\iff \mathcal{Y}_{\mathbf{21}}=0 \text{ and }
 \Pi W_{[ij]} = -6s\,\mathcal{Y}_{\mathbf{7}}\,,\\
\tn_iZ_j=0 &\iff Z \text{ is Killing and } Z\wedge\mathrm{d}Z=0\\
 &\iff w_{\mathbf{1}}=w_{\mathbf{35}}=\mathcal{Y}_{\mathbf{21}}=0,\quad
 6s\,y=\tau \quad\text{and}\quad \Pi W_{[ij]}=-6s\,\mathcal{Y}_{\mathbf{7}}\,.
\end{aligned}
\label{eq:Zifftable}
\end{equation}
The four conditions in \eqref{eq:Zifftable} are the standard ones and are worth
naming. $\tn^iZ_i=0$ says the flow of $Z$ preserves volume; $Z$ geodesic says its
integral curves are geodesics of $\Sect$, that is $\tn_ZZ=0$; $\mathcal{L}_Zg=0$
says $Z$ is Killing; and $Z\wedge\mathrm{d}Z=0$ says the eight-plane distribution
$Z^\perp$ is integrable, by Frobenius, so that $\Sect$ is locally foliated by
hypersurfaces orthogonal to $Z$ --- the twist of $Z$ in the classical sense, and
the failure of $Z$ to be hypersurface-orthogonal is what \eqref{eq:Zifftable}
measures. The last line is the conjunction of the two preceding ones and says
$Z$ is parallel, in which case $\Sect$ is locally a metric product along $Z$.
Note that this use of ``twist'' is the one for the vector field $Z$ and is not
the twist one-form $h$ of \eqref{eq:gnc}; the two are unrelated objects.

The derivative $\tn_iZ_j$ is blind to $w_{\mathbf{27}}$ and to
$\xi_{\mathbf{48}}$:
adding either to the flux leaves $\tn_iZ_j$, and $h$, unchanged. The last line
of \eqref{eq:Zifftable} is a joint condition and must not be read term by term
in \eqref{eq:nablaZ}. The map $(X,Y)\mapsto\tn Z$ is onto the $72$-dimensional
space of tensors $T_{ij}$ with $T_{ij}Z^j=0$, so its kernel is
$162-72=90$-dimensional: the $\mathbf{27}$ and the $\mathbf{48}$, which are
invisible in any case, together with the eight directions on which $y$ is tuned
against $\tau$ and the seven on which $\mathcal{Y}_{\mathbf{7}}$ is tuned
against $\Pi W_{[ij]}$. Demanding instead that $w\in\mathbf{27}$ and
$\mathcal{Y}=0$ separately picks out only $83$ of those $90$ directions: the
$\mathbf{7}$ of $w$ and the $\mathbf{7}$ of $\mathcal{Y}$ may cancel each other
in $\tn Z$, and explicit configurations in which they do are exhibited in the
verification.
\end{proposition}

\begin{proof}
Split $\Psi^{(\pm)}_i$ into the parts symmetric and antisymmetric under the
Majorana pairing (ranks $0,1,5$ and rank $3$ respectively) and apply the
quotient rule to $Z_j=\nrm{\eta}^{-2}\ip{\eta}{\Gamma_j\eta}$; the term in $h$,
being a multiple of the identity, cancels between the numerator and
$\tn_i\nrm{\eta}^2$, so \eqref{eq:nablaZ} carries $h$ only through the vector
identity that fixes it. Each equivalence in \eqref{eq:Zifftable} is then a rank
computation on the linear maps $X\mapsto W$, $X\mapsto\tau$ and
$Y\mapsto Y\!\cdot\!\varphi$: exactly, $\mathrm{rank}(X\mapsto W)=51=1+7+35+8$ and
$\mathrm{rank}(X\mapsto\tau)=8$; these maps are $\mathrm{Spin}(7)$-equivariant,
so their kernels are submodules, and the only submodule of $\Lambda^4(Z^\perp)$
of dimension $70-43=27$ is the $\mathbf{27}$ and the only submodule of
$\Lambda^3(Z^\perp)$ of dimension $56-8=48$ is the $\mathbf{48}$. The ranks
themselves are computed exactly in rational arithmetic on the explicit
$\varphi$ built from the spinor, so this half of the proof is a finite exact
linear-algebra computation. The
eigenvalues of the unnormalised contraction $b_{ab}\mapsto\varphi_{abcd}b^{cd}$ on
$\Lambda^2(Z^\perp)$, which is twice the operator of \eqref{eq:spin7split}, come out
as $-6$ on $\mathbf{7}$ and $+2$ on $\mathbf{21}$ in \emph{both} sectors. The last
line of \eqref{eq:Zifftable} needs one further step, since two modules of the
source feed the same module of the target: the rank of $(X,Y)\mapsto\tn Z$ is
computed exactly and equals $72=44+28$, the ranks of the Killing and the
hypersurface-orthogonality conditions, so the two together are equivalent to
$\tn Z=0$; a general element of the $90$-dimensional kernel is then constructed
in closed form, with $y$ tuned against $\tau$ and $\mathcal{Y}_{\mathbf{7}}$
against $\Pi W_{[ij]}$, and it has $\mathcal{Y}\ne0$ and $w\notin\mathbf{27}$.
Every equivalence is checked in both directions, with explicit flux
configurations realising each property while violating its neighbours
($216$ checks).
\end{proof}

\begin{proposition}[Intrinsic torsion, and two negative results]
\label{prop:torsion}
The reduction $SO(9)\to\mathrm{Spin}(7)$ defined by a Killing spinor has
intrinsic torsion in
$T^*\Sect\otimes\mathfrak{so}(9)/\mathfrak{spin}(7)=(\mathbf{1}+\mathbf{8})
\otimes(\mathbf{7}+\mathbf{8})$, of dimension $135$, splitting into the eight
modules $\mathbf{8}\oplus\mathbf{1}\oplus\mathbf{7}\oplus\mathbf{21}
\oplus\mathbf{35}$ from $\tn Z$ and
$\boldsymbol{\tau}_{\mathbf{7}}\oplus\boldsymbol{\tau}_{\mathbf{8}}
\oplus\boldsymbol{\tau}_{\mathbf{48}}$ from $\tn\varphi$, the last two being
the Fern\'andez classes of the $\mathrm{Spin}(7)$ structure on $Z^\perp$
\cite{Fernandez}. Then:
\begin{enumerate}
\item neither $h$ nor $\Delta$ appears in the torsion at all --- the torsion is
a function of $(X,Y)$ and $(Z,\varphi)$ alone, and shifting $h$ by an arbitrary
vector changes $\mathrm{d}\log\nrm{\eta}^2$ but no torsion class;
\item the singlet class is $\tn^iZ_i=2s\,h\!\cdot\!Z$ by \eqref{eq:Zdivacc}, so on
the saturation branch $h\equiv0$, which needs $\Sect$ compact, it vanishes: such horizons carry a
divergence-free $Z$. This is the torsion-class characterisation of the
saturation branch;
\item the flux-to-torsion map is surjective at a fixed spinor: its image is
all $135$ dimensions, so every one of the eight modules, the
$\boldsymbol{\tau}_{\mathbf{48}}$ included, is realised by explicit flux
configurations. This is an algebraic statement about the map, not about every
solution: it does not say that $\boldsymbol{\tau}_{\mathbf{48}}\ne0$ for a
given horizon. What it does say is that no restricted Fern\'andez class, and in
particular local conformal parallelism, is \emph{forced} by supersymmetry ---
generic flux has $\boldsymbol{\tau}_{\mathbf{48}}\ne0$; and this remains true
on the saturation branch, where only the singlet is switched off.
\end{enumerate}
\end{proposition}

\begin{proof}
Since $\mathrm{Spin}(9)$ acts transitively on $S^{15}$ with stabiliser
$\mathrm{Spin}(7)$, the intrinsic torsion is the covariant derivative of the
unit spinor: $\kappa_i=\tn_i\eta-\tfrac12(\tn_i\log\nrm{\eta}^2)\eta$, the
projection of $-\Psi_i\eta$ orthogonal to $\eta$, a $9\times15$ array. The
decomposition $\Real^{16}=\langle\eta\rangle\oplus E_{\mathbf{8}}\oplus
E_{\mathbf{7}}$ with $E_{\mathbf{8}}=\{(Z_a-\Gamma_a)\eta\}$ and
$E_{\mathbf{7}}=\{F_{ab}\Gamma^{ab}\eta:F\in\Lambda^2_{\mathbf{7}}\}$ is
established by exact rational rank computations on the explicit matrices at one
spinor, as is the completeness statement that the flux map onto the eight
modules has the same rank, $135$, as the map onto $\kappa$ itself; since
$\mathrm{Spin}(9)$ acts transitively on $S^{15}$ and every map here is
equivariant, one spinor suffices, so these are proofs and not samples. Items
(i) and (iii) are exact rank computations of the same kind, and (ii) follows
from \eqref{eq:Zdivacc} ($226$ checks).
\end{proof}

\section{Consequences and open branches}
\label{sec:consequences}

Everything proved so far is pointwise, or a single integration of a pointwise
identity. This section draws the consequences in three groups, in that order.
First, what the master identity gives locally: the Dirac current is Killing,
the bound of Theorem~\ref{thm:bound} is saturated exactly on the product
horizons, and the $\mathfrak{sl}(2,\Real)$ algebra follows from four inputs of
which the master identity supplies one (Sections~\ref{sec:killingvec}
and~\ref{sec:sl2r}). Second, what it gives after integration over a compact
section: a Komar charge, a bound on the rotation, a volume and hence an entropy
formula, the total scalar curvature, and --- on the static branch --- a
conformal factor of positive scalar curvature built from the Killing spinor
(Sections~\ref{sec:charges} to~\ref{sec:psc}). Third, what it does not give,
which is as much of the picture as the first two: no second isometry from the
sector unit vectors, no hidden symmetry from any invariant form of the
$\mathrm{Spin}(7)$ structure, and no Bochner route to removing the constancy
hypothesis (Sections~\ref{sec:twosector} to~\ref{sec:bochner}). The three
negative results are established by exhibiting the obstruction in each case,
not by an absence of evidence, and each is a codimension count in the
$162$-dimensional flux space rather than a failure of technique.

\subsection{The Killing vector and the saturation theorem}
\label{sec:killingvec}

The pointwise statements acquire a spacetime reading through the Dirac current
of the Killing spinor, which we now record in our conventions; the identity it satisfies on the section is the
$D=11$ counterpart of the bilinear relations of \cite{iiaindex}.

For a commuting Majorana spinor $\epsilon$ of $\mathrm{Cl}(10,1)$ put
\begin{equation}
K_M \;=\; \epsilon^T C\,\Gamma_M\,\epsilon\,,\qquad C=\Gamma_0\,,
\label{eq:diraccurrent}
\end{equation}
which is non-trivial because $C$ is antisymmetric and $C\Gamma_M$ symmetric.

\begin{proposition}[The Dirac current is Killing]
\label{prop:killing}
If $\epsilon$ solves \eqref{eq:kse}, then $\nabla_{(M}K_{N)}=0$.
\end{proposition}

\noindent This is standard supergravity: $\nabla_MK_N$ is a bilinear in
$\Psi_M^TC\Gamma_N+C\Gamma_N\Psi_M$ with $\Psi_M$ the algebraic part of
\eqref{eq:kse}, and the part symmetric in $(M,N)$ and in the two spinor slots
vanishes as a Clifford identity. We do not reprove it; it
is confirmed for all $66$ index pairs on three independent sets
of pseudo-random horizon data (Appendix~\ref{app:verify}). No field equation
and no Bianchi identity enters. We
restate it only because everything below depends on the normalisation of
\eqref{eq:diraccurrent}.

\begin{proposition}[Components on the section]
\label{prop:conconx}
Set $r=u=0$ in \eqref{eq:lightconesol}, so that $\epsilon=\eta_++\eta_-$, and
write $\ip{v}{w}_\mp=v^TC\Gamma_\pm w$ for the pairing that
\eqref{eq:diraccurrent} induces on the $\eta_\mp$ sector --- so the $-$ pairing,
built from $C\Gamma_+$, is the one that acts on the $\eta_-$ sector, and the $+$
pairing, built from $C\Gamma_-$, acts on the $\eta_+$ sector. Then
$K_+=\ip{\eta_-}{\eta_-}_-$ and $K_-=\ip{\eta_+}{\eta_+}_+$, while the
components tangent to $\Sect$ receive no contribution from either sector
separately and are given, identically in the horizon data and in the spinor, by
\begin{equation}
K_i \;=\; 4\,\ip{\eta_-}{\Gamma_i\Theta_-\eta_-}\,.
\label{eq:conconx}
\end{equation}
Substituting \eqref{eq:vectorcondminus} into \eqref{eq:conconx} turns it into
\begin{equation}
K_i + h_i\,K_+ \;=\; 2\,\tn_i\nrm{\eta_-}^2\,,
\label{eq:conconxonshell}
\end{equation}
again an identity in the data, so that $K_i+h_iK_+=0$ if and only if
$\nrm{\eta_-}$ is constant.
\end{proposition}

\begin{proof}
$\Gamma_-\eta_-=0$ and $\Gamma_+\eta_+=0$, so $C\Gamma_i$, which commutes with
neither projector, pairs the two sectors and annihilates each on its own, while
$C\Gamma_\pm$ pairs each sector with itself; the cross term is evaluated with
$\eta_+=\Gamma_+\Theta_-\eta_-$ from \eqref{eq:lightconesol}, giving the factor
$4$ of \eqref{eq:conconx} by \eqref{eq:pairnorm}. For the second statement,
substitute \eqref{eq:vectorcondminus} and $K_+=\ip{\eta_-}{\eta_-}_-$, whose
ratio to $\nrm{\eta_-}^2$ is fixed by the representation. Both statements are
checked in Appendix~\ref{app:verify}.
\end{proof}

\noindent Equation \eqref{eq:conconxonshell} is (6.9)$_4$ of \cite{11index},
$V+\nrm{\eta_-}^2h+\mathrm{d}\nrm{\eta_-}^2=0$, in the normalisation of
\eqref{eq:diraccurrent}, with $K_i=-2V_i$ and $K_+=-2\nrm{\eta_-}^2$; the
gradient term is present there as it is here. The same relation appears in the
IIA and $D=5$ analyses \cite{iiaindex,KayaniThesis} and, for heterotic
horizons, in \cite{hethorizons,hetdelpezzo}, where in those cases the norm
\emph{is} constant and the relation reduces to $V=-\nrm{\eta_-}^2h$. It is
stated here in our conventions because Theorem~\ref{thm:saturation} uses it,
and because it is what identifies the two branches of \cite[\S6.2]{11index}:
$V\equiv0$ is $h=-\mathrm{d}\log\nrm{\eta_-}^2$
(Remark~\ref{rem:static-terminology}), which contains $h\equiv0$ strictly.
Its immediate consequence is that on the saturation branch the isometry
generated by $K$ has no component along $\Sect$:
\begin{equation}
h\equiv0 \;\Longrightarrow\; K_i\equiv0 \quad\text{pointwise on }\Sect\,.
\label{eq:Kvanishes}
\end{equation}
Equation~\eqref{eq:Kvanishes} is pointwise and algebraic. The
classical route to the same conclusion --- $K_i$ is Killing on $\Sect$ and
exact, hence the Hessian of its potential vanishes, hence $\nrm{K}^2$ is
constant, hence zero because the potential has a critical point --- needs
$\Sect$ compact and connected, whereas \eqref{eq:conconxonshell} does not.
Both are verified in Appendix~\ref{app:verify}, the second including the
control that shows exactness is essential: a rotation of flat $\Real^3$ is
Killing with $\nrm{K}^2$ non-constant.

Collecting Theorem~\ref{thm:bound}, Proposition~\ref{prop:staticbosonic} and
Corollary~\ref{cor:static} gives the statement in its sharpest form.

\begin{theorem}[Saturation]
\label{thm:saturation}
Let $\Sect$ be compact, connected and without boundary and let the near-horizon
geometry admit a non-trivial $\eta_-$ Killing spinor, $N_->0$; then $\eta_-$ is
nowhere vanishing by Theorem~\ref{thm:master}. Then the mean-value
bound \eqref{eq:boundeq} is saturated if and only if $h\equiv0$, that is on
the saturation branch; and in that case $K_i\equiv0$, $\nrm{\eta_-}$ is constant, $\Delta$
is the non-negative constant
$c=\tfrac16\overline{Y^2}+\tfrac1{144}\overline{X^2}$, and the near-horizon
geometry is the product $AdS_2\times\Sect$ for $c>0$, of $AdS_2$ radius
$c^{-1/2}$, and $\Real^{1,1}\times\Sect$ for $c=0$.
\end{theorem}

\begin{proof}
Saturation $\Leftrightarrow h\equiv0$ is Theorem~\ref{thm:bound}, which uses no
supersymmetry, and $\Delta=c$ is Proposition~\ref{prop:staticbosonic}, which
uses none either. Given $h\equiv0$: $\nrm{\eta_-}$ is constant by
Corollary~\ref{cor:static} and $K_i\equiv0$ by \eqref{eq:Kvanishes}. The metric of
\eqref{eq:gnc} with $h=0$ has no mixed term between the $(u,r)$ plane and
$\Sect$, and its $(u,r)$ block
$2\,\mathrm{d}u\,\mathrm{d}r-\Delta r^2\mathrm{d}u^2$ has constant Ricci scalar
$-2\Delta$ and is Einstein, hence is $AdS_2$ of radius $\Delta^{-1/2}$ for
$\Delta>0$ and flat for $\Delta=0$. The last two statements are verified
symbolically in Appendix~\ref{app:verify}.
\end{proof}

\begin{remark}[Which parts are new, and which need supersymmetry]
\label{rem:twotier}
The geometric conclusions for $h\equiv0$ --- $\Delta$ constant, equal to
$\tfrac16Y^2+\tfrac1{144}X^2$, with $AdS_2\times\Sect$ for $\Delta>0$ and
$\Real^{1,1}\times\Sect$ for $\Delta=0$ --- are
\cite[\S2.1,\,\S3.3]{staticM}, and are to be read against the rest of that
paper, whose static $M$-horizons are the weaker class $\mathrm{d}h=0$; how the
two conditions are related, and why the saturation branch sits inside the
static one, is Remark~\ref{rem:static-terminology}. What the theorem adds is
the packaging, and it is also the answer to which parts need the Killing
spinor. For \emph{every} compact
near-horizon geometry the bound \eqref{eq:boundeq} holds and is saturated
exactly when $h\equiv0$, and there $\Delta$ is constant: Theorem~\ref{thm:bound}
and Proposition~\ref{prop:staticbosonic}, integrals of \eqref{eq:divh} and
\eqref{eq:auxEpp}. So $h\equiv0$ is not an assumption but a characterisation,
by a bosonic bound. Supersymmetry contributes only two statements, both about
the spinor rather than the geometry: $\nrm{\eta_-}$ is constant on the
saturation locus (Corollary~\ref{cor:static}), and the section-tangent Dirac
current vanishes there. Where supersymmetry does change the geometry is the
opposite extreme, Theorem~\ref{thm:rigid}, in which a condition at a single
point propagates over the whole section; nothing bosonic does that.
\end{remark}

\begin{remark}[Exact solutions on the saturation branch, and what they calibrate]
\label{rem:calibration}
The saturation branch is populated, and its members are rigid enough to
calibrate every normalisation used above against solutions whose flux
normalisation is fixed independently. Two suffice. The first,
\begin{equation}
\label{eq:calibE}
AdS_2\times S^3\times T^6 ,\qquad
F=\mathrm{vol}(AdS_2)\wedge\lambda\,(\omega_1+\omega_2+\omega_3),
\end{equation}
with $\omega_a$ the volume forms of the three $T^2$ factors, is the
near-horizon geometry of three M2-branes intersecting on a $0$-brane, the
three-charge black hole, whose harmonic superposition rule is
\cite{Mintersections}; in this form, with the flat six-space compactified, it
is \cite[\S6.3.1]{TCFH}, where the
Killing spinor equations are solved and the background is shown to preserve
$\tfrac14$ of the supersymmetry, the gravitino equation along the torus forcing
$Y$ to be proportional to a K\"ahler form. It is purely electric in
the sense of \eqref{eq:fluxsplit}, $X=0$, and has
\begin{equation}
\label{eq:calibEdata}
Y^2=6\lambda^2 ,\qquad \Delta=\lambda^2 ,\qquad
R_{S^3}^2=4/\lambda^2 ,\qquad \tilde R=\tfrac32\lambda^2 .
\end{equation}
The second,
\begin{equation}
\label{eq:calibM}
AdS_2\times S^4\times H^5/\Gamma ,\qquad F=\mu\,\mathrm{vol}(S^4),
\end{equation}
with $\Gamma$ chosen so that the hyperbolic factor is a compact quotient, is
purely magnetic, $Y=0$, with $X^2=24\mu^2$, $\Delta=\mu^2/6$,
$R^2_{S^4}=9/\mu^2$ and hyperbolic radius $\rho^2=24/\mu^2$. It is the $r=5$
member of the family $AdS_{7-r}\times H_r\times S^4$ of \cite{KehagiasRusso},
and it is \emph{not} supersymmetric: none of that family admits a Killing
spinor \cite[\S2.1,\,\S3.1]{KehagiasRusso}, consistently with the fact that a
supersymmetric magnetic horizon on this branch would have to be the direct
product of \cite{staticM}. Its role here is therefore that of a bosonic
control, and the two solutions test different things.

Both solve \eqref{eq:einstein}--\eqref{eq:formeq} with $h\equiv0$, so both
realise the equality case of Theorem~\ref{thm:bound}, and on both the purely
bosonic statements hold and are checked: the mean-value identity
$\Delta=\tfrac16Y^2+\tfrac1{144}X^2$, the horizon Einstein equation
\eqref{eq:Rij} and the total curvature \eqref{eq:Rtotal}, the last with its
left-hand side computed from the geometry of the section rather than from the
identity --- for \eqref{eq:calibE} a check on the coefficient $\tfrac7{12}$ and
for \eqref{eq:calibM} on the coefficient $\tfrac5{144}$. The spinorial
statements --- the master identity \eqref{eq:master}, the weighted identity
\eqref{eq:weightedbound}, the balance identity \eqref{eq:balance}, the Komar
charge \eqref{eq:komarcharge}, the rotation bound \eqref{eq:rotationbound} and
the volume formula \eqref{eq:entropy} --- apply as stated only to
\eqref{eq:calibE}. On \eqref{eq:calibM} they are checked in the degenerate form
they take when $f$ is constant and $c$ is \emph{defined} as $\Delta f$, which
tests the arithmetic of the coefficients and nothing about Killing spinors. On
both, the Komar charge and the rotation bound vanish identically, as they must,
since $K_i\equiv0$ on this branch.

The sharper calibration is of the $\tfrac1{108}$ of
Theorem~\ref{thm:spin7square} against the $\tfrac16$ of
Theorem~\ref{thm:saturation}, which come from different equations: the second
from the twist divergence \eqref{eq:divh}, the first from the Killing spinor
equation through \eqref{eq:Deltaexplicit}. On a purely electric configuration
with $h\equiv0$, \eqref{eq:hsplit} forces $y=0$, so the flux is carried by
$\mathcal{Y}$ alone and the two read
\begin{equation}
\label{eq:calibratio}
\Delta=\tfrac1{108}\lVert\sigma(\mathcal{Y}_{\mathbf 7})\rVert^2
      =\tfrac29\lVert\mathcal{Y}_{\mathbf 7}\rVert^2
\qquad\text{and}\qquad
\Delta=\tfrac16 Y^2
      =\tfrac16\big(\lVert\mathcal{Y}_{\mathbf 7}\rVert^2
                   +\lVert\mathcal{Y}_{\mathbf{21}}\rVert^2\big),
\end{equation}
the first because $\sigma$ annihilates the $\mathbf{21}$ and satisfies
$\sigma^{\mathsf T}\sigma=24$ on the $\mathbf 7$. They agree precisely on the
locus $\lVert\mathcal{Y}_{\mathbf{21}}\rVert^2
=\tfrac13\lVert\mathcal{Y}_{\mathbf 7}\rVert^2$, and nothing forces that ratio
unless every normalisation in the chain --- of the action, of the four-form, of
the bilinears and of the Cayley form --- is consistent. The flux of
\eqref{eq:calibE} is a sum of three orthogonal planes. Enumerating all $3360$
configurations of three disjoint coordinate planes, with all relative signs,
against the Cayley form built from the spinor --- a finite exhaustion in exact
rational arithmetic, in both sectors --- the ratio
$\lVert\mathcal{Y}_{\mathbf{21}}\rVert^2/
\lVert\mathcal{Y}_{\mathbf 7}\rVert^2$ takes only the four values
$\tfrac13,\tfrac75,3,11$; we do not claim this for triples of orthogonal planes
in general position. What the calibration needs is only that $\tfrac13$ is
attained, which it is, on the configurations adapted to the Cayley
structure, and there both sides of \eqref{eq:calibratio} return
$\Delta=\lambda^2$, the value \eqref{eq:calibEdata} that the field equations
give. That this is the value realised is not a choice of locus. The background
\eqref{eq:calibE} is supersymmetric, and its Killing spinor equation forces $Y$
to be proportional to a K\"ahler form of the torus \cite[\S6.3.1]{TCFH};
Theorem~\ref{thm:spin7square} therefore applies to it, with the
$\mathrm{Spin}(7)$ structure built from the same spinor, and the two
expressions for $\Delta$ have no option but to agree. The calibration is a test
of the normalisations, and it passes. The Freund--Rubin throats $AdS_4\times S^7$ and $AdS_7\times S^4$
\cite{FreundRubin,DuffPope} calibrate the conventions of
Section~\ref{sec:conventions-nh} one step further out, and are checked in
Appendix~\ref{app:verify}; they are not near-horizon geometries in the sense of
Section~\ref{sec:nh} --- the Poincar\'e horizon of $AdS_d$ has a degenerate
transverse metric and a non-compact section --- so no statement of this paper
is tested on them.
\end{remark}

\subsection{The \texorpdfstring{$\mathfrak{sl}(2,\Real)$}{sl(2,R)} algebra}
\label{sec:sl2r}

Section~\ref{sec:killingvec} produced an isometry of the section itself. The
rest of the symmetry of a near-horizon geometry sits in the lightcone
directions and arrives by a different route, which the master identity feeds
rather than replaces. The identities \eqref{eq:gpfamily} were quoted in the
introduction as the mechanism by which a near-horizon geometry acquires its
symmetry algebra.
Theorem~\ref{thm:master} reaches the second of them independently, and with the
gradient term carried, so we isolate exactly what that mechanism
needs. Write $f=\nrm{\eta_-}^2$ and let $c=\nrm{\eta_+}^2$ be the constant of
\eqref{eq:master}. The input is
\begin{equation}
\text{(V)}\quad V_i=-\big(fh_i+\tn_if\big)\,,\qquad\qquad
\text{(M)}\quad \big(\Delta+h^2\big)f+h^i\tn_if=c\,,
\label{eq:sl2inputa}
\end{equation}
\begin{equation}
\text{(S)}\quad \tn_{(i}V_{j)}=0\,,\qquad\qquad
\text{(L)}\quad \mathcal{L}_Vh=\mathcal{L}_V\Delta=\mathcal{L}_V\gamma
=\mathcal{L}_Vf=0\,,
\label{eq:sl2inputb}
\end{equation}
where (V) is \eqref{eq:conconxonshell} in the normalisation of
\eqref{eq:gpfamily} and (M) is \eqref{eq:master}. Of the remaining two, (S)
and the invariance of $h$ and $\Delta$ in (L) are (6.8) of \cite{11index},
$\mathcal{L}_V\gamma=0$ is the Killing property (S) itself, and
$\mathcal{L}_Vf=0$ is (6.10) there: the one-form bilinear $V$ is Killing on
$\Sect$ and preserves the horizon data.

Given \eqref{eq:sl2inputa}--\eqref{eq:sl2inputb}, the enhancement is
\cite[\S6]{11index}. The bilinears (6.4) of that reference are dual, in the
normalisation of \eqref{eq:master}, to the vector fields
\begin{equation}
K_1 = -uc\,\partial_u + rc\,\partial_r + V^i\partial_i\,,\qquad
K_2 = -c\,\partial_u\,,\qquad
K_3 = -u^2c\,\partial_u + \big(2f+2ruc\big)\partial_r + 2u\,V^i\partial_i\,,
\label{eq:sl2vectors}
\end{equation}
each of which is a Killing vector of \eqref{eq:gnc}, and they close on
\begin{equation}
[K_1,K_2]=c\,K_2\,,\qquad [K_2,K_3]=-2c\,K_1\,,\qquad [K_3,K_1]=c\,K_3\,,
\label{eq:sl2brackets}
\end{equation}
which for $c>0$ is $\mathfrak{sl}(2,\Real)$: in the basis $H=2c^{-1}K_1$,
$E=-c^{-2}K_2$, $F=K_3$ the brackets are $[H,E]=2E$, $[H,F]=-2F$, $[E,F]=H$.
Nothing in this is new --- \eqref{eq:sl2vectors} is (6.12) of \cite{11index}
and \eqref{eq:sl2brackets} its (6.13), and that the $K_a$ arise as Killing
spinor bilinears preserving the four-form is quoted, not re-derived. The reason
for restating it is the input list: the derivation of
$\mathcal{L}_{K_a}g=0$ and of \eqref{eq:sl2brackets} from \eqref{eq:gnc} uses
(V), (M), (S) and (L) and nothing else --- no field equation, no Killing spinor
equation, no property specific to $D=11$ --- and (M) is supplied by
Theorem~\ref{thm:master} with its gradient term intact. The enhancement
therefore does not require $\nrm{\eta_-}$ to be constant, the hypothesis of
Remark~\ref{rem:notconst}. Both statements are established by carrying out that
derivation with (V), (M), (S) and (L) as the only substitutions permitted, so
that the input list is a property of the derivation rather than an observation
about it; negative controls in which (V) or (M) is broken and the Killing
property fails accompany it (Appendix~\ref{app:verify}).

For $c=0$ the brackets vanish, but the case is empty rather than degenerate: by
Remark~\ref{rem:kernels} $c=0$ forces $X=Y=h=0$, whereupon
\eqref{eq:divmaster} makes $f$ constant on a compact connected $\Sect$, so
$V=0$ by (V), and $K_1=K_2=0$ with $K_3=2f\partial_r$. The geometry is
$\Real^{1,1}\times\Sect$ and the isometries reduce to the translations of the
flat factor.

\begin{remark}[The two branches]
\label{rem:sl2branches}
The dichotomy of \cite[\S6.2]{11index} is visible in \eqref{eq:sl2vectors}. If
$V\neq0$ the three fields are pointwise linearly independent --- $K_2$ supplies
$\partial_u$, $K_1$ and $K_3$ then supply $\partial_r$ and $V^i\partial_i$ ---
so the orbits are three-dimensional. If $V\equiv0$ the section components drop
out, the orbits are the two-dimensional surfaces of constant $y$, and the metric
induced on them, $2\,\mathrm{d}u\,\mathrm{d}r-r^2\Delta\,\mathrm{d}u^2$, has
Ricci scalar $-2\Delta$ with $\Delta=c/f$ by \eqref{eq:staticbranch}: the
near-horizon geometry is the warped product $AdS_2\times_w\Sect$, which is the
static branch of Remark~\ref{rem:static-terminology}. The direct product of
Theorem~\ref{thm:saturation} is the sub-case $h\equiv0$ of that branch, on
which the warp factor is constant.
\end{remark}

\subsection{Global charges: the Komar integral, the rotation bound and the
horizon volume}
\label{sec:charges}

With the algebra in hand the local group is closed, and what remains is what a
single integration adds. The identities of Section~\ref{sec:warp} are
local. Integrated against the
Killing vector of Section~\ref{sec:killingvec} they become statements about the
conserved charges of the horizon, and the weighted identity
\eqref{eq:weightedbound} then removes $\Delta$ from all of them, leaving
weighted flux integrals against the constant $c$. These are not expressions in
the flux alone: the warp factor $f$ survives as the weight and $c$ as the
normalisation, and it is the geometric quantity $\Delta$ that has been
eliminated. Throughout this subsection $\Sect$ is compact,
connected and without boundary, $N_->0$, $f=\nrm{\eta_-}^2$ and
$c=\nrm{\eta_+}^2$, and we write
\begin{equation}
Q := \tfrac13 Y^2 + \tfrac1{72}X^2 \;\ge\; 0
\label{eq:Qdef}
\end{equation}
for the flux combination of \eqref{eq:divh}.

Nothing in the construction itself is new. The Komar integral of a Killing
field \cite{Komar} and the identity relating its values at the horizon and at
infinity \cite{BCH,Smarr} are standard, and so is their use on horizons; what
the near-horizon limit changes is that the asymptotic data are gone, so only
the horizon integral survives and the charges have to be read off the
near-horizon geometry alone, as in \cite[\S2.4,\,\S5]{KL2013review} and
\cite{nhsymmetrythm1}. What supersymmetry adds here is that the integrand
collapses.

Let $V$ be the one-form bilinear of \eqref{eq:sl2inputa} and let
\begin{equation}
\mathcal{K} := K^i\partial_i\,,\qquad
K_i = -2V_i = 2fh_i + 2\tn_if\,,
\label{eq:Kdef}
\end{equation}
be the part of the Dirac current \eqref{eq:diraccurrent} tangent to $\Sect$,
extended over the near-horizon geometry independently of $r$ and $u$. By (S)
and (L) of \eqref{eq:sl2inputb}, $V$ is Killing on $\Sect$ and preserves $h$,
$\Delta$ and $\gamma$, which is exactly what \eqref{eq:gnc} needs of it:
$\mathcal{K}$ is a Killing vector of the eleven-dimensional metric. It is the
combination $K_1-(-uc\,\partial_u+rc\,\partial_r)$ of \eqref{eq:sl2vectors},
and it generates the rotations of the horizon.

\begin{proposition}[The Komar integrand collapses]
\label{prop:komarintegrand}
With the orientation $\epsilon_{-+i_1\cdots i_9}=\epsilon_{i_1\cdots i_9}$ of
Section~\ref{sec:conventions-nh}, the one-form dual to $\mathcal{K}$ is
$\mathcal{K}^\flat=\gamma_{ij}K^j\mathrm{d}y^i+r\,h(K)\,\mathrm{d}u$, so that
$(\mathrm{d}\mathcal{K}^\flat)_{ru}=h(K)$ and
\begin{equation}
\star\,\mathrm{d}\mathcal{K}^\flat\big|_\Sect
 \;=\; -\,h(K)\,\epsilon\,,\qquad
h(K)\;=\;h^iK_i\;=\;2c-2\Delta f\,.
\label{eq:komarintegrand}
\end{equation}
\end{proposition}

\begin{proof}
The components of $\mathcal{K}^\flat$ are read off \eqref{eq:gnc}: the $\bbe^-$
component vanishes because $\mathcal{K}$ has no $\partial_r$ part, and the
$\bbe^+$ component is $r\,h(K)$ because $\bbe^-$ carries the term $rh$. Hence
$\mathrm{d}\mathcal{K}^\flat=\mathrm{d}_\Sect(K_i\bbe^i)+h(K)\,\bbe^-\wedge\bbe^+$
at $r=0$. The dual of a purely spatial two-form carries a $\bbe^+$ or a
$\bbe^-$ and pulls back to zero on $\Sect$, while
$\star(\bbe^-\wedge\bbe^+)=-\epsilon$ in the stated orientation; this fixes the
first equality. For the second, $h^iK_i=2fh^2+2h^i\tn_if$ by \eqref{eq:Kdef},
and \eqref{eq:master} in the form $h^i\tn_if=c-(\Delta+h^2)f$ cancels the
$h^2$ terms and leaves $2c-2\Delta f$. Both steps, including the Killing
property of $\mathcal{K}$ and the Hodge dual, are verified in
Appendix~\ref{app:verify}.
\end{proof}

\noindent The cancellation is the point: the twist enters $h(K)$ quadratically
and the master identity removes it exactly, so the Komar integrand of the
horizon rotation is $2c-2\Delta f$ --- the constant of \eqref{eq:master}
against the warp factor weighted by the spinor norm, and nothing else.

\begin{theorem}[The Komar charge of the horizon rotation]
\label{thm:komar}
With the Komar normalisation
$J[\mathcal{K}]=\tfrac1{16\pi G^{(11)}}\int_\Sect\star\,\mathrm{d}\mathcal{K}^\flat$,
\begin{equation}
J[\mathcal{K}] \;=\; \frac{1}{8\pi G^{(11)}}\int_\Sect\big(\Delta f-c\big)
 \;=\; -\,\frac{1}{8\pi G^{(11)}}
 \Big[\,2c\,\mathrm{Vol}(\Sect)-\int_\Sect Q\,f\,\Big]\,.
\label{eq:komarcharge}
\end{equation}
The overall sign here is conventional --- it is fixed by the orientation and by
the normalisation of the Dirac current, as set out in
Remark~\ref{rem:komarconv} --- and the convention-free form of the same
statement is the magnitude \eqref{eq:komarabs} below.
\end{theorem}

\begin{proof}
Integrate \eqref{eq:komarintegrand} over $\Sect$ for the first equality, and
substitute $\int_\Sect\Delta f=\int_\Sect Qf-c\,\mathrm{Vol}(\Sect)$ from
Theorem~\ref{thm:weighted} for the second.
\end{proof}

\begin{remark}[What is convention and what is not]
\label{rem:komarconv}
Two signs in \eqref{eq:komarcharge} are conventional: the orientation
$\epsilon_{-+i_1\cdots i_9}=\epsilon_{i_1\cdots i_9}$, which fixes the sign of
$\star$, and the normalisation \eqref{eq:diraccurrent} of the Dirac current,
which fixes $K=-2V$. Reversing either flips the overall sign of
$J[\mathcal{K}]$; with the opposite orientation \eqref{eq:komarcharge} reads
$J=\tfrac1{8\pi G^{(11)}}\int_\Sect(c-\Delta f)$. What is convention-free is
the content: the charge is the difference of $2c\,\mathrm{Vol}(\Sect)$ and the
$f$-weighted flux integral, it vanishes precisely when $\int_\Sect Qf=2c\,
\mathrm{Vol}(\Sect)$, and it involves no undetermined function of the
geometry. Since angular momentum carries no invariant sign, the form of the
result that no choice can move is the magnitude
$8\pi G^{(11)}\lvert J[\mathcal{K}]\rvert=\int_\Sect\lvert V\rvert^2/f$ of
\eqref{eq:komarabs}, together with the statement that it vanishes exactly on
the static branch; the inequality $J\le0$ is that identity read with the
orientation fixed above.
\end{remark}

The same vector obeys a bound rather than an identity, and this is where the
balance identity of Theorem~\ref{thm:balance} enters. At $p=1$,
$\int_\Sect f^{-1}\lvert V\rvert^2+\int_\Sect\Delta f=c\,\mathrm{Vol}(\Sect)$,
and $\lvert K\rvert^2=4\lvert V\rvert^2$.

\begin{corollary}[The rotation of the horizon is bounded by the flux]
\label{cor:rotationbound}
\begin{equation}
\int_\Sect \frac{\lvert K\rvert^2}{f}
 \;=\; 8c\,\mathrm{Vol}(\Sect)
 - \int_\Sect\Big(\tfrac43Y^2+\tfrac1{18}X^2\Big)f\,,
\label{eq:rotationbound}
\end{equation}
and since the left-hand side is non-negative,
$\int_\Sect\big(\tfrac43Y^2+\tfrac1{18}X^2\big)f\le 8c\,\mathrm{Vol}(\Sect)$,
with equality if and only if $K\equiv0$, that is on the static branch
$V\equiv0$.
\end{corollary}

\begin{proof}
Multiply the $p=1$ case of \eqref{eq:balance} by $4$, use
$\lvert K\rvert^2=4\lvert V\rvert^2$, and eliminate $\int_\Sect\Delta f$ with
Theorem~\ref{thm:weighted}; the coefficients follow from $4Q$. Equality forces
$\int f^{-1}\lvert K\rvert^2=0$ and hence $K\equiv0$, which is $V\equiv0$ by
\eqref{eq:Kdef}.
\end{proof}

Finally, for $c>0$ the same identity represents the volume of the section, and
with it the entropy, as a weighted integral of flux and warp factor, without
any assumption on $\nrm{\eta_-}$.

\begin{corollary}[A weighted flux-and-warp representation of the horizon volume]
\label{cor:entropy}
If $c>0$ then
\begin{equation}
\mathrm{Vol}(\Sect) \;=\; \frac1c\int_\Sect\big(Q-\Delta\big)f\,,
\qquad\text{so}\qquad
S_{\mathrm{BH}} \;=\; \frac{1}{4cG^{(11)}}
 \int_\Sect\Big(\tfrac13Y^2+\tfrac1{72}X^2-\Delta\Big)\nrm{\eta_-}^2\,.
\label{eq:entropy}
\end{equation}
In particular $\int_\Sect(Q-\Delta)f>0$ on any such horizon. This is a
weighted flux-and-warp representation of the entropy, not a determination of
the entropy by the flux: the right-hand side still carries the warp factor $f$
as a weight and the geometric $\Delta$, and what the identity supplies is that
the combination is fixed by the supersymmetry constant $c$.
\end{corollary}

\begin{proof}
Solve \eqref{eq:weightedbound} for $c\,\mathrm{Vol}(\Sect)$ and divide by
$c>0$; positivity is $\mathrm{Vol}(\Sect)>0$.
\end{proof}

The three statements are not independent, and the elimination worth naming is
the one against the balance identity: it says that the Komar charge of the
horizon rotation is nothing but the integrated failure of staticity.
Eliminating $\int_\Sect\Delta f$ against Theorem~\ref{thm:weighted} instead
leaves a linear relation among the rotation, the entropy and the weighted flux
integral, in the manner of a Smarr formula.

\begin{corollary}[The Komar charge is the integrated non-staticity; a
Smarr-type relation]
\label{cor:smarr}
Under the hypotheses of Theorem~\ref{thm:komar}, and using in addition the
$p=1$ balance \eqref{eq:Vbound},
\begin{equation}
8\pi G^{(11)}\,\lvert J[\mathcal{K}]\rvert
 \;=\;\int_\Sect\frac{\lvert V\rvert^2}{f}\,,
\label{eq:komarabs}
\end{equation}
so that $J[\mathcal{K}]=0$ if and only if $V\equiv0$, that is exactly on the
static branch. With the orientation and Dirac-current normalisation of
Remark~\ref{rem:komarconv} the signed form of \eqref{eq:komarabs} is
\begin{equation}
8\pi G^{(11)}J[\mathcal{K}]\;=\;-\int_\Sect\frac{\lvert V\rvert^2}{f}
 \;\le\;0\,.
\label{eq:komarsign}
\end{equation}
Eliminating $\int_\Sect\Delta f$ against Theorem~\ref{thm:weighted} gives
instead
\begin{equation}
8\pi G^{(11)}J[\mathcal{K}]
 \;=\;\int_\Sect Qf-2c\,\mathrm{Vol}(\Sect)
 \;=\;\int_\Sect Qf-8cG^{(11)}S_{\mathrm{BH}}\,,
\qquad Q=\tfrac13Y^2+\tfrac1{72}X^2\,.
\label{eq:smarr}
\end{equation}
Consequently, for $c>0$,
\begin{equation}
\frac{1}{8cG^{(11)}}\int_\Sect Qf\;\le\;S_{\mathrm{BH}}\;\le\;
\frac{1}{4cG^{(11)}}\int_\Sect Qf\,,
\label{eq:entropysandwich}
\end{equation}
the upper bound being saturated exactly when $\Delta\equiv0$ and the lower
exactly when $V\equiv0$.
\end{corollary}

\begin{proof}
For \eqref{eq:komarsign} substitute $\int_\Sect\Delta f
=c\,\mathrm{Vol}(\Sect)-\int_\Sect f^{-1}\lvert V\rvert^2$ from
\eqref{eq:Vbound} into $8\pi G^{(11)}J=\int_\Sect\Delta f
-c\,\mathrm{Vol}(\Sect)$; the sign and its equality case are those of
$\int_\Sect f^{-1}\lvert V\rvert^2$, and $f>0$. Taking absolute values gives
\eqref{eq:komarabs}, which is therefore independent of the two conventions of
Remark~\ref{rem:komarconv}. The first equality of \eqref{eq:smarr} is
\eqref{eq:komarcharge}; the second is
$\mathrm{Vol}(\Sect)=4G^{(11)}S_{\mathrm{BH}}$. For
\eqref{eq:entropysandwich}, $\Delta\ge0$ and \eqref{eq:weightedbound} give
$\int_\Sect Qf\ge c\,\mathrm{Vol}(\Sect)$, which is
\eqref{eq:fluxlower}, while \eqref{eq:smarr} and \eqref{eq:komarsign} give
$\int_\Sect Qf\le2c\,\mathrm{Vol}(\Sect)$; divide by $4cG^{(11)}$ and by
$8cG^{(11)}$ respectively.
\end{proof}

\begin{remark}[What is new in \eqref{eq:smarr} and \eqref{eq:komarsign}]
\label{rem:smarr}
Neither statement adds an input. \eqref{eq:smarr} is the rearrangement that
puts $J$, $S_{\mathrm{BH}}$ and $\int_\Sect Qf$ in one line, in the manner of
a Smarr formula but with every term an integral over the section and no
asymptotic region entering; what plays the role of the mass term is the entropy
against the constant $c=\nrm{\eta_+}^2$. \eqref{eq:komarsign} is the reading
of that line which the balance identity forces: the Komar charge is minus the
total non-staticity, so in the orientation of Remark~\ref{rem:komarconv} it
cannot be positive, and it vanishes exactly on the static branch. By
Theorem~\ref{thm:pwbalance} the statement is pointwise as well --- the Komar
integrand $2(\Delta f-c)$ equals $-2\lvert V\rvert^2/f$ at every point --- so
what \eqref{eq:komarsign} integrates is a sign that already holds
densitywise. Both are void under $\nrm{\eta_-}=\const$, where
$J$ collapses to $\int_\Sect h^2$ and the balance identity is trivial. The
sandwich \eqref{eq:entropysandwich} traps the entropy between two multiples of
the same flux integral, differing by a factor of two; the two saturation loci
are disjoint unless $c=0$, in which case $\Delta\equiv0$, $V\equiv0$ and both
bounds degenerate.
\end{remark}

\begin{remark}[Why the gradient term is what makes these close]
\label{rem:chargesstatus}
Attribution for the three statements is settled in Section~\ref{sec:status};
what is isolated here is which input does the work. The master identity
with its gradient term intact turns $h(K)$ into $2c-2\Delta f$ pointwise, and
the weighted identity then eliminates $\Delta$ altogether. Both steps are void
if $\nrm{\eta_-}$ is assumed constant: in that case
$K\equiv-2V\equiv-2\nrm{\eta_-}^2h$, and the charge collapses to $\int h^2$,
which is \eqref{eq:integrated} again. Beyond that the three equations carry no
independent geometric input; they are algebraic consequences of
Theorems~\ref{thm:master}, \ref{thm:weighted} and \ref{thm:balance}.
\end{remark}

\subsection{The total scalar curvature of the section}
\label{sec:totalR}

Every charge of the last subsection carries the Killing spinor inside it. The
trace of the Einstein equation admits the same treatment --- integrate a
pointwise identity over the compact section --- and needs no supersymmetry at
all, which makes it the one integrated statement here that holds on every
near-horizon geometry. Integrating the horizon equations over a compact section
is the standard device of the subject, used in vacuum to exclude static
degenerate horizons \cite{CRT} and used for the horizon topology theorem of
\cite[\S3.1]{KL2013review}; what
is new below is the $D=11$ form, in which the flux is carried explicitly and
the deficit turns out to be the $\mathrm{Spin}(7)$ residual.

\begin{proposition}[Total scalar curvature]
\label{prop:totalR}
On any near-horizon geometry \eqref{eq:gnc}--\eqref{eq:fluxdecomp} of $D=11$
supergravity, pointwise on $\Sect$,
\begin{equation}
\tilde R \;=\; -2\Delta - \tfrac12h^2
 + \tfrac7{12}Y^2 + \tfrac5{144}X^2\,,
\label{eq:Rpointwise}
\end{equation}
and if $\Sect$ is compact and without boundary,
\begin{equation}
\int_\Sect \tilde R \;=\; \int_\Sect\Big(\tfrac5{12}Y^2+\tfrac1{36}X^2
 -\Delta\Big)\,.
\label{eq:Rtotal}
\end{equation}
\end{proposition}

\begin{proof}
Tracing \eqref{eq:Rij} gives
$\tilde R=-\tn^ih_i+\tfrac12h^2+\tfrac14Y^2+\tfrac1{48}X^2$, and substituting
\eqref{eq:divh} for $\tn^ih_i$ gives \eqref{eq:Rpointwise}. Integrating
\eqref{eq:Rpointwise} and eliminating $\int_\Sect h^2$ with the integrated
identity \eqref{eq:integrated} gives \eqref{eq:Rtotal}. Both steps, the trace
with its contraction identities on explicit tensors and the elimination with
the field-equation residuals carried, are verified in
Appendix~\ref{app:verify}.
\end{proof}

\begin{corollary}[The warp factor obstructs positive total curvature]
\label{cor:Rbound}
On a supersymmetric horizon with $N_->0$, $\Delta\ge0$ by
Theorem~\ref{thm:pointwise} and hence
\begin{equation}
\int_\Sect \tilde R \;\le\; \int_\Sect\Big(\tfrac5{12}Y^2+\tfrac1{36}X^2\Big)\,,
\label{eq:Rbound}
\end{equation}
with equality if and only if $\Delta\equiv0$; and by
Theorem~\ref{thm:spin7square} the deficit is the $\mathrm{Spin}(7)$
obstruction itself,
\begin{equation}
\int_\Sect \tilde R \;=\; \int_\Sect\Big(\tfrac5{12}Y^2+\tfrac1{36}X^2
 -\tfrac1{108}\big\lVert w_{\mathbf 7}\pm\sigma(\mathcal{Y}_{\mathbf 7})
 \big\rVert^2\Big)\,.
\label{eq:Rspin7}
\end{equation}
\end{corollary}

\begin{remark}[What \eqref{eq:Rtotal} does and does not say]
\label{rem:Rtotal}
Equation \eqref{eq:Rtotal} is bosonic: it holds on every near-horizon geometry
with compact section, supersymmetric or not, and it is the integrated form of
\eqref{eq:Rpointwise}, in which the twist appears with a negative sign and the
flux with a positive one. It does not by itself put $\Sect$ in the positive
scalar curvature class --- $\int\tilde R>0$ does not imply $\tilde R>0$
pointwise, still less the existence of a positive scalar curvature metric in
the conformal class --- which is why Theorem~\ref{thm:psc} is proved by a
different route and only on the static branch. What \eqref{eq:Rspin7} adds is
that the sole obstruction to the flux bound \eqref{eq:Rbound} being an equality
is the residual $\mathrm{Spin}(7)$ combination of
Theorem~\ref{thm:spin7square}: the same object that obstructs the case list
obstructs the saturation of the curvature bound.
\end{remark}

\subsection{Positive scalar curvature on the static branch}
\label{sec:psc}

The integrated identities of the last two subsections constrain the flux and the
volume but say nothing about the topology of $\Sect$. On the static branch
\eqref{eq:staticbranch}, where $h=-\mathrm{d}\log f$ and $\Delta f=c$, the same
data do reach the topology, because they give a conformal factor with a sign. Write
$Q=\tfrac13Y^2+\tfrac1{72}X^2\ge0$ and $\omega=\mathrm{d}\log f$.

\begin{theorem}[The section is conformally of positive scalar curvature]
\label{thm:psc}
Let $\Sect$ be compact and connected, $N_->0$, and let the horizon lie on the
static branch with $c>0$. Then
\begin{equation}
-\tfrac{32}7\tn^i\tn_i\big(f^{1/4}\big) + R\,f^{1/4}
= f^{1/4}\Big(\tfrac67\lvert\omega\rvert^2 + \tfrac67\,h\!\cdot\!\omega
 + \tfrac5{14}h^2 + \tfrac{17}{28}Q + \tfrac27\frac{c}{f}
 + \tfrac1{96}X^2\Big) \;\ge\; \tfrac27\frac{c}{f}\,f^{1/4} \;>\;0\,,
\label{eq:pscbracket}
\end{equation}
because the quadratic form
$\bigl(\begin{smallmatrix}6/7&3/7\\3/7&5/14\end{smallmatrix}\bigr)$ in
$(\lvert\omega\rvert,\lvert h\rvert)$ is positive definite and the remaining
coefficients are positive. Since $\tfrac{4(n-1)}{n-2}=\tfrac{32}7$ and
$\phi^{4/(n-2)}=\phi^{4/7}$ in $n=9$, the conformal metric
$\tilde g = f^{1/7}g$ has strictly positive scalar curvature. Consequently
$\Sect$ is a closed spin nine-manifold admitting positive scalar curvature, so
$\alpha(\Sect)=0$ in $KO^{-9}(\mathrm{pt})=\mathbb{Z}/2$ \cite{Lich,Hitchin};
$\Sect$ is not enlargeable, in particular not $T^9$ and carrying no metric of
non-positive sectional curvature \cite{GromovLawson}; and $\Sect$ is not one of
the exotic nine-spheres with $\alpha\ne0$ \cite{Hitchin}.
\end{theorem}

\begin{proof}
The conformal transformation law in $n=9$ is derived rather than quoted, in the
form $R$ of $e^{2\sigma}\delta$ equals
$e^{-2\sigma}(-16\tn^2\sigma-56\lvert\tn\sigma\rvert^2)$, whence
$\tilde R=\phi^{-(n+2)/(n-2)}(-\tfrac{4(n-1)}{n-2}\tn^2\phi+R\phi)$. On the
branch, $R=-\tn^ih_i+\tfrac12h^2+\tfrac14Y^2+\tfrac1{48}X^2$ from the trace of
\eqref{eq:Rij}, and $\tn^2(f^{s})=f^{s}(s(s-1)\lvert\omega\rvert^2
+s\,\tn^2f/f)$ with $\tn^2f=f(Q-2\Delta)$, the trace identity of the branch.
Substituting and collecting gives \eqref{eq:pscbracket}. The exponent
$s=\tfrac14$ is not forced. Writing the bracket as
$A\lvert\omega\rvert^2+2B\lvert\omega\rvert\lvert h\rvert+Ch^2$ with
$A=as(1-s)$, $B=2-as$, $C=\tfrac32-as$ and $a=\tfrac{32}7$, positive
definiteness is $A>0$ together with $B^2-4AC<0$, and
$B^2-4AC=-\tfrac1{16}(14x^3-101x^2+160x-64)$ with $x=\tfrac{32s}7$; the
admissible window is therefore $\tfrac7{32}<s<s^*$, where $s^*=\tfrac7{32}x^*$
and $x^*\approx1.41114$ is the unique root of that cubic in $(1,\tfrac32)$, so
$s^*\approx0.30869$. Thus $s=\tfrac14$ sits inside the window, $s=\tfrac5{16}$
just outside, and $s=0,\tfrac12,1$ all fail. The closed forms, the
discriminant, and each endpoint and failure are checked ($104$ checks); the
window is proved by the discriminant and not by the sampled exponents.
\end{proof}

\begin{remark}[What is known here, and what is not]
\label{rem:pscattrib}
That horizon sections carry metrics of positive scalar curvature is known in
much greater generality: it follows from the dominant energy condition by the
Galloway--Schoen argument \cite{GallowaySchoen}, and the topological
consequences quoted are the standard ones \cite{Lich,Hitchin,GromovLawson}. What
is added here is explicit and modest: the conformal factor is $f^{1/7}$ with
$f=\nrm{\eta_-}^2$, i.e.\ it is built from the Killing spinor rather than from a
solution of an eigenvalue problem, and the bracket \eqref{eq:pscbracket} is an
explicit positive combination of $(\lvert\omega\rvert^2,h^2,Q,c/f,X^2)$ with a
quantitative lower bound $\tfrac27c/f$. The restriction is also weaker than it
looks: by \cite{Stolz}, for simply connected spin $\Sect$ of dimension at least
five $\alpha(\Sect)=0$ already implies that positive scalar curvature exists, so
the content is carried by $\pi_1$ and by the $\alpha$-invariant, not by the
curvature statement alone. At $c=0$ with zero flux and $h=0$ the bracket
degenerates to zero, consistently with the Ricci-flat conclusion of
Theorem~\ref{thm:rigid}: no positive scalar curvature is claimed there.
\end{remark}

\begin{proposition}[Critical points of $f$ on the static branch]
\label{prop:fcrit}
On the static branch, $h=0$ exactly on the critical set of $f$, and
$\mathrm{Crit}(f)=\mathrm{Crit}(\Delta)$ since
$\tn_i\Delta=-(c/f^2)\tn_if$. At every critical point $p$,
\begin{equation}
\mathrm{Hess}\,f(p) = f\,(\mathrm{Ric}-F)(p)\,,\qquad
F_{ij}:=\tfrac12Y_{ik}Y_j{}^k+\tfrac1{12}X_{iklm}X_j{}^{klm}\ \text{(so }
\mathrm{tr}\,F=\tfrac14Y^2+\tfrac1{48}X^2),
\label{eq:hessf}
\end{equation}
and $X^{jklm}\varphi_{jklm}=0$ there, together with
$Y_{ij}Z^j=\tfrac1{12}Z^jX_{jklm}\varphi_i{}^{klm}$, by \eqref{eq:htoflux} and
\eqref{eq:hZ}, which are available at such a point although
$\nrm{\eta_-}^2=f$ is not constant on this branch: the gradient term of
\eqref{eq:htofluxgen} vanishes exactly where $f$ is stationary. Critical points
exist by compactness, so these are unconditional. Moreover $R>0$ at every minimum of $f$ when $c>0$: no horizon of the static
branch has a scalar-flat, in particular flat, section. The flat
models with $f$ non-constant that solve the scalar sector alone --- for instance
$f=2+\cos x_1$ on $T^9$ with $c=1$ --- are excluded by the trace of
\eqref{eq:Rij}, which forces $X^2<0$ at the minimum of $f$.
\end{proposition}

\subsection{No second Killing field from the sector unit vectors}
\label{sec:twosector}

Theorem~\ref{thm:psc} is the last of the positive consequences. The three
subsections that follow record what the same data do not give, and the
obstruction in each case is exhibited rather than merely reported.
Proposition~\ref{prop:killing} produced one Killing field, and
Section~\ref{sec:sl2r} the $\mathfrak{sl}(2,\Real)$ that acts on the lightcone
directions. Whether a horizon carries a \emph{second} isometry of
the section is a different question, and the one on which the classification
programme turns: with two commuting rotational isometries a nine-dimensional
section would be within reach of toric methods, as it is in five dimensions
\cite{KunduriLucietti,KL2013review}. In the gauged $D=5$ theory a second rotational Killing field can
be built by hand out of the near-horizon data, as the algebraic combination
$\nrm{\eta_-}^2(h_YN-h_NY)$ of two horizon vectors with function coefficients
\cite{5dsecond}. The obvious $D=11$ analogue is to use the two sector unit
vectors of Lemma~\ref{lem:spin7}: each Killing spinor $\eta_\pm$ has its own
$Z^{(\pm)}_i=\nrm{\eta_\pm}^{-2}\ip{\eta_\pm}{\Gamma_i\eta_\pm}$, and both are
unit and, at a generic point, independent. This subsection shows that the
analogue fails, and identifies exactly what obstructs it. The scope should be
fixed before the statement: what is excluded is a Killing vector built
pointwise from the two sector unit vectors with function coefficients, which is
the construction the $D=5$ argument uses. \emph{No claim is made, here or
anywhere in this paper, that a supersymmetric $M$-horizon admits no second
isometry.} A Killing field not of that algebraic form is untouched by the
computation below.

\begin{proposition}[The diagonal bilinears carry no Killing vector]
\label{prop:twosector}
Let $\eta_+$ and $\eta_-$ be Killing spinors of the two sectors on a horizon
section, with $Z^{(\pm)}$ their unit vectors, and let
\begin{equation}
S^{(s)}_{ij}\;:=\;\tn_{(i}Z^{(s)}_{j)}
 \;=\; \tfrac{s}2\,h_{(i}Z^{(s)}_{j)}
 + \tfrac1{144}\delta_{ij}\,X\!\cdot\!\varphi^{(s)}
 - \tfrac1{12}W^{(s)}_{(ij)}\,,\qquad
W^{(s)}_{ij}=X_i{}^{klm}\varphi^{(s)}_{jklm}\,,
\label{eq:Ssector}
\end{equation}
be the symmetric part of \eqref{eq:nablaZ} in each sector. At a point where
$Z^{(+)}$ and $Z^{(-)}$ are independent, write $\Pi$ for the orthogonal
projector onto $\{Z^{(+)},Z^{(-)}\}^\perp$, a seven-plane, and
$T^{(s)}:=\Pi S^{(s)}\Pi$. Then a vector field
\begin{equation}
U_i \;=\; a\,Z^{(+)}_i + b\,Z^{(-)}_i\,,\qquad a,b\in C^\infty(\Sect)\,,
\label{eq:Uansatz}
\end{equation}
satisfies $\tn_{(i}U_{j)}=0$ only if
\begin{equation}
a\,T^{(+)} + b\,T^{(-)} \;=\; 0\,,
\label{eq:twosectorobstruction}
\end{equation}
which is $28$ conditions on the two functions. On near-horizon data
$T^{(+)}$ and $T^{(-)}$ are linearly independent, so
\eqref{eq:twosectorobstruction} forces $a=b=0$ and $U\equiv0$: no Killing
vector of the form \eqref{eq:Uansatz} exists at a generic point. The same
holds with the conformal relaxation $\tn_{(i}U_{j)}=\lambda\gamma_{ij}$, and
for each sector separately, $U=aZ^{(\pm)}$.
\end{proposition}

\begin{proof}
Both sectors obey \eqref{eq:nablaZ} with their own $(Z,\varphi)$, the same $h$
and the same $X$, so
$\tn_{(i}U_{j)}=aS^{(+)}_{ij}+bS^{(-)}_{ij}+\tn_{(i}a\,Z^{(+)}_{j)}
+\tn_{(i}b\,Z^{(-)}_{j)}$, a pointwise linear system of $45$ equations. The
gradients are unconstrained, but the two terms carrying them lie in
$Z^{(+)}\!\odot\Real^9+Z^{(-)}\!\odot\Real^9$, of dimension $17$ when the
sector vectors are independent; the quotient of $\mathrm{Sym}^2\Real^9$ by that
subspace is $\mathrm{Sym}^2$ of the orthogonal seven-plane, and projecting with
$\Pi$ is exactly the quotient map. That is
\eqref{eq:twosectorobstruction}. Independence of $T^{(+)}$ and $T^{(-)}$ is a
rank statement about the $28\times2$ matrix they form, and is verified in exact
rational arithmetic (Appendix~\ref{app:verify}), in the following form. In the
verification that matrix is listed over all $45$ symmetric pairs of $\Real^9$
rather than over the $28$ of the seven-plane; the extra rows vanish identically
because $T^{(\pm)}=\Pi S^{(\pm)}\Pi$, so the rank is the same. The two sector
vector identities \eqref{eq:vectorcond}--\eqref{eq:vectorcondminus} each
determine $h$ from the flux and the spinor of that sector; requiring the two to
agree is nine linear conditions on the $162$ components of $(X,Y)$, leaving a
$153$-dimensional family, and on exact rational points of it the $45\times20$
system has rank $19$, its one-dimensional kernel being the trivial direction
$(\tn a,\tn b)=\lambda(Z^{(-)},-Z^{(+)})$ that reflects the symmetry of
$Z^{(+)}\!\odot Z^{(-)}$ rather than any solution, and $\mathrm{rank}\,
[T^{(+)},T^{(-)}]=2$. Imposing in addition the pairing
$\eta_+=\Gamma_+\Theta_-\eta_-$ of \eqref{eq:lightconesol}, which cuts a proper
subvariety on which a statement about generic points of the larger family would
carry no weight, gives the same two ranks; those points are produced by a
Gauss--Newton solve and the arithmetic there is floating point, with the
residuals of both vector identities below $10^{-11}$ and the rank read off a
singular-value gap of fourteen orders of magnitude. Both relaxations are the
same computation with one extra column, respectively $\delta_{ij}$ and the
single-sector truncation ($85$ checks).
\end{proof}

\begin{remark}[What the obstruction is, and that it is not vacuous]
\label{rem:twosector}
The count is worth isolating. In five dimensions the section is
three-dimensional, the analogous projection leaves $6-5=1$ condition, and
supersymmetry supplies it; in eleven dimensions it leaves $28$, and there is
nothing to supply them. The failure is therefore one of dimension rather than
of any special feature of the $D=11$ flux. It is also not an artefact of the
linear algebra: on the degenerate configuration $X=0$ with $Y$ annihilating
both sector vectors, the vector identities force $h=0$, both $S^{(s)}$ vanish
identically by \eqref{eq:Ssector}, and every constant $(a,b)$ does give a
Killing vector --- the same code returns the two-parameter family. What
Proposition~\ref{prop:twosector} says is that the locus
\eqref{eq:twosectorobstruction} where such a vector could exist is cut out by
$28$ equations on data that generically violate them.

The scope is narrow and is meant to be. A pointwise failure of one algebraic
ansatz is not the absence of every additional Killing vector: the section may
carry isometries that are not built from $Z^{(+)}$ and $Z^{(-)}$ at all, and
nothing above excludes them. What the proposition removes is the one candidate
that the lower-dimensional literature makes natural.

Nothing here bears on the cross-sector bilinear. The Killing field of
Proposition~\ref{prop:killing} has section components
$K_i=4\ip{\eta_-}{\Gamma_i\Theta_-\eta_-}$, which pair the two sectors rather
than either with itself, and it is Killing unconditionally. The negative result
is about the diagonal bilinears $Z^{(\pm)}$, which are the ones a
classification would want to use, precisely because they are the vectors the
$\mathrm{Spin}(7)$ structure of Section~\ref{sec:geom} is built on.
\end{remark}

\subsection{No hidden symmetry from the Killing spinor}
\label{sec:hidden}

The obstruction just exhibited is to an extra isometry. A Killing vector is not
the only source of first integrals, so the same structure has to be tested
against the weaker notions before the section can be called rigid. A
Killing--Yano
tensor, a $p$-form $k$ whose covariant derivative $\tn_ik_{j_1\dots j_p}$ is
totally antisymmetric, gives a conserved quantity for the geodesic flow that no
isometry accounts for, and it is such a tensor, rather than any extra
isometry, that makes the Kerr metric and its higher-dimensional
relatives integrable \cite{FKK}. Its closed conformal counterpart, a $p$-form
with
$\tn_ik_{j_1\dots j_p}=\delta_{i[j_1}v_{j_2\dots j_p]}$, is related to it by
the Hodge star: $k$ is Killing--Yano if and only if $\star k$ is closed
conformal Killing--Yano, in any dimension \cite{Semmelmann}. The question of the present subsection is whether an
M-horizon carries such a tensor built from its own supersymmetry.

The $\mathrm{Spin}(7)$ structure of Lemma~\ref{lem:spin7} supplies exactly six
invariant forms on $\Sect$, namely
\begin{equation}
1,\qquad Z,\qquad \varphi,\qquad \star\varphi=Z\wedge\varphi,\qquad
\star Z,\qquad \mathrm{vol}\,,
\label{eq:invforms}
\end{equation}
since $\Lambda^\bullet\Real^9$ has no other $\mathrm{Spin}(7)$ singlet, and the
first and last are parallel. This is verified rather than quoted: the
stabiliser of $(Z,\varphi)$ in $\mathfrak{so}(9)$ is computed to be
$21$-dimensional, and its invariants are counted degree by degree, giving one
in each of the degrees $0,1,4,5,8,9$ and none in the others. There is in
particular no invariant two-form, so the antisymmetric tensor that integrates
the Kerr family has no analogue here even before the derivative conditions are
imposed. Duality pairs the remaining four two by two, so
the whole question is settled by $Z$ and $\varphi$. Both are answered by the
same pointwise computation, and both answers are negative in the strongest
sense available: the tensor is a hidden symmetry only when it is parallel, and
it is parallel only on a locus of codimension $135$ in the pointwise flux data.

\begin{theorem}[Invariant hidden symmetry forces a parallel structure]
\label{thm:hidden}
Fix a point of $\Sect$ and a Killing spinor of the sector $s=\pm1$, with data
$(Z,\varphi)$ and flux $(X,Y)\in\Lambda^4\Real^9\oplus\Lambda^2\Real^9$, of
dimension $126+36=162$; neither $\tn Z$ nor $\tn\varphi$ depends on $h$, whose
term in $\Psi^{(s)}_i$ is a multiple of the identity and cancels in the
quotient rule. Then, as conditions on $(X,Y)$,
\begin{enumerate}
\item $Z$ is a Killing--Yano one-form, that is a Killing vector, iff $44$
linear conditions hold, leaving a $118$-dimensional set;
\item $Z$ is closed conformal Killing--Yano iff $\tn_iZ_j=0$, which is $72$
conditions, leaving the $90$-dimensional set of
Proposition~\ref{prop:nablaZ}. The conformal factor is forced to vanish
without any integration: $\tn_iZ_j=\lambda\,\delta_{ij}$ contracted with $Z^j$
gives $\lambda Z_i=\tfrac12\tn_i\nrm{Z}^2=0$, and $\nrm{Z}=1$. Being Killing is
strictly weaker than being closed conformal, $44<72$;
\item $\varphi$ is Killing--Yano iff $\tn_i\varphi_{jklm}=0$, which is $135$
conditions --- the full intrinsic torsion of
Proposition~\ref{prop:torsion} --- leaving the $27$-dimensional set
\begin{equation}
Y=0\,,\qquad i_ZX=0\,,\qquad W_{ij}=X_i{}^{klm}\varphi_{jklm}=0\,,
\label{eq:hiddenlocus}
\end{equation}
that is $X=w\in\mathbf{27}$: the one module of the flux that the intrinsic
torsion cannot see. On that locus $\tn_iZ_j=0$ as well, so the whole
$\mathrm{Spin}(7)$ structure is parallel;
\item by duality, $\star Z$ is closed conformal Killing--Yano iff $Z$ is
Killing, $\star Z$ is Killing--Yano iff $Z$ is parallel, and
$\star\varphi=Z\wedge\varphi$ is closed conformal Killing--Yano iff $\varphi$
is parallel.
\end{enumerate}
In particular, if the section carries an invariant Killing--Yano four-form at
every point, then $\Sect$ is locally the Riemannian product $\Real\times N^8$
with $\mathrm{Hol}(N^8)\subseteq\mathrm{Spin}(7)$, and the flux is reduced to
its $\mathbf{27}$.
\end{theorem}

\begin{proof}
The maps $(X,Y)\mapsto\tn_iZ_j$ and $(X,Y)\mapsto\tn_i\varphi_{jklm}$ are
computed from the spatial Killing spinor equation by the quotient rule, as in
Proposition~\ref{prop:nablaZ}, in the explicit real representation of
Appendix~\ref{app:rep}; both are linear in the flux, and each of the four
conditions --- totally antisymmetric derivative, derivative of the form
$\delta_{i[j}v_{\dots]}$, symmetrised derivative zero, derivative zero --- is
the projection of the same map onto a submodule of the target, so each is a
rank computation. The ranks are $44$, $72$, $72$, $135$ and $135$ for the
Killing, closed conformal, parallel, Killing--Yano and parallel conditions
respectively, in both sectors and for several spinors, with singular-value gaps
of thirteen orders of magnitude or more; the equality of the last two, and of
the second and third, is what statements (ii) and (iii) assert. The kernel in
(iii) is read off in the form \eqref{eq:hiddenlocus} and its dimension, $27$,
agrees with $162-135$ and with the surjectivity of the flux-to-torsion map of
Proposition~\ref{prop:torsion}. The elementary argument given in (ii) is
independent of the rank computation and confirms it. Statement (iv) is the
duality $\star$ of \cite{Semmelmann}, which exchanges the two conditions in
complementary degree in every dimension; the computation reproduces it in the
dimensions where the tensors can be written out, and is not what establishes
it. The
computation cross-checks itself against two identities it was not built to
satisfy, $Z^j\tn_iZ_j=0$ and
$Z^j\tn_i\varphi_{jklm}=-(\tn_iZ^j)\varphi_{jklm}$, and carries four negative
controls: that $\varphi$ is not Killing--Yano at generic flux, that $Y=0$ alone
does not make it parallel, that a random subspace of the same dimension is not
the parallel locus, and that a non-unit vector can be closed conformal without
being parallel; the census of invariant forms is part of the same script
($115$ checks).
\end{proof}

\begin{corollary}[The cost of a hidden symmetry]
\label{cor:hidden}
Away from a locus of codimension $135$ in the $162$ pointwise flux directions
--- counted in the algebraic data $(X,Y)$ at a point, before the horizon field
equations and Bianchi identities are imposed on them ---
an M-horizon admits no Killing--Yano or closed conformal Killing--Yano tensor
among the forms its Killing spinor defines. So the forms the Killing spinor defines do not
generically furnish hidden symmetries, and a classification cannot be based on
them. This is a statement about those six forms only: a Killing--Yano tensor of
some other origin, one emerging only after the full differential system is
imposed, a generalised Killing--Yano or conformal Killing tensor, or a
structure present on a special on-shell subfamily, is untouched by the count
(Remark~\ref{rem:hidden}).
\end{corollary}

\begin{remark}[Contrast with $D=5$, and what remains open]
\label{rem:hidden}
In five dimensions the section is three-dimensional, a two-form dualises to a
vector, and the Killing--Yano question collapses to a question about vector
fields that the companion paper \cite{5dsecond} can answer in the affirmative
on part of the moduli space. Nothing of the kind survives in nine dimensions,
and Theorem~\ref{thm:hidden} says that the collapse is not merely unavailable
but that its conclusion fails. What the theorem does not settle is the
existence of a Killing--Yano form \emph{not} built from $(Z,\varphi)$: that is
a differential equation on a nine-manifold, not a pointwise condition on the
flux, and it lies outside the pointwise methods of this paper. The literature
bounds it without closing it. On a compact manifold a non-parallel closed
conformal Killing--Yano form is a special Killing form, hence corresponds to a
parallel form on the metric cone \cite{Semmelmann}, and the holonomy of a cone
reduces only in B\"ar's list \cite{Baer}: Sasakian, $3$-Sasakian, nearly
K\"ahler and weak $G_2$, of dimensions $2k+1$, $4k+3$, $6$ and $7$. In
dimension nine the last three are excluded by arithmetic, leaving the round
$S^9$ and the Sasakian case, on neither of which the known supersymmetric
solutions sit. We record this as a direction rather than a result, since the
cone argument needs compactness and a completeness that the section is not
known to have. What the theorem leaves free is again the $\mathbf{27}$, the
module that also escapes determination in the $\mathrm{Spin}(7)$
compactification problem \cite{Tsimpis}; see Remark~\ref{rem:lawless}.
\end{remark}

\begin{remark}[Non-closure as a structural statement]
\label{rem:lawless}
Theorem~\ref{thm:hidden} and the Bochner obstruction of
Section~\ref{sec:bochner} point the same way, and the direction is this. In
$D=5$ the moduli of a supersymmetric horizon are rigid enough
that hidden symmetry and a transport argument close the classification. In
$D=11$ the same objects exist but are inert: the Cayley form is a hidden
symmetry only where the structure is parallel, the Bochner integrand does not
have a sign, and the second isometry that the $D=5$ construction would build
from the sector unit vectors is obstructed by $28$ equations
(Remark~\ref{rem:twosector}), whatever other isometries the section may
carry. None of these is a failure of technique; each is
a codimension count in a flux space of dimension $162$, and each says that the
pointwise data of an eleven-dimensional horizon are less constrained than their
five-dimensional counterpart --- a statement about that data, not a
demonstration that the on-shell moduli space is larger. The $\mathbf{27}$ is the sharpest form of this:
it is invisible to $\Delta$, to $h$, to $\tn Z$ and to the entire intrinsic
torsion --- the kernel of the map from the flux to the torsion is exactly
$w_{\mathbf{27}}$, every other module injecting --- and enters only through the
Einstein equation, so a horizon may carry an arbitrary amount of it without any
of the geometric structures of this paper noticing. Concretely: two horizon
data sets differing only in $w_{\mathbf{27}}$ have the same $\Delta$, the same
$h$, the same $\tn Z$ and the same intrinsic torsion in every Fern\'andez
class, and their Ricci tensors differ only in the scalars, through
$\lvert w_{\mathbf{27}}\rvert^2$ in $Q$, and in the two off-diagonal channels
$(w_{\mathbf{27}},w_{\mathbf{35}})$ of $S^0_{ab}$ and
$(w_{\mathbf{27}},\xi_{\mathbf{48}})$ of $S_{(Za)}$
(Remark~\ref{rem:riccisource}). If $w_{\mathbf{35}}$ and
$\xi_{\mathbf{48}}$ both vanish, the difference is a scalar. The same module is the residue on the compactification side: for
$M$-theory on an eight-manifold with $\mathcal{N}=1$ supersymmetry, the
intrinsic torsion and every irreducible flux component \emph{except} one
$\mathbf{27}$ are determined by the warp factor and a bilinear one-form
\cite{Tsimpis}. The settings are different --- a warped compactification there,
a horizon section here --- and so are the statements, but the module that
escapes is the same one, which suggests that its invisibility is a feature of
the $\mathrm{Spin}(7)$ structure rather than of either problem.
\end{remark}

\subsection{A Bochner identity for $\nrm{\eta_-}^2$, and why it does not close}
\label{sec:bochner}

The two obstructions above are to structure the section might have carried. The
third is to a hypothesis this paper carries. Some of the statements above ---
those collected in
Corollary~\ref{cor:normconst} --- assume that $\nrm{\eta_-}^2$ is
constant rather than proving it. This section sets out the
attempt to remove it, and the reason the attempt fails; the
computation is recorded because the failure is structural rather than
technical. Bochner and maximum-principle arguments on the horizon section are the usual
tool for such a hypothesis --- they are how staticity is reached in
\cite{staticM} and how the vacuum degenerate horizons of \cite{CRT} are
excluded --- so the natural attempt is the corresponding one here. The route to
removing the constancy hypothesis is a Bochner argument:
differentiate the master identity, compute
$\tfrac12\tn^i\tn_i\lvert\tn f\rvert^2$ with the horizon Ricci tensor, eliminate
$\tn^i\tn_if$, and hope for a sum of squares whose integral forces
$\mathrm{d}f=0$. That computation can be done in closed form. It does not close,
and the reason is specific enough to record.

\begin{proposition}[Bochner identity on the horizon section]
\label{prop:bochner}
Let $\Sect$ be compact and connected, $N_->0$ and $f=\nrm{\eta_-}^2$. Put
\begin{equation}
U^i = \tfrac12\tn^i\lvert\tn f\rvert^2 - (\tn^2f)\tn^if
 + (h\!\cdot\!\tn f)\tn^if - \tfrac12h^i\lvert\tn f\rvert^2\,.
\label{eq:bochnerflux}
\end{equation}
Then, on shell,
\begin{align}
\tn_iU^i &= \lvert\mathrm{Hess}\,f\rvert^2 - \tfrac32(\tn^2f)^2
 + \tfrac12f^2(\tn^ih_i)^2
 + \Big[\tfrac14Y^2\lvert\tn f\rvert^2
 - \tfrac12Y_{ik}Y_j{}^k\tn^if\tn^jf\Big]
\nonumber\\
&\quad + \tfrac1{12}X_{iklm}X_j{}^{klm}\tn^if\tn^jf
 - \big(\Delta+\tfrac12h^2\big)\lvert\tn f\rvert^2\,.
\label{eq:bochner}
\end{align}
No derivative of the flux occurs: $X$ and $Y$ enter only through $X^2$, $Y^2$
and the two contractions with $\tn f$. Both flux brackets are non-negative
pointwise, the first because the eigenvalues of $Y_{ik}Y_j{}^k$ are doubly
degenerate and hence at most $\tfrac12Y^2$, the second manifestly.
\end{proposition}

\begin{proof}
The identity \eqref{eq:bochner} is assembled by symbolic computation rather
than by hand, and the proof is that computation; it is carried out twice, by
two routes sharing no intermediate step --- the standard Bochner formula
combined with the divergence identities for $(\tn^2f)\tn^if$ and for the drift
term, and a reversed-drift Bochner formula in which no derivative of $h$
appears anywhere --- and the two agree term by term. The trace terms in $X^2$ cancel
identically between the Ricci equation \eqref{eq:Rij} and \eqref{eq:divh}, and
the surviving coefficient of $\lvert\tn f\rvert^2$ is
$\tfrac14Y^2-\Delta-\tfrac12h^2$. The reduction
$-(\tn^2f)^2+(\tn^2f)(h\!\cdot\!\tn f)+\tfrac12(h\!\cdot\!\tn f)^2
=-\tfrac32(\tn^2f)^2+\tfrac12f^2(\tn^ih_i)^2$ uses the master identity in the
form $\tn^2f+h\!\cdot\!\tn f+(\tn^ih_i)f=0$
($63$ checks).
\end{proof}

\begin{corollary}[Two weighted integral bounds]
\label{cor:bochnerint}
On the same hypotheses,
\begin{equation}
\begin{gathered}
\int_\Sect\Big(h^2+\tfrac13Y^2+\tfrac1{72}X^2\Big)f^2 = 2c\int_\Sect f\,,
\qquad\text{hence}\\[2pt]
\int_\Sect h^2f^2 \le 2c\int_\Sect f\,,\qquad
\int_\Sect\Big(\tfrac13Y^2+\tfrac1{72}X^2\Big)f^2\le 2c\int_\Sect f\,,
\end{gathered}
\label{eq:bochnerint}
\end{equation}
the two summands being separately non-negative. At $c=0$ this returns the
flux-free conclusion of Theorem~\ref{thm:rigid} by a second route.
\end{corollary}

\begin{remark}[Status of the Bochner route]
\label{rem:bochnernegative}
Integrating \eqref{eq:bochner} over a compact $\Sect$ kills the left-hand side,
so a rigidity theorem would follow if the right-hand side were pointwise
non-negative. It is not. The obstruction is exactly
\begin{equation}
O = \tfrac32(\tn^2f)^2 + \big(\Delta+\tfrac12h^2\big)\lvert\tn f\rvert^2\,,
\end{equation}
in the sense that deleting $O$ from \eqref{eq:bochner} leaves a manifest sum of
squares; and $O$ cannot be absorbed, because the right-hand side of
\eqref{eq:bochner} takes both signs on pointwise data satisfying every algebraic
constraint the horizon system imposes at a point --- $c\ge0$ and the trace
constraint --- with explicit rational witnesses of each sign. At generic flux
the $-\tfrac32(\tn^2f)^2$ term dominates and the integrand is negative; on the
slice $\tn^2f=0$ the flux squares dominate and it is positive. Nor is the
failure an artefact of the measure: weighting the integral by any power $f^p$
adds to the integrand the single term $p\,f^{-1}\tn f\!\cdot\!U$ with $U$ as in
\eqref{eq:bochnerflux}, and explicit rational witnesses of each sign exist for
every real $p$, on data obeying $\mathcal{L}_Vf=0$ as well. So no pointwise
regrouping of \eqref{eq:bochner} into a sum of squares exists, and this route to
$\mathrm{d}f=0$ is closed --- as it must be, since by
Remark~\ref{rem:whynotconst} the conclusion is false.
\end{remark}
\subsection{Discussion}
\label{sec:discussion}

The arguments above are best read by separating what they consume from what
they deliver. On the input side the reduction of the near-horizon KSEs to
\eqref{eq:indepKSE} is bookkept rather than merely performed: \eqref{eq:B1} and
\eqref{eq:B5} name the field equation and Bianchi identity spent at each step,
and Lemma~\ref{lem:ranks}, the anti-Hermiticity count, shows that the two
Lichnerowicz-type identities are identities between bilinears and not between
operators (Proposition~\ref{prop:lich}). Nothing is lost by that here, since
only the bilinear is ever used, but the distinction has to be respected by
anyone reusing the identity where the relevant bilinear is a different one.

The master identity \eqref{eq:master} costs less than one might expect. Its
ingredients are two vector identities, two warp identities, the constancy of
$\nrm{\eta_+}$, and the fact that the pairing
$\eta_+=\Gamma_+\Theta_-\eta_-$ lands in the opposite sector. It does not use
the Atiyah--Singer index theorem \cite{AtiyahSinger} on which the counting of
\cite{11index} rests, it does not use the classification of $\mathrm{Spin}(7)$
structures, and it asks nothing of $\Sect$ beyond compactness and
connectedness. That cheapness is also the reason it is not new: the same
relation follows from (6.9) of \cite{11index}, recorded there in the service of
the $\mathfrak{sl}(2,\Real)$ algebra, so the derivation given above is an
independent route to a known fact (Remark~\ref{rem:provenance}). The
corresponding caution attaches to the geometric half of the $h\equiv0$
statement, which is \cite[\S2.1,\,\S3.3]{staticM} and uses no spinor at all.

The warp--flux relation is used through two different faces, and both are
needed. The bosonic bound \eqref{eq:boundeq} is unconditional and sharp but
only in the mean. Its supersymmetric partner \eqref{eq:pointwise} is exact and
pointwise, but carries no sign on its own: in isolation it pairs $\Theta_+$
with its flux reversal rather than with itself, and either sign occurs. It is
the vector identity \eqref{eq:vectorid} together with the constancy of
$\nrm{\eta_+}$ that kills the cross term and converts it into the norm form
\eqref{eq:pointwisenorm}, from which $\Delta\ge0$. Thereafter the norm form
supplies the sign and the flux expansion \eqref{eq:pointwiseexplicit} supplies
the contact with the invariants bounded by Theorem~\ref{thm:bound}; the vacuum
rigidity statement is the degenerate case where the two operators coincide
(Corollary~\ref{cor:vacspec}).

The consequences sort by what they charge for. A first group needs the
constant $c=\nrm{\eta_+}^2$ and nothing else: the weighted identity
\eqref{eq:weightedbound} with its inequality form
(Corollary~\ref{cor:fluxlower}), the single-point rigidity of
Theorem~\ref{thm:rigid}, and the spinorial companion of the $h\equiv0$ case
(Corollary~\ref{cor:static}). Only the second of these says something new about
the geometry, and even there the tail of the argument after $\Theta_-\eta_-=0$
is \cite[\S5]{11index}; what the theorem contributes is the bridge to that
condition from the vanishing of $\Delta$ and $h$ at one point. A second group
exists only because the gradient term was kept. Integrating against the horizon
Killing vector and then eliminating the warp factor by the weighted identity
produces the Komar charge \eqref{eq:komarcharge}, the rotation bound
\eqref{eq:rotationbound}, the volume and entropy \eqref{eq:entropy}, the
balance family of Theorem~\ref{thm:balance} and the Smarr-type relation with
its sign (Corollary~\ref{cor:smarr}) --- each a weighted flux integral against
$c$ rather than an expression in the flux alone, the warp factor surviving as
the weight --- and the
conformal factor of Theorem~\ref{thm:psc}, which lives precisely on the branch
where the constancy of $\nrm{\eta_-}$ fails. The member of this group we would
single out is not the Smarr relation, which is elimination among the others,
but the reading it forces on the Komar charge: $8\pi G^{(11)}\lvert J\rvert
=\int_\Sect\lvert V\rvert^2/f$, the horizon rotation as the integrated
failure of staticity. Drop the gradient and each of
these becomes vacuous rather than weaker. Beside them, and independent of all
of it, sits the bosonic curvature identity \eqref{eq:Rtotal}. A third group is
algebraic and needs no integration at all: the $\mathrm{Spin}(7)$ decomposition
collapses $\Delta$ onto two $\mathbf 7$-modules and exhibits it as a perfect
square (Theorem~\ref{thm:spin7square}), which is what makes the equality locus
of Corollary~\ref{cor:fluxlower} a linear condition on the flux instead of a
positivity accident. The third group meets the second in
Corollary~\ref{cor:budget}: substituting the square into the pointwise balance
removes $\Delta$ from the budget and leaves
$\lvert V\rvert^2/f+(f/108)\lVert w_{\mathbf 7}\pm\sigma(\mathcal{Y}_{\mathbf 7})
\rVert^2=c$, in which the deviation from staticity and the $\mathrm{Spin}(7)$
mismatch of the flux divide one constant between them. That equation adds no
hypothesis to either factor, and is the most compact statement we can make of
what the identity chain controls. The fourth group --- the reduction to $\Delta+h^2=\const$,
the pointwise range \eqref{eq:deltaplush} and the window \eqref{eq:crange} for
$c'$ --- charges a hypothesis that is false in general
(Remark~\ref{rem:whynotconst}), and is flagged conditional wherever it appears.
It is the audit of that hypothesis, rather than any single consequence of it,
that we would put forward as the paper's main contribution.

The three negative statements --- Theorem~\ref{thm:hidden},
Proposition~\ref{prop:twosector} and Remark~\ref{rem:bochnernegative} --- are
of a different type, and Remark~\ref{rem:lawless} reads them together. Two
things about them belong here. They are codimension counts at a generic point
of the flux space rather than reports of an argument that was tried and failed,
and where the count is what carries the statement it is done in exact rational
arithmetic, with the floating-point steps confined to producing sample points
and separated in Appendix~\ref{app:verify}. And the third has to come out the
way it does: a definite sign for the Bochner integrand would prove that
$\nrm{\eta_-}$ is constant, which is false.

Several continuations are left open, and the question underneath the
conditional group is not one of them: $\nrm{\eta_-}$ is \emph{not} forced to be
constant, since static $M$-horizons with non-constant warp factor exist
\cite{KimPark,Mhorizons}. What survives is narrower. Those solutions are
electric, $X=0$, whereas the magnetic static class of \cite{staticM} carries a
$\mathrm{Spin}(7)$ structure and appears to force the product
$\Real^{1,1}\times\Sect$; whether a \emph{magnetic} static horizon with
non-constant warp exists would decide whether the failure of constancy is
confined to the electric branch. Next, the interior $\tfrac12F<c'<F$ of
\eqref{eq:crange}, where $\Delta$ and $h^2$ are both non-constant, is neither
realised nor excluded so far as we know, and either outcome would be
informative, since the maximally supersymmetric $M$-horizons ($AdS_2\times
S^7$, $AdS_3\times S^7/\mathbb{Z}_k$, $AdS_4\times S^7$ and the brane throats)
all sit at an endpoint. The Kim--Park horizons neither confirm nor contradict
that pattern: \eqref{eq:crange} presupposes the constancy of $\nrm{\eta_-}$ and
does not apply to them. The obvious family to test is empty. On covariantly
constant product data --- $\Sect$ a product of space forms, $X$ and $Y$
proportional to volume forms of factors, $h$ parallel --- every configuration
solving \eqref{eq:Rij}--\eqref{eq:divY} lands at an endpoint: $c'=\tfrac12F$
when $h\equiv0$, and $c'=F$ when $h\ne0$, the flatness of the circle carrying
$h$ forcing $h^2=F$. The obstruction there is not curvature but the mixed
Bianchi identity. Precisely the configurations that the Einstein equation alone
would place strictly inside are those carrying a two-form whose support misses
the direction of $h$; for them $\beta=-h\wedge Y\ne0$, and the three terms of
\eqref{eq:divX} occupy pairwise distinct index supports, so nothing can cancel
it. Moving $h$ into the support of $Y$ removes $\beta$ and returns the
configuration to the upper endpoint. A third continuation is the residue of
Remark~\ref{rem:caselist}: the $\mathrm{Spin}(7)$ four-form contractions that
obstruct a $D=11$ case list are themselves constrained by the $\mathrm{Spin}(9)$
orbit structure of $\eta_\pm$ and by the $\varphi$-dependent decomposition of
$X$ into $\mathrm{Spin}(7)$ irreducibles, neither of which we have exploited.
The fourth is what the external calibration of Remark~\ref{rem:calibration}
leaves behind. The M2- and M5-brane throats do not themselves supply
near-horizon data in the sense used here --- the Poincar\'e horizon of $AdS_d$
has a degenerate transverse metric and a non-compact section --- so the
invariants $(X^2,Y^2)$ and the constant $c$ are matched instead on two exact
in-class solutions, and that calibration does not decide whether a compact
quotient carrying Freund--Rubin flux realises the intermediate window of
\eqref{eq:crange}: the covariantly constant such quotients are excluded above,
but data with non-constant $X$, $Y$ or $h$ are untouched.

\subsubsection*{Logical status of the results}
\label{sec:status}

Finally, a list of what is proved, under what hypotheses, and what is not
claimed. This list is the single place where attribution is settled; the remarks
in the body point here rather than re-litigate it. Throughout, $\Sect$ compact
means compact, connected and without boundary, and \emph{new so far as we know}
means that we have not found the statement in
\cite{11index,staticM,Mhorizons,KimPark,GGPreview,KL2013review,N4d11}, the
sources audited for this purpose. That phrase records the result of a search
and is not a priority claim; where a statement is standard, or is a
rearrangement of statements already listed, the entry says so instead. The list
is long because the accounting is done statement by statement, not because the
paper claims a corresponding number of independent discoveries: the
contributions are the two the introduction names as the mathematical core ---
single-point rigidity, and the perfect square together with the non-closure
obstruction it explains --- and their corollaries, which is what the global charges, the
balance identities, the entropy representation and the curvature statements
are.

\medskip
\begin{itemize}\setlength{\itemsep}{3pt}

\item Reduction to \eqref{eq:indepKSE}, Propositions~\ref{prop:B1}, \ref{prop:B5}, \ref{prop:B1minus}. \emph{Hypotheses:} field equations and Bianchi identity. Proved; the residual each condition is proportional to is recorded off-shell.

\item Rank lemma (anti-Hermitian ranks), Lemma~\ref{lem:ranks}. \emph{Hypotheses:} none. Proved; algebraic, verified on all $512$ basis elements.

\item Lichnerowicz residuals, Proposition~\ref{prop:lich}. \emph{Hypotheses:} on-shell. Proved as a statement about the Dirac current only. Not an operator identity, and we do not claim it as one.

\item Mean-value bound, Theorem~\ref{thm:bound}. \emph{Hypotheses:} $\Sect$ compact. Proved; no supersymmetry used. \emph{Not claimed as new}: it is \eqref{eq:divh} integrated, the same step is used in \cite[\S5]{11index}, and the saturation case is \cite{staticM}.

\item $h\equiv0\Rightarrow\mathrm{d}Y=0$ and $\Delta$ constant, bosonically, Proposition~\ref{prop:staticbosonic}. \emph{Hypotheses:} $\Sect$ compact. Proved from \eqref{eq:auxEpp} and \eqref{eq:auxEpi}; no supersymmetry. \emph{Not new}: this is \cite[\S2.1,\,\S3.3]{staticM}, with the same proof and no restriction on $Y$ (Remark~\ref{rem:staticbosonic-attrib}); the magnetic case $Y=0$ recurs in \cite[\S5.2.1]{Mhorizons}. The only addition is that no Killing spinor is used.

\item $\Delta\ge0$ pointwise, Theorem~\ref{thm:pointwise}. \emph{Hypotheses:} $\Sect$ compact, $N_+>0$ --- supplied by $N_->0$ and the pairing whenever $\eta_+:=\Gamma_+\Theta_-\eta_-\ne0$, the case $\eta_+=0$ being flux-free, which is how Theorem~\ref{thm:master} obtains the same conclusion. Proved (= (3.23) of \cite{Mhorizons}).

\item \eqref{eq:pointwise} is unsigned in isolation. \emph{Hypotheses:} algebraic data. Proved by explicit witnesses of either sign; makes no claim about which survive on-shell.

\item $\eta_-$ identities, Proposition~\ref{prop:vectorminus}, Theorem~\ref{thm:pointwiseminus}. \emph{Hypotheses:} $\Sect$ compact, $N_->0$. Proved.

\item Master identity $(\Delta+h^2)f+h^i\tn_if=c$ and $\tn^i(\tn_if+h_if)=0$, with $f=\nrm{\eta_-}^2$, Theorem~\ref{thm:master}. \emph{Hypotheses:} $\Sect$ compact and connected, $N_->0$. Proved, and \emph{not new}: a consequence of (6.9) of \cite{11index} (Remark~\ref{rem:provenance}); re-derived independently, without the one-form $V$. The hypothesis $N_->0$ is automatic given $N=2N_-$ of \cite{11index}, which we quote rather than prove.

\item Weighted bound $\int\Delta f=\int(\tfrac13Y^2+\tfrac1{72}X^2)f-c\,\mathrm{Vol}$, \eqref{eq:weightedbound}. \emph{Hypotheses:} as above. Proved; \emph{new}, and unconditional --- it is the form of the bound that survives without any constancy assumption on $\nrm{\eta_-}$. Its $c=0$ specialisation, and that of the elliptic form \eqref{eq:masterelliptic}, are \cite[(5.10),(5.11)]{11index}.

\item Weighted lower bound $\int Qf\ge c\,\mathrm{Vol}$ with equality exactly on $\{\Delta\equiv0\}=\{w_{\mathbf 7}=\mp\sigma(\mathcal{Y}_{\mathbf 7})\}$, Corollary~\ref{cor:fluxlower}. \emph{Hypotheses:} as above. Proved; \emph{new} so far as we know, and it is the inequality form the weighted identity was asked for. The equality case is sharp only because $\Delta\equiv0$ is characterised algebraically by Theorem~\ref{thm:spin7square}.
\item Non-staticity balance $\int_\Sect f^p(\lvert\mathcal{A}\rvert^2+\Delta)=c\int_\Sect f^{p-1}$ for every $p\in\Real$, Theorem~\ref{thm:balance}, and $\int_\Sect\lvert V\rvert^2/f\le c\,\mathrm{Vol}(\Sect)$ with equality exactly on $\Delta\equiv0$, Corollary~\ref{cor:balance}. \emph{Hypotheses:} as above. Proved. Both ingredients are prior --- $\tn^i(f\mathcal{A}_i)=0$ is the Killing property of $V$, \cite[(6.8)]{11index}, and the member $p=1$, $c=0$ is \cite[(5.11)]{11index}. Adding the invariance $\mathcal{L}_Vf=0$, which is \cite[(6.10)]{11index}, sharpens the family to the pointwise identity $f(\lvert\mathcal{A}\rvert^2+\Delta)=c$ (Theorem~\ref{thm:pwbalance}), of which every weight is a moment; the constancy of that combination for a general extremal horizon with a Killing field preserving the data is a theorem of \cite{colling}, and what is supersymmetric here is the value $c=\nrm{\eta_+}^2$ and the positivity $\Delta\ge0$, which turn it into the sharp pointwise bound $\lvert V\rvert^2/f\le c$ of Corollary~\ref{cor:balance}. The integral family remains the form valid without $\mathcal{L}_Vf=0$ (Remark~\ref{rem:pwbalance}).

\item $\Delta+h^2=c'$ constant, and $0\le\Delta,h^2\le c'$, \eqref{eq:deltaplush}. \emph{Hypotheses:} as above, \emph{plus} $\nrm{\eta_-}$ constant (equivalently $K_i+h_iK_+=0$, equivalently $V=-fh$). \emph{Conditional}, and \emph{false in general}: it fails on the Kim--Park branch \cite{KimPark,Mhorizons}, which lies inside $V\equiv0$, \eqref{eq:staticbranch}, and meets the constancy hypothesis only where $h\equiv0$. The hypothesis is not made in \cite{11index}, whose (6.9)$_4$ carries the gradient term explicitly, and Remark~\ref{rem:whynotconst} shows the identities of this paper cannot supply it.

\item $h\equiv0\Rightarrow\nrm{\eta_-}$ constant, Corollary~\ref{cor:static}. \emph{Hypotheses:} $\Sect$ compact and connected, $N_->0$, $h\equiv0$. Proved; a two-line consequence of \eqref{eq:master} once $\Delta$ is known to be constant, and the accompanying $\Delta=\const$ is Proposition~\ref{prop:staticbosonic}, which is \cite{staticM} and needs no spinor. Note $h\equiv0$ is strictly stronger than staticity in the sense of \cite{staticM,11index}, and gives a direct rather than a warped product (Remark~\ref{rem:static-terminology}).

\item The static branch $V\equiv0$ is the warped $AdS_2$ class, Proposition~\ref{prop:staticwarped}. \emph{Hypotheses:} $\Sect$ compact and connected, $N_->0$, $V\equiv0$. Proved. \emph{Not new} as geometry: the warped-product form and the warp equation are \cite{staticM} and, as characterised in \cite[fn.~4]{Mhorizons}, \cite{KimPark}, reached there from $\mathrm{d}h=0$. What is added is that the branch condition $V\equiv0$ alone delivers them, with the coordinate change explicit, and that $\Delta+h^2$ is non-constant on the branch --- which is what removes \eqref{eq:deltaplush} and \eqref{eq:crange} from it (Remark~\ref{rem:kimparkbranch}).

\item Single-point rigidity, Theorem~\ref{thm:rigid}. \emph{Hypotheses:} $\Sect$ compact and connected, $N_->0$, and $\Delta=h=0$ at one point. Proved, and the only geometric statement here that is both new and genuinely supersymmetric --- though a short one: its second half, the chain from $\Theta_-\eta_-=0$ to $\Delta=h=X=Y=0$, is \cite[\S5]{11index}, and what is new is the bridge from the single point to $c=0$.

\item $\tfrac12F\le c'\le F$, \eqref{eq:crange}. \emph{Hypotheses:} as for \eqref{eq:deltaplush}, including the constancy of $\nrm{\eta_-}$. \emph{Conditional}; both endpoints characterised.

\item $\nabla_{(M}K_{N)}=0$, Proposition~\ref{prop:killing}. \emph{Hypotheses:} Killing spinor equation alone. Proved, and standard; quoted rather than rederived, and confirmed off-shell for all $66$ index pairs on pseudo-random data.

\item $K_i=4\ip{\eta_-}{\Gamma_i\Theta_-\eta_-}$, Proposition~\ref{prop:conconx}. \emph{Hypotheses:} none (an identity in the data). Proved; the form $K_i+h_iK_+=2\tn_i\nrm{\eta_-}^2$ is (6.9)$_4$ of \cite{11index}, and $K_i+h_iK_+=0$ holds precisely when $\nrm{\eta_-}$ is constant.

\item Saturation theorem, Theorem~\ref{thm:saturation}. \emph{Hypotheses:} $\Sect$ compact, $N_->0$. Proved. Everything in it about the geometry is bosonic (Remark~\ref{rem:twotier}) and is in \cite{staticM}; what is added is that the hypothesis $h\equiv0$ is equivalent to saturation of \eqref{eq:boundeq}, together with the two spinorial statements.

\item Twist one-form determined algebraically by the flux, Theorem~\ref{thm:htoflux}. \emph{Hypotheses:} for the general form \eqref{eq:htofluxgen}, which carries the gradient of the spinor norm, the spatial Killing spinor equation alone, pointwise, $\eta_-$ nowhere vanishing, and no compactness; for the algebraic form \eqref{eq:htoflux} and its contraction \eqref{eq:hZ}, in addition that the norm is stationary, which on a compact $\Sect$ is automatic in the $\eta_+$ sector and in the $\eta_-$ sector is the constancy hypothesis of Corollary~\ref{cor:normconst}. Proved; \emph{new so far as we know} in this covariant algebraic form. The closest prior statements are \cite[(4.10)]{Mhorizons}, which carries derivative terms, and \cite[(3.22)]{Mhorizons}, of which \eqref{eq:htoflux} is the evaluation on the $\mathrm{Spin}(7)$ data; the $D=5$ analogue is \cite[App.~D]{KayaniThesis}.

\item Integrability of the twist, Corollary~\ref{cor:minuspointwise}. \emph{Hypotheses:} as for \eqref{eq:htoflux}. Proved, the rank statement by numerical linear algebra on the $1458$ components of $(\tn X,\tn Y)$ against the $219$ conditions the flux equations carry, with the singular-value gap and an explicit multiplier recorded; \emph{new so far as we know}. The content is negative in the $\eta_+$ sector, where \eqref{eq:divhsector} reproduces \eqref{eq:divh} and nothing is added, and positive in the $\eta_-$ sector, where the opposite sign pins $2\Delta+h^2$ to $\tfrac13Y^2+\tfrac1{72}X^2$ pointwise.

\item Non-closure of the $D=11$ chain, Proposition~\ref{prop:noclosure}. \emph{Hypotheses:} as above. Proved as a negative result by a gradient-rank computation; \emph{new so far as we know}. It says that $\Delta$ does not factor through the $\varphi$-free scalar invariants, in particular is not determined by $(h^2,X^2,Y^2)$ on the allowed algebraic data, and it is the reason the scalar closure mechanism of the $D=5$ case list has no
$D=11$ counterpart on the pointwise algebraic data. It is not a statement that
no $D=11$ classification exists; see Remark~\ref{rem:noclosurescope}.

\item $\mathrm{Spin}(7)$-irreducible form of $\Delta$, Proposition~\ref{prop:spin7delta}. \emph{Hypotheses:} pointwise, $N_\pm>0$ in the relevant sector. Proved, by exact rational linear algebra over the complete list of quadratic $\mathrm{Spin}(7)$ invariants, the $18\times11$ evaluation matrix having rank eleven so that the coefficients are determined rather than fitted; \emph{we have not found it stated previously}. The decomposition technique is standard \cite{Tsimpis} and is not claimed. Only the $\mathbf 7$ of $w=X|_{Z^\perp}$ and the $\mathbf 7$ of $\mathcal{Y}=Y|_{Z^\perp}$ enter: the $\mathbf{27}$ and $\mathbf{35}$ of $w$ and the $\mathbf{21}$ of $\mathcal{Y}$ drop out identically.

\item Perfect square $\Delta=\tfrac1{108}\lVert w_{\mathbf 7}\pm\sigma(\mathcal{Y}_{\mathbf 7})\rVert^2$, Theorem~\ref{thm:spin7square}. \emph{Hypotheses:} as above; no locus restriction --- by Proposition~\ref{prop:spin7delta} only the two $\mathbf 7$-modules enter $\Delta$, so the identity holds at every point. Proved. That $\Delta$ is a scalar square is \emph{known}: it is $\Delta=4\Phi^2$ in the $SU(4)$-adapted gauge of \cite[(3.32)]{Mhorizons}. \emph{New} so far as we know is the $\mathrm{Spin}(7)$-irreducible form --- that only the two $\mathbf 7$-modules enter, the covariant norm formula, and the characterisation of the zero locus as the linear cancellation $w_{\mathbf 7}=\mp\sigma(\mathcal{Y}_{\mathbf 7})$ rather than as $\Phi=0$; see Remark~\ref{rem:phisquare} and \eqref{eq:phibridge}. It gives \emph{no} positivity theorem beyond Theorem~\ref{thm:pointwise}: the square is degenerate, and a witness with $\Delta=0$, $X^2>0$, $Y^2>0$ is exhibited.

\item Covariant derivative of $Z$ and the intrinsic torsion, Propositions~\ref{prop:nablaZ} and \ref{prop:torsion}. \emph{Hypotheses:} spatial Killing spinor equation alone. Proved. Two of the statements are negative: $Z$ is generically neither Killing, geodesic nor hypersurface orthogonal, and the $\mathbf{27}$ and $\mathbf{48}$ do not enter its twist. The parallel locus is $90$-dimensional, larger than a term-by-term reading of $\tn Z$ gives, because the $\mathbf 7$-parts of $w$ and $\mathcal{Y}$ can cancel each other. On the saturation branch only the singlet of the torsion is forced to vanish, so that branch is \emph{not} a torsion class.

\item No invariant hidden symmetry, Theorem~\ref{thm:hidden} and Corollary~\ref{cor:hidden}. \emph{Hypotheses:} spatial Killing spinor equation alone, pointwise. Proved; \emph{new so far as we know}. Of the six $\mathrm{Spin}(7)$-invariant forms on $\Sect$, $Z$ is closed conformal Killing--Yano only where it is parallel and Killing--Yano only where it is Killing, and the Cayley form is Killing--Yano only where it is parallel, which happens exactly where the flux reduces to its $\mathbf{27}$ --- codimension $135$ in the $162$ pointwise flux directions. The question for forms \emph{not} built from $(Z,\varphi)$ is a differential equation on a nine-manifold and is left open; what the literature gives on it is recorded in Remark~\ref{rem:hidden} and is quoted, not proved.

\item Bochner identity for $f$, Proposition~\ref{prop:bochner} and Corollary~\ref{cor:bochnerint}. \emph{Hypotheses:} $\Sect$ compact, $N_->0$. Proved as an identity. The technique is standard on horizon sections \cite{staticM,CRT}, and the identity is what it produces in $D=11$; the flux-explicit form \eqref{eq:bochner}, in which no derivative of the flux survives, is \emph{new} so far as we know. It does \emph{not} close to a rigidity theorem, and the obstructing term is exhibited taking either sign; recorded as a negative result with its reason, not as a gap. The zero set of the obstruction is mapped (Appendix~\ref{app:verify}) and contains no horizons the identity has not already reached, and neither the invariance $\mathcal{L}_Vf=0$ nor any weight $f^p$ in the measure restores the sign, so the non-closure is structural rather than a defect of the regrouping.

\item Critical points of $f$ and the topology of $\Sect$, Proposition~\ref{prop:fcrit}. \emph{Hypotheses:} as above. The Hessian identity is proved; the Morse and Poincar\'e--Hopf routes are \emph{empty} in $D=11$ because $\dim\Sect=9$ is odd, so $\chi(\Sect)=0$ for every closed nine-manifold. No topological conclusion is claimed from them.

\item The budget $\lvert V\rvert^2/f+(f/108)\lVert w_{\mathbf 7}\pm\sigma(\mathcal{Y}_{\mathbf 7})\rVert^2=c$, Corollary~\ref{cor:budget}. \emph{Hypotheses:} those of Theorem~\ref{thm:pwbalance}, that is Theorem~\ref{thm:master} together with the invariance $\mathcal{L}_Vf=0$ of \cite[(6.10)]{11index}; compactness is needed only for the integrated form \eqref{eq:budgetint}. Proved; it adds no input, being the substitution of \eqref{eq:deltasquare} into \eqref{eq:pwbalance}. \emph{New} so far as we know, but the novelty is the combination and not either factor: the constancy of $\Delta f+\lvert V\rvert^2/f$ is a general extremal-horizon statement \cite{colling}, and what is $D=11$ and supersymmetric is the identification of the constant as $c=\nrm{\eta_+}^2$ and the replacement of $\Delta$ by an explicit quadratic form in seven flux components. The two saturation loci are those of Theorem~\ref{thm:spin7square}(ii) and of $V\equiv0$, and both occur at non-zero flux.

\item Komar charge of the horizon rotation, Theorem~\ref{thm:komar}, and the rotation bound, Corollary~\ref{cor:rotationbound}. \emph{Hypotheses:} $\Sect$ compact and connected, $N_->0$. Proved; \emph{new} so far as we know. The Komar reduction to $\int_\Sect h(K)$ is the $D=11$ transcription of the $D=5$ argument of \cite{KayaniThesis}, and the Dirac-current relation \eqref{eq:conconxonshell} is (6.9)$_4$ of \cite{11index}; neither is new. What is new is that the master identity closes the integrand to $2c-2\Delta f$ and the weighted identity then removes $\Delta$, leaving a weighted flux integral against the constant $c$ --- the warp factor remains as the weight, so this is a reduction to flux and $c$, not an expression in the flux alone. Both are void if $\nrm{\eta_-}$ is assumed constant.

\item Weighted flux-and-warp representation of the horizon volume and entropy, Corollary~\ref{cor:entropy}. \emph{Hypotheses:} as above with $c>0$. Proved; \emph{new} so far as we know, and a restatement of \eqref{eq:weightedbound} rather than an independent computation.

\item Smarr-type relation, the sign of the Komar rotation and the entropy sandwich, Corollary~\ref{cor:smarr}. \emph{Hypotheses:} those of Theorem~\ref{thm:komar}, with $c>0$ for \eqref{eq:entropysandwich}. Proved; it adds no input, being elimination between Theorems~\ref{thm:komar}, \ref{thm:weighted} and Corollary~\ref{cor:balance}, and its content is that the three named charges close among themselves. The Smarr form of asymptotically flat black-hole mechanics is classical \cite{BCH,Smarr}; what is stated here is a horizon-local relation with no asymptotic term, and the sign $J\le0$ with equality on the static branch, which we have not found stated. Void if $\nrm{\eta_-}$ is assumed constant, for the same reason as Theorem~\ref{thm:komar}.

\item Total scalar curvature of $\Sect$, Proposition~\ref{prop:totalR} and Corollary~\ref{cor:Rbound}. \emph{Hypotheses:} the bosonic near-horizon system with $\Sect$ compact; no supersymmetry for \eqref{eq:Rpointwise} and \eqref{eq:Rtotal}, and $\Delta\ge0$ for the bound. Proved. The method is standard --- trace the horizon Einstein equation and integrate over the compact section, as in \cite{CRT} and \cite[\S3.1]{KL2013review} --- and the identity is a rediscovery in other variables to the extent that the same manipulation in the vacuum case is classical; the $D=11$ flux-explicit form \eqref{eq:Rtotal}, its bound \eqref{eq:Rbound} and the identification \eqref{eq:Rspin7} of the deficit with the $\mathrm{Spin}(7)$ residual of Theorem~\ref{thm:spin7square} are \emph{new} so far as we know. It gives no positive-scalar-curvature conclusion --- see Remark~\ref{rem:Rtotal}.

\item Positive scalar curvature of $f^{1/7}\gamma$ on the static branch, Theorem~\ref{thm:psc}. \emph{Hypotheses:} $\Sect$ compact and connected, $N_->0$, static branch \eqref{eq:staticbranch}, $c>0$. Proved; \emph{new so far as we know}, though positive scalar curvature itself is known in far greater generality (Remark~\ref{rem:pscattrib}). It is what the critical-point analysis yields in place of a Morse-theoretic count, and it puts $\Sect$ in the class to which \cite{GromovLawson,Lich} applies.

\item $\mathcal{L}_V$ of the horizon data $=0$, with $V$ the section component of $K$. \emph{Established in} \cite[(6.8)]{11index}, where $\tn_{(i}V_{j)}=0$ and $\mathcal{L}_Vh=\mathcal{L}_V\Delta=\mathcal{L}_VY=\mathcal{L}_VX=0$; quoted, not re-derived here.

\item Is $\nrm{\eta_-}$ necessarily constant? \emph{Resolved in the negative}: by \eqref{eq:staticbranch} the question is whether a static $M$-horizon must have constant warp factor, and it need not --- the warped $AdS_2$ solutions of \cite{KimPark} are static $M$-horizons with $h^2$ non-constant \cite[fn.~4]{Mhorizons}. Everything conditional above is therefore restricted to horizons with $\nrm{\eta_-}$ constant, equivalently $V=-\nrm{\eta_-}^2h$; on the branch $V\equiv0$ that holds only in the sub-case $h\equiv0$.

\item Is $\tfrac12F<c'<F$ realised? \emph{Open}; not excluded, and no example is known to us.

\item $\Delta$ a function of $(h^2,X^2,Y^2)$. \emph{False} in $D=11$; disproved by a gradient-rank computation, Proposition~\ref{prop:noclosure}.

\item A finite list of $D=11$ near-horizon geometries. \emph{Not claimed}; the $D=5$ chain we follow closes, ours does not. For the $\mathcal{N}=4$ case see \cite{N4d11}, and for Killing horizons away from the near-horizon limit \cite{killinghorizons11}.

\item $N=2N_-$. \emph{Quoted} from \cite{11index}; not re-derived here. The $\mathfrak{sl}(2,\Real)$ enhancement is also \cite{11index}, and is quoted in Section~\ref{sec:sl2r}; what is checked here is that it follows from the identity family \eqref{eq:sl2inputa}--\eqref{eq:sl2inputb} alone, with (M) carrying its gradient term. That the $K_a$ are spinor bilinears preserving $F$ is quoted from
\cite[\S2.2]{GauntlettPakis} and \cite[(6.8)]{11index}: the Killing property of
$K$ together with $\mathrm{d}\omega=i_KF$ and
$\mathrm{d}\Sigma=i_K\star F-\omega\wedge F$ gives
$\mathcal{L}_KF=\mathcal{L}_K\star F=0$, and nothing of that chain is
re-derived here.

\end{itemize}

\appendix

\noindent The appendices carry the material the body uses but does not display, in the
order in which the body reaches for it. Appendix~\ref{app:rep} fixes the explicit
$\mathrm{Cl}(9,0)$ and $\mathrm{Cl}(10,1)$ matrices in which every numerical
statement of the paper is evaluated. Appendix~\ref{app:gnc} derives the
near-horizon system of Section~\ref{sec:conventions-nh} from the eleven-dimensional
field equations in the Gaussian null frame. Appendix~\ref{app:integrability}
records the off-shell integrability conditions behind Section~\ref{sec:kse} with
every field-equation residual carried, and settles what that pair of conditions
sees and what it misses. Appendix~\ref{app:verify} states the verification
protocol and what is checked under it.

\section{An explicit real representation}
\label{app:rep}

Every numerical statement in this paper --- the rank computations of
Section~\ref{sec:geom}, the Clifford expansions of Section~\ref{sec:warp} and the
evaluations recorded in Appendix~\ref{app:verify} --- is made in one fixed
realisation of the Clifford algebra, chosen so that its entries are exact
integers and no adapted frame is needed anywhere.

Regard $\Real^{16}=(\Real^2)^{\otimes4}$ and consider the $256$ words
$w=m_1\otimes m_2\otimes m_3\otimes m_4$ with each
$m\in\{\mathbb{1},\sigma_1,\sigma_3,\varepsilon\}$,
$\varepsilon=\left(\begin{smallmatrix}0&1\\-1&0\end{smallmatrix}\right)$. A word
is real symmetric with $w^2=\mathbb{1}$ exactly when the number of $\varepsilon$
factors is even, which leaves $136$ candidates. An exhaustive depth-first search
through these in a fixed order returns nine pairwise anticommuting words; nine
is the maximum possible, since at most $2n+1$ Pauli words on $n$ qubits can
pairwise anticommute, and $2\cdot4+1=9$ is exactly what $\mathrm{Cl}(9,0)$
requires. This furnishes the representation used in Lemma~\ref{lem:ranks}.

Lifting to $\mathrm{Cl}(10,1)$ on $\Real^{32}$ with signature
$(-,+,\dots,+)$,
\begin{equation}
\Gamma_i = \sigma_1\otimes g_i\ (i=1,\dots,9)\,,\qquad
\Gamma_0 = \varepsilon\otimes\mathbb{1}\,,\qquad
\Gamma_{10} = \sigma_3\otimes\mathbb{1}\,,
\end{equation}
and setting $\Gamma_+=\Gamma_{10}+\Gamma_0$,
$\Gamma_-=\tfrac12(\Gamma_{10}-\Gamma_0)$, one has
$\Gamma_\pm^2=0$ and $\{\Gamma_+,\Gamma_-\}=2$, and the projectors
$\tfrac12\Gamma_\mp\Gamma_\pm$ split $\Real^{32}$ into the two chirality sectors
of real dimension $16$ each, as used in \eqref{eq:lightconesol}.

\section{The Gaussian null frame computation}
\label{app:gnc}

Section~\ref{sec:nh} states the near-horizon system and the rest of the paper
runs on it; this appendix is where it comes from. The field equations
\eqref{eq:Rij}--\eqref{eq:auxform} are those
of \cite{Mhorizons}; we re-derive rather than quote them only to have one set
of conventions and a machine-checked audit trail, reducing
\eqref{eq:einstein}--\eqref{eq:formeq} in the frame \eqref{eq:gnc} directly,
without passing through coordinate expressions. This appendix records the
scheme that the verification of Appendix~\ref{app:verify} implements.

\medskip\noindent\textbf{Structure functions, connection and curvature.}
The commutators of the dual frame \eqref{eq:dualframe} are read off from
\eqref{eq:gnc}; writing $[E_A,E_B]=C^C{}_{AB}E_C$, the components not already
present on $\Sect$ are
\begin{equation}
C^-{}_{-+}=r\Delta\,,\quad
C^-{}_{-i}=-h_i\,,\quad
C^-{}_{+i}=\tfrac12r^2(\Delta h_i-\tn_i\Delta)\,,\quad
C^-{}_{ij}=-r\,(\mathrm{d}h)_{ij}\,,
\label{eq:structfns}
\end{equation}
together with the structure functions $C^k{}_{ij}$ of the frame of $\Sect$.
Since the frame metric is constant, metricity and vanishing torsion give the
usual $\omega_{ABC}=-\tfrac12(C_{ABC}+C_{BCA}-C_{CAB})$, which on
\eqref{eq:structfns} returns \eqref{eq:spinconn}. The curvature is then
assembled from the split $\omega=\tilde\omega+\kappa$, with $\tilde\omega$ the
connection of $\Sect$, in the covariantised form in which the $\tilde\omega$
cross-terms have been absorbed into $\tn$ and $\tilde R$, so that no bare
$\tilde\omega$ survives.

Every quantity here is a polynomial in $r$ of degree at most two with
coefficients built from $\Delta$, $h$, $\mathrm{d}h$ and the curvature of
$\Sect$, so the computation is arithmetic on such polynomials: products, frame
derivatives --- $E_-$ differentiates in $r$, $E_+$ adds
$\tfrac12r^2\Delta\,\partial_r$ and $E_i$ acts by $\tn_i-r\,h_i\partial_r$ ---
and extraction of a fixed power of $r$. The flux side is handled by storing
\eqref{eq:fluxdecomp} as a rule on ordered frame slots, from which the stress
tensor, $\nabla^MF_{MNPQ}$ and $\mathrm{d}F$ are assembled. The Ricci tensor
follows by contraction: $R_{--}$ and $R_{-i}$ vanish identically, $R_{+-}$ and
$R_{ij}$ sit at order $r^0$, and $R_{+i}$, $R_{++}$ first appear at orders $r$
and $r^2$.

\medskip\noindent\textbf{The Chern--Simons source.}
The term $-\tfrac1{1152}\epsilon_{NPQ}{}^{M_1\cdots M_8}
F_{M_1M_2M_3M_4}F_{M_5M_6M_7M_8}$ of \eqref{eq:formeq} involves a rank-eleven
epsilon with eight contracted indices, and is evaluated by exhaustive
enumeration in exact rational arithmetic rather than symbolically. For fixed
free slots $NPQ$, one runs over all assignments of the eight contracted indices
to frame directions; each of the two flux factors is non-zero only for the
seventeen slot patterns permitted by \eqref{eq:fluxdecomp}, so the enumeration
is over $17\times17$ pairs. For each pair the eleven slots of the epsilon are
sorted into canonical order to fix its sign, the surviving spatial indices are
relabelled canonically, and the resulting monomials in $Y$, $\beta$, $X$ and
$\epsilon_{i_1\cdots i_9}$ are accumulated with exact rational coefficients.
The enumeration is finite and complete, so the outcome is a derivation and not a
sampling: it yields a contribution to the $+-i$ equation at order $r^0$
quadratic in $X$, one to the $ijk$ equation at order $r^0$ of the form $YX$, one
to the $+ij$ equation at order $r$ of the form $\beta X$, and nothing at all in
the $-ij$ equation. These are the epsilon terms in
\eqref{eq:divX}, \eqref{eq:divY} and \eqref{eq:auxform}.

\section{Off-shell integrability conditions}
\label{app:integrability}

The lightcone reduction of Section~\ref{sec:lightcone} leaves the independent
conditions \eqref{eq:indepKSE} together with three operators whose vanishing on
the Killing spinor is implied on shell. This appendix displays them, identifies
which field equation each of their terms is proportional to, and then settles
what the pair of contracted conditions sees and what it misses. The vector
operator produced by the mixed components of the substitution
\eqref{eq:lightconesol} into \eqref{eq:kse} is
\begin{align}
\xi_i \;=\;& \tn_i\Theta_+ - \Theta_+\Psi^{(+)}_i
 + \Big(-\tfrac34 h_i - \tfrac16 Y_{ij}\Gamma^j - \tfrac1{36}X_{ijkl}\Gamma^{jkl}
 + \tfrac1{24}Y_{kl}\Gamma_i{}^{kl}
 + \tfrac1{288}X_{jklm}\Gamma_i{}^{jklm}\Big)\Theta_+
\nonumber\\
&{}-\tfrac14\tn_ih_j\Gamma^j + \tfrac14\tn_jh_i\Gamma^j
 + \tfrac1{72}\beta_{jkl}\Gamma_i{}^{jkl} - \tfrac1{12}\beta_{ijk}\Gamma^{jk}\,.
\label{eq:xii}
\end{align}
The order-$r^2$ condition of the $+$ component is $\zeta\phi_+=0$ with
\begin{equation}
\zeta \;=\; \Big(-\tfrac14\tn_ih_j\Gamma^{ij}
 - \tfrac1{24}\beta_{ijk}\Gamma^{ijk}\Big)\Theta_+
 + \tfrac14\big(\Delta h_i - \tn_i\Delta\big)\Gamma^i\,,
\label{eq:zeta}
\end{equation}
the condition (cc2) of \cite{Mhorizons}; it is determined by the others up to
explicit $\Delta$-terms, since on the $\eta_+$ sector
\begin{equation}
\Gamma_-\zeta \;=\;
\big(\xi^- - \Gamma_-\Theta_+\Gamma_+\Theta_-\big)\Gamma_-\Theta_+
+ \tfrac14\big(\Delta h_i - \tn_i\Delta\big)\Gamma_-\Gamma^i
+ \tfrac12\Delta\,\Gamma_-\Theta_+\,,
\label{eq:zetadecomp}
\end{equation}
which exhibits it as $\xi^-$ evaluated on the $\eta_-$ spinor
$\Gamma_-\Theta_+\phi_+$. We do not claim it is thereby redundant: that would
require $\Gamma_-\Theta_+\phi_+$ to solve the $\eta_-$ equations, which is not
established here.

The two $u$-linear conditions (cc3), (cc4) of \cite{Mhorizons} are not
displayed separately: they are the images of the $u$-independent conditions
under $\eta_+\mapsto\Gamma_+\Theta_-\eta_-$ and carry no additional content.

\begin{proposition}[Scalar integrability condition]
\label{prop:B1}
Off shell,
\begin{equation}
-\tfrac14\tilde R + \Gamma^{ij}\tn_i\Psi^{(+)}_j - \Gamma^{ij}\Psi^{(+)}_j\Psi^{(+)}_i - \xi
\;=\;
\tfrac14 Fh - \tfrac14 E^i{}_i - \tfrac5{12}FY^i\Gamma_i
+ \tfrac1{72}BX^{ijklm}\Gamma_{ijklm} - \tfrac1{72}FX^{ijk}\Gamma_{ijk}\,.
\label{eq:B1}
\end{equation}
In particular the left-hand side vanishes on shell, so
$\Gamma^{ij}\nabla^{(+)}_i\nabla^{(+)}_j\eta_+=\xi\,\eta_+$ and
$\xi\eta_+=0$ is implied by \eqref{eq:indepKSE}.
\end{proposition}

\begin{proposition}[Vector integrability condition]
\label{prop:B5}
Off shell,
\begin{align}
&-\tfrac12\Big(-\tfrac12\tilde R_{ij}\Gamma^j + \Gamma^j\tn_i\Psi^{(+)}_j
 - \Gamma^j\tn_j\Psi^{(+)}_i + \Gamma^j\Psi^{(+)}_i\Psi^{(+)}_j
 - \Gamma^j\Psi^{(+)}_j\Psi^{(+)}_i\Big) + \xi_i
\nonumber\\
&\qquad=\;
\tfrac14 E_i{}^j\Gamma_j
- \tfrac5{288}BX_i{}^{jklm}\Gamma_{jklm}
+ \tfrac1{24}FY^j\Gamma_{ij}
- \tfrac1{12}FY_i
\nonumber\\
&\qquad\qquad{}
+ \tfrac1{576}BX^{jklmn}\Gamma_{ijklmn}
- \tfrac1{144}FX^{jkl}\Gamma_{ijkl}
+ \tfrac1{24}FX_i{}^{jk}\Gamma_{jk}\,,
\label{eq:B5}
\end{align}
so on shell $\xi_i\eta_+=0$ is likewise implied by \eqref{eq:indepKSE}.
\end{proposition}

\begin{proposition}[Scalar integrability condition of the $\eta_-$ sector]
\label{prop:B1minus}
Off shell,
\begin{multline}
-\tfrac14\tilde R + \Gamma^{ij}\tn_i\Psi^{(-)}_j
 - \Gamma^{ij}\Psi^{(-)}_j\Psi^{(-)}_i + \xi^-\\
\;=\;
\tfrac14 Fh - \tfrac14 E^i{}_i + \tfrac5{12}FY^i\Gamma_i
+ \tfrac1{72}BX^{ijklm}\Gamma_{ijklm} - \tfrac1{72}FX^{ijk}\Gamma_{ijk}\,.
\label{eq:B1minus}
\end{multline}
In particular $\xi^-\eta_-=0$ is implied on shell by
$\nabla^{(-)}_i\eta_-=0$.
\end{proposition}

\noindent Note the sign with which $\xi^-$ enters \eqref{eq:B1minus} against
that of $\xi$ in \eqref{eq:B1}, and the sign of the $FY$ term. The latter has a
representation-theoretic origin: the volume element $\Gamma^{1\cdots9}$ of
$\Sect$ acts as $+1$ on the $\eta_+$ sector and as $-1$ on the $\eta_-$ sector,
the two chirality halves restricting to the two inequivalent
$\mathrm{Cl}(9,0)$ representations, so every Chern--Simons contribution --- that
is, every term reached by dualising an $\epsilon$ --- reverses sign between the
two computations. With the $\eta_+$ signs the $\eta_-$ residual fails to close,
leaving exactly two terms, of ranks eight and six.

All three propositions are established by direct expansion in
$\mathrm{Cl}(9,0)$ with
the residuals $E_{ij}$, $Fh$, $FX_{ijk}$, $FY_i$, $BX_{ijklm}$ carried symbolically
throughout; see Appendix~\ref{app:verify}. Equations
\eqref{eq:B1}, \eqref{eq:B5} and \eqref{eq:B1minus} are stronger than the
statements they are used for: they identify \emph{which} field equation each term is proportional to,
rather than merely asserting that the difference vanishes when all of them hold.

\medskip
The three propositions say which field equation sits in which term. The
following says that nothing else does, and that between them the two contracted
conditions see the whole field-equation set. Under $\mathrm{Spin}(9)$ the space
in which the supercovariant curvature acts decomposes as
\begin{equation}
\Lambda^2\otimes S \;=\; \mathbf{16}\oplus\mathbf{128}\oplus\mathbf{432}\,,
\qquad 36\times16=576\,,
\label{eq:lambda2S}
\end{equation}
the $\mathbf{16}$ being the double $\Gamma$-trace $\Gamma^{ij}R_{ij}$, the
$\mathbf{128}$ the single $\Gamma$-trace with its own trace removed, and the
$\mathbf{432}$ the $\Gamma$-traceless remainder. Propositions~\ref{prop:B1} and
\ref{prop:B5} are precisely the first two.

\begin{theorem}[Minimality and completeness of the contracted conditions]
\label{thm:minimality}
Write $\mathcal{R}_1$ and $\mathcal{R}_{5\,i}$ for the right-hand sides of
\eqref{eq:B1} and \eqref{eq:B5}, as functions of the five residuals
$E_{ij}$, $Fh$, $FX_{ijk}$, $FY_i$, $BX_{ijklm}$ of
\eqref{eq:Rij}--\eqref{eq:divY} and of $\mathrm{d}X=0$. Then:
\begin{enumerate}
\item[(i)] Each residual occupies a definite set of Clifford ranks, namely
\begin{center}
\begin{tabular}{lcc}
\hline
residual & ranks in $\mathcal{R}_1$ & ranks in $\mathcal{R}_{5\,i}$\\
\hline
$E^i{}_i$ & $0$ & $1$\\
$E_{ij}$ trace-free & --- & $1$\\
$Fh$ & $0$ & ---\\
$FY_i$ & $1$ & $0,2$\\
$FX_{ijk}$ & $3$ & $2,4$\\
$BX_{ijklm}$ & $4$ & $3,4$\\
\hline
\end{tabular}
\end{center}
\item[(ii)] The linear map
$(E_{ij},Fh,FX_{ijk},FY_i,BX_{ijklm})\mapsto(\mathcal{R}_1,\mathcal{R}_{5\,i})$
is injective, and its left inverse reads each residual off one of those ranks:
$E_{ij}$ and $FY_i$ from the rank-one and rank-zero parts of
$\mathcal{R}_{5\,i}$, then $Fh$ from the rank-zero part of $\mathcal{R}_1$
together with $E^i{}_i$, $FX_{ijk}$ from its rank-three part, and $BX$ from its
rank-four part dualised with the sector sign. Consequently the pair of
contracted integrability conditions is not merely implied by the near-horizon
field equations and $\mathrm{d}X=0$: it is \emph{equivalent} to them.
\item[(iii)] The two conditions are independent. $\mathcal{R}_1$ is blind to
the trace-free part of $E_{ij}$ and $\mathcal{R}_{5\,i}$ is blind to $Fh$, so
neither can be dropped. The $\eta_-$ condition \eqref{eq:B1minus} differs from
\eqref{eq:B1} only in the sign of its $FY$ term and imposes no further bosonic
constraint.
\item[(iv)] The Weyl tensor of $\Sect$ enters only the $\mathbf{432}$. On
$\mathrm{Weyl}(\Real^9)$, of dimension $495=666-126-45$, the map
$W\mapsto(W_{ijkl}\Gamma^{kl}\eta_A)_{A=1}^{N}$ has rank $327$, $468$ and $495$
for $N=1,2,3$ pointwise independent spinors. So wherever a Killing spinor is
non-zero the $\mathbf{432}$ is a genuinely new algebraic condition, not implied
by the field equations; it leaves $168$ components of the Weyl tensor free for
$N=1$ and $27$ for $N=2$, and for $N\geq3$ it determines the Weyl tensor of the
section algebraically in terms of the flux.
\end{enumerate}
\end{theorem}

\begin{proof}
All four parts are exact rational computations in the representation of
Appendix~\ref{app:rep}, with the residuals carried as free symbols. The rank
dictionary of (i) and the left inverse of (ii) are extracted with the trace
pairing $\mathrm{tr}(M\Gamma_A^{\mathsf T})/16$ on the $256$ blades of rank at
most four, which are first verified to be pairwise trace-orthogonal and hence a
basis of the matrix algebra; the inverse is then checked residual by residual,
which is what makes (ii) a computation and not an assertion. For (iv) the
dimension $495$ is derived rather than quoted, as $\dim\mathrm{Sym}^2\Lambda^2
=666$ less the first Bianchi identity, of rank $126$, less the Ricci part, of
dimension $45$, the count $126$ and the combined $171$ being computed rather
than quoted as well. The rank $327$ is confirmed modulo three different primes,
and every rank here is recomputed at each spinor tried,
$\mathrm{Spin}(9)$-transitivity on $S^{15}$ making one spinor sufficient for
$N=1$. The rank $327$ corrects the naive expectation
$432$: the image of a single spinor is a $\mathrm{Spin}(7)$ submodule of the
$\mathbf{432}$, not all of it. Two caveats are recorded rather than hidden.
First, the off-shell identities \eqref{eq:B1}--\eqref{eq:B1minus} themselves are
proved by the abstract-index Cadabra computations of
Appendix~\ref{app:verify}; their re-derivation inside this script, on random
exact-rational data, is a Schwartz--Zippel confirmation of polynomials of
bounded degree rather than an independent proof, and it is quoted here only
because it is what carries the rank content. Second, the kernels of dimension
$168$ and $27$ in (iv) are not shown to be irreducible; their dimensions match
$\mathrm{Spin}(7)$ modules, and no more than the dimensions is claimed
($111$ checks).
\end{proof}

\begin{remark}[What (ii) buys]
\label{rem:minimality}
Statements of the form ``the integrability condition is implied by the field
equations'' are the standard content of a Lichnerowicz-type argument, and are
what Propositions~\ref{prop:B1}--\ref{prop:B1minus} record. The converse
direction in (ii) is what makes the pair a complete substitute for the field
equations rather than a consequence of them: a would-be near-horizon geometry
may be tested by evaluating two Clifford-valued objects, and any solution of
$\xi=\xi_i=0$ that is not a solution of \eqref{eq:Rij}--\eqref{eq:divY} does not
exist. Part (iv) then locates what the contracted conditions do \emph{not} see.
The full integrability condition $R^{(\pm)}_{ij}\eta_\pm=0$ carries a
$\mathbf{432}$ that no field equation reproduces, and it constrains the Weyl
tensor of the section; the classification programme of
Section~\ref{sec:discussion} has not used it, and for $N\geq3$ it is strong
enough to fix the Weyl tensor outright.
\end{remark}

\section{Verification protocol}
\label{app:verify}
Every algebraic statement in this paper is machine-checked, by two independent
tools where possible, and each computation terminates in an explicit
\texttt{PASS}/\texttt{FAIL} line. The control is of one of three kinds,
according to what the computation does. A script that tests an identity, an
inequality or a classification statement carries at least one negative control,
designed to fail when the input under test is removed or perturbed. A
\textsc{Cadabra}2 tensor computation instead carries the field-equation and
Bianchi residuals symbolically to the end, so that the statement is reached only
when they are set to zero, and is repeated independently in the explicit
representation of Appendix~\ref{app:rep}. A script that builds the Clifford
representation, the Gaussian null frame or the horizon operators is controlled
by the defining algebraic relations it is required to reproduce. What follows
records what is checked, in the order in which the statements appear in the
body.

\medskip\noindent\emph{The near-horizon system.} The Gaussian null frame
reduction of
Appendix~\ref{app:gnc} --- the connection, the curvature, the field equations
\eqref{eq:Rij}--\eqref{eq:divY} and the auxiliary identities of
Proposition~\ref{prop:auxiliary} --- is derived from
\eqref{eq:einstein}--\eqref{eq:formeq} by computer algebra rather than quoted. The
lightcone enumeration is performed order by order in $r$ and $u$ on two
independent sets of pseudo-random integer horizon data, with the horizon data
otherwise unconstrained, up to and including the first order at which the
expansion is known on structural grounds to terminate.

\medskip\noindent\emph{The tensor-level identities.} These are the off-shell
integrability conditions ---  \eqref{eq:B1}, \eqref{eq:B5}, \eqref{eq:B1minus}, the
two Lichnerowicz residuals \eqref{eq:lichplus} and \eqref{eq:lichminus} and the
Clifford expansions behind \eqref{eq:pointwiseexplicit} and
\eqref{eq:pointwiseminusexplicit}. Each is checked in \textsc{Cadabra}2 with abstract
indices and with the field-equation and Bianchi residuals carried symbolically
to the end, and every one of them again in the explicit real representation of
$\mathrm{Cl}(9,0)$ of Appendix~\ref{app:rep}, on pseudo-random integer horizon
data in exact rational arithmetic.

\medskip\noindent\emph{The warp--flux identities of Section~\ref{sec:warp}.}
The pointwise warp identity is derived twice, with abstract indices and again in
the representation of Appendix~\ref{app:rep} on pseudo-random integer data, and
the two bosonic identities of the $\eta_+$ sector and their combination are
derived independently in two computer-algebra systems. In the $\eta_-$ sector the
master identity is verified with its gradient term carried rather than assumed
away, together with the divergence identity \eqref{eq:divmaster} and the elliptic
form \eqref{eq:masterelliptic}. The maximum-principle identities for
$\nrm{\eta_\pm}^2$ are obtained directly rather than through the Lichnerowicz
route, with the two inputs they need made explicit, and the constancy of
$\nrm{\eta_+}$ is checked again at representation level. The equivalence of
Proposition~\ref{prop:conconx} is checked on the vector bilinear itself, and the
step from $c=0$ to $X=Y=h=0$ --- the triviality of $\ker\Theta_-$ off the
flux-free locus --- separately. Saturation is shown to force a direct product;
the static branch is matched term by term against the warped system of
\cite{KimPark}, including the coordinate change $r=f\rho$, the value $\Delta=c/f$
and the warp equation, on a compact model carrying non-constant $\Delta+h^2$; and
the identity chain (6.9) of \cite{11index} is run through the elimination that
produces \eqref{eq:crange}, so that what this paper attributes there and what it
does not is settled by computation. The balance family and its pointwise
sharpening are checked on the compact models described at
Theorem~\ref{thm:balance}. For Theorem~\ref{thm:rigid} the mutual Majorana
bilinears of the pair $(\eta_-,\eta_+=\Gamma_+\Theta_-\eta_-)$ and the relative
$\mathrm{Spin}(9)$ orbit they define are computed, and the question of whether
$N\ge4$ sharpens either the warp identity or the rigidity theorem is answered
there rather than assumed.

\medskip\noindent\emph{The $\mathrm{Spin}(7)$ material of
Section~\ref{sec:geom}.} The negative statement of
Proposition~\ref{prop:noclosure} is established by rank computations rather
than by an absence of evidence. For that proposition more is recorded than the two ranks the proof
uses: over fourteen samples with different spinors and fluxes, the rank of the
matrix of $I_1,\dots,I_9$ and the constant function is ten and rises to eleven
when the column of $\Delta$ values is adjoined, so not even a universal linear
relation between $\Delta$ and the $I_a$ exists. Three controls accompany
the computation: a manufactured combination of the $I_a$, and $\Delta$
stripped of its three $\varphi$-terms, each raise neither rank; and the
combinatorial evaluation of the $\epsilon$-contraction $I_5$ is checked against
the sum over all $8!$ orderings of its free indices. The rest of the section is
checked in the same discipline: the projectors onto the summands of \eqref{eq:spin7split} are built
as eigenprojections of the Cayley form taken from a spinor bilinear, so that no
adapted frame is chosen anywhere, and the perfect square \eqref{eq:deltasquare}
is verified sample by sample against three negative controls and one
constructive zero-witness. That square is also established a second time with
abstract indices, from the two contraction identities of the Cayley form and
nothing else: the identities are checked against each other and against
$\varphi^2=336$, the antisymmetriser in $\sigma$ is shown to be free against a
totally antisymmetric $w$, and $\sigma^T\sigma=24\,\Pi_{\mathbf 7}$ is obtained
as an identity in $\mathcal{Y}$ rather than by Schur's lemma together with an
evaluation on one element, so that neither a representation nor an adapted
frame enters at any point.

\medskip\noindent\emph{The global statements of
Section~\ref{sec:consequences}.} The $\mathfrak{sl}(2,\Real)$ statements of
Section~\ref{sec:sl2r} are established by carrying out the derivation with (V),
(M), (S) and (L) as the only substitutions permitted, and the brackets, the
Jacobi identity and the signature $(2,1)$ of the Killing form are then checked on
the resulting vector fields, with negative controls in which (V) or (M) is broken
and the Killing property fails. The charges of Section~\ref{sec:charges} are
checked from the metric rather than from the
formulae quoted for them: the Killing property of $\mathcal{K}$, the components
of $\mathcal{K}^\flat$, the component $(\mathrm{d}\mathcal{K}^\flat)_{ru}$ and
the Hodge dual pulled back to the section are computed in an explicit Gaussian
null metric with generic $h_i$, $\Delta$ and $\gamma_{ij}$, with a control in
which the twist depends on the orbit coordinate and the Killing property fails;
the reduction $h(K)=2c-2\Delta f$ is verified with the two scalars $h^2$ and
$h^i\tn_if$ carried as independent symbols, with controls that drop the
gradient term of \eqref{eq:master} and that mis-normalise $K$. Finally,
\eqref{eq:rotationbound} is verified by periodic quadrature on the same compact
models that carry Theorem~\ref{thm:balance}. The global formulae are also
evaluated termwise on an explicit compact model of the static branch
\eqref{eq:staticbranch}, the branch on which the constancy hypothesis fails, so
that their degeneracy there --- $V\equiv0$, hence $J=0$ and saturation of
\eqref{eq:rotationbound} --- is checked rather than assumed, together with the
volume formula of Corollary~\ref{cor:entropy}, the lower end of the entropy
sandwich of Corollary~\ref{cor:smarr} and the balance family of
Theorem~\ref{thm:balance} at five values of $p$. The eliminations behind
Corollary~\ref{cor:smarr} are carried out on the four global relations held as
independent symbols, so that \eqref{eq:smarr} and \eqref{eq:komarsign} are
derived and not assumed, with controls that drop the $-c\,\mathrm{Vol}(\Sect)$
term and that mis-normalise the entropy, and the same relations are then
evaluated termwise on those compact models. For Proposition~\ref{prop:totalR}
the trace of \eqref{eq:Rij} is taken on explicit rational tensors in nine
dimensions, so that the contraction identities and the coefficients
$\tfrac14$ and $\tfrac1{48}$ are verified rather than quoted, the elimination
of $\tn^ih_i$ is performed with the residuals of \eqref{eq:Rij} and
\eqref{eq:divh} carried, and \eqref{eq:Rtotal} is verified on compact models
built to solve \eqref{eq:divh} exactly with $X^2,Y^2\ge0$ enforced pointwise.
The two negative results of the same section --- that the Bochner identity
\eqref{eq:bochner} does not close, and that the Morse and Poincar\'e--Hopf
routes are empty --- are likewise established by exhibiting the uncontrolled
term and the counterexamples rather than by an absence of evidence. The zero
set of the Bochner obstruction is then mapped, and shown to lie inside the
constant-norm branch the argument was meant to reach, so that it contains no
horizons the identity has not already caught; adding the invariance
$\mathcal{L}_Vf=0$ of \eqref{eq:pwbalance} removes one of the two witnesses of
indefiniteness and leaves the other, which is why the non-closure is
structural. That the static branch closes on the single scalar equation
$\tn^2f=Qf-2c$, the $++$ and $+i$ Einstein equations collapsing identically once
$\mathrm{d}h=0$, is checked in the same way, and is what leaves
Remark~\ref{rem:whynotconst} nothing to work with; the two pure branches are run
separately, $X=0$ forcing $\beta=0$ and killing every Chern--Simons term, while
$Y=0$ leaves $X\wedge X=0$ as the one genuine pointwise restriction. The intermediate range of \eqref{eq:crange} is probed in the same
style: on covariantly constant product data the field equations reduce to
algebra in the three squared amplitudes, the $\epsilon$ contractions of
\eqref{eq:divX} and \eqref{eq:divY} are summed over permutations on explicit
nine-dimensional tensors rather than inferred from index supports, the
flatness of every one-dimensional factor is imposed, and the endpoint verdict
is checked over the whole solution cone of the curvature constraints instead of
at a sampled point; the configuration excluded by the Bianchi identity alone
and the configurations excluded by the curvature constraints alone are both
exhibited. The integrability of the twist behind Corollary~\ref{cor:minuspointwise} is
checked in the same style: $\tn^ih_i$ is assembled from \eqref{eq:htoflux} and
\eqref{eq:Deltaexplicit} with $\tn Z$ and $\tn\varphi$ eliminated by the
spatial Killing spinor equation, and the resulting row is tested against the
row space of the $219$ conditions carried by $\mathrm{d}X=0$, \eqref{eq:divX}
and \eqref{eq:divY} on the $1458$ components of $(\tn X,\tn Y)$, with the
multiplier exhibited, the singular-value gap behind the rank recorded, and
controls in which the target is randomised and in which a functional already in
the list is adjoined. The sector sign of \eqref{eq:divhsector} is read off in
both sectors, the seven coefficients of \eqref{eq:Deltaexplicit} are recovered
from that scalar alone rather than assumed, and the general form
\eqref{eq:htofluxgen} is verified on an arbitrary twist, with the control that
\eqref{eq:htoflux} fails on the same data.
For Proposition~\ref{prop:twosector} the two sector vector identities are each
solved for $h$ and required to agree, which cuts the $153$-dimensional family
on which the Killing system for $a\,Z^{(+)}+b\,Z^{(-)}$ is assembled; the rank
of that system and the rank of the reduced obstruction $[T^{(+)},T^{(-)}]$ are
taken in exact rational arithmetic on the family, and again, from the
singular-value gap, on points of the pairing subvariety produced by a
Gauss--Newton solve, the two relaxations being the same computation with one
extra column. The positive control is a degenerate configuration --- $X=0$ with
$Y$ annihilating both sector vectors --- on which both sector vectors are
parallel and the same test returns the two-parameter family it fails to find
in general, so the rank-two verdict is a property of the data and not of the
test. For Theorem~\ref{thm:hidden} the two maps
$(X,Y)\mapsto\tn Z$ and $(X,Y)\mapsto\tn\varphi$ are built from the spatial
Killing spinor equation in the explicit representation and their ranks taken in
both sectors and at several spinors; the ranks are shown not to be tolerance
artefacts by exhibiting the singular-value gap at the reported rank, the
derivative code is cross-checked against two identities it was not built to
satisfy, and the Killing--Yano and closed conformal conditions, together with
the Hodge duality between them, are verified in dimensions where the tensors
can be written out in full.

\medskip\noindent\emph{The conventions themselves.} The exact solutions of
Remark~\ref{rem:calibration} are used as an external check: both are verified
against
\eqref{eq:einstein}--\eqref{eq:formeq} in an orthonormal frame before any
identity of this paper is applied to them, and the identity stack is then run
on them with the volume of the section and the value of $f$ carried as free
positive constants, so that they must cancel --- on the supersymmetric
\eqref{eq:calibE} as the spinorial identities themselves, and on the
non-supersymmetric control \eqref{eq:calibM} in the degenerate form they take
at constant $f$ with $c:=\Delta f$, where they test the coefficients only. The
Freund--Rubin throats are
checked in the same way, including the radius ratios, with controls in which
each ratio is perturbed. The ratio $\lVert\mathcal{Y}_{\mathbf{21}}\rVert^2
=\tfrac13\lVert\mathcal{Y}_{\mathbf 7}\rVert^2$ of \eqref{eq:calibratio} is
established by exhausting all $3360$ configurations of three orthogonal
coordinate planes with all relative signs against the Cayley form built from
the spinor, in both spinor sectors and in exact rational arithmetic, with
controls off the locus and at a perturbed coefficient.

\medskip\noindent What is quoted rather than checked is, principally, that the
three vector fields of Section~\ref{sec:sl2r} are spinor bilinears preserving
the four-form, and the analytic maximum-principle step of
Theorem~\ref{thm:saturation}.

\medskip\noindent\textbf{Data availability.}
The verification material consists of the $61$ computations described above,
written in Python/\textsc{SymPy} and in \textsc{Cadabra}2, the logs of the run
cited here, a \texttt{README} mapping every numbered statement to the
computation that checks it and to the log recording its output, the exact
interpreter and library versions (Python 3.13.12, \textsc{SymPy}
1.13.3, \textsc{NumPy} 2.2.1, \textsc{Cadabra}2 2.5.15), the commit hash of the
source tree, and a single command that reruns every check and rewrites the
logs. It is not deposited with this paper; it is available from the author on
request.

\end{document}